\documentclass[acmsmall]{acmart}

\AtBeginDocument{%
  \providecommand\BibTeX{{%
    \normalfont B\kern-0.5em{\scshape i\kern-0.25em b}\kern-0.8em\TeX}}}

\usepackage[utf8]{inputenc} %
\usepackage{hyperref}       %
\usepackage{url}            %
\usepackage{booktabs}       %
\usepackage{amsfonts}       %
\usepackage{nicefrac}       %
\usepackage{microtype}      %

\usepackage{amsmath,amsfonts}
\usepackage{algorithmic}
\usepackage{graphicx}
\usepackage{textcomp}
\usepackage{xcolor}
\usepackage{multirow}
\usepackage{algorithm}
\usepackage{bm}
\usepackage{subfigure}
\usepackage{caption}
\usepackage{amsthm}
\usepackage{afterpage}
\usepackage{longtable}

\newtheorem{definition}{Definition}[section]
\newtheorem{thm}[definition]{Theorem}

\newtheorem{cor}[definition]{Corollary}
\newtheorem{prop}{Property}[section]

\setcopyright{acmlicensed}
\acmJournal{TOIS}
\acmYear{2021} \acmVolume{1} \acmNumber{1} \acmArticle{1} \acmMonth{1} \acmPrice{15.00}\acmDOI{10.1145/3473339}

\begin{document}

\title{Exploiting Positional Information for Session-based Recommendation}

\author{Ruihong Qiu}
\affiliation{%
  \institution{The University of Queensland}
  \city{Brisbane}
  \country{Australia}
}
\email{r.qiu@uq.edu.au}

\author{Zi Huang}
\affiliation{%
  \institution{The University of Queensland}
  \city{Brisbane}
  \country{Australia}
}
\email{huang@itee.uq.edu.au}

\author{Tong Chen}
\affiliation{%
  \institution{The University of Queensland}
  \city{Brisbane}
  \country{Australia}
}
\email{tong.chen@uq.edu.au}

\author{Hongzhi Yin}
\authornote{Corresponding author.}
\affiliation{%
  \institution{The University of Queensland}
  \city{Brisbane}
  \country{Australia}
}
\email{h.yin1@uq.edu.au}

\thanks{The work was supported by Australian Research Council Discovery Project (ARC DP190102353, DP190101985, CE200100025)}

\renewcommand{\shortauthors}{Ruihong Qiu, et al.}

\begin{abstract}
For present e-commerce platforms, it is important to accurately predict users' preference for a timely next-item recommendation. To achieve this goal, session-based recommender systems are developed, which are based on a sequence of the most recent user-item interactions to avoid the influence raised from outdated historical records. Although a session can usually reflect a user's current preference, a local shift of the user's intention within the session may still exist. Specifically, the interactions that take place in the early positions within a session generally indicate the user's initial intention, while later interactions are more likely to represent the latest intention. Such positional information has been rarely considered in existing methods, which restricts their ability to capture the significance of interactions at different positions. To thoroughly exploit the positional information within a session, a theoretical framework is developed in this paper to provide an in-depth analysis of the positional information. We formally define the properties of \textit{forward-awareness} and \textit{backward-awareness} to evaluate the ability of positional encoding schemes in capturing the initial and the latest intention. According to our analysis, existing positional encoding schemes are generally \textit{forward-aware} only, which can hardly represent the dynamics of the intention in a session. To enhance the positional encoding scheme for the session-based recommendation, a dual positional encoding (DPE) is proposed to account for both \textit{forward-awareness} and \textit{backward-awareness}. Based on DPE, we propose a novel Positional Recommender (PosRec) model with a well-designed Position-aware Gated Graph Neural Network module to fully exploit the positional information for session-based recommendation tasks. Extensive experiments are conducted on two e-commerce benchmark datasets, \textit{Yoochoose} and \textit{Diginetica} and the experimental results show the superiority of the PosRec by comparing it with the state-of-the-art session-based recommender models.
\end{abstract}



\maketitle

\section{Introduction}
\label{intro}
Nowadays, recommender systems (RS) play an essential role in e-commerce platforms. Traditional RS~\cite{item-knn,bprmf,fpmc} predict a user's preference by equally taking the historical interactions into consideration, e.g., clicks of items, listening to songs or watching movies. Generally, a user's preference shifts as time goes on, where the traditional RS are less capable of predicting it. To enable a model to deal with this shift, session-based recommender systems (SBRS) have recently emerged, which predict the users' current preferences based on a session~\cite{sum,mdp,gru4rec,gru4rec+,narm,Liu18STAMP,srgnn,fgnn,gc-san}. A session is defined as a short sequence of user-item interactions within a certain period.

\begin{figure}
    \centering
    \includegraphics[width=.8\linewidth]{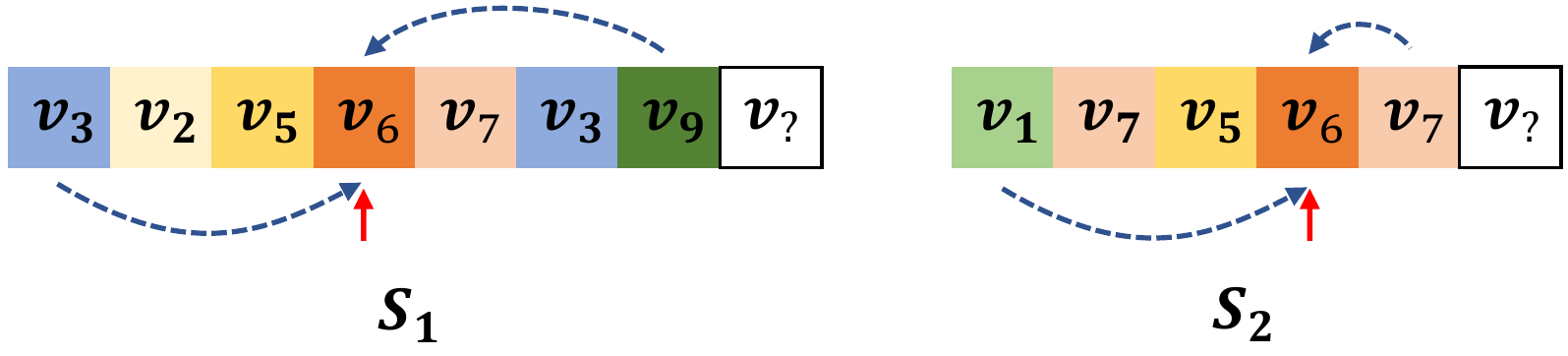}
    \caption{Illustration of the relationship between positions and intention dynamics in a session. Although the forward counting positions of $v_6$ in session $\mathcal{S}_1$ and $\mathcal{S}_2$ are the same, their backward counting positions are different, leading to different relative positions in a session.}
    \label{fig:pos}
\end{figure}

Although a session is assumed to imply the current preference, interactions happening at different stages in a session usually represent different intentions. On one hand, interactions in earlier positions may reflect the initial intention of a user. On the other hand, interactions closer to the end of a session usually demonstrate a better alignment with the latest intention. Such a difference is illustrated in Fig.~\ref{fig:pos} with two sample sessions. An item in a certain position in a session carries the positional information that reflects the initial and the latest intention. They are referred to as \textit{forward} and \textit{backward} positional information respectively in this work. A main purpose of our paper is to develop a positional encoding scheme to capture these two types of information.

How to effectively represent these two types of positional information remains a challenge in the session-based recommendation. For models using RNN as encoders~\cite{gru4rec,gru4rec+,latent-cross}, interactions are fed in the model according to their time order. These models implicitly make use of the positional information by considering the interactions sequentially. They suffer from easily forgetting the initial intention because the recurrent structure will potentially focus more on recent data. Attention-based approaches~\cite{narm,Liu18STAMP,pos-sbrs} apply the self-attention mechanism to compute the session representation. The attention mechanism utilizes the positional information in two ways: (1) including a positional encoding; and (2) using the last interaction in a session to attend to other interactions in the same session. When using the positional encoding~\cite{pos-sbrs}, it captures the \textit{forward} positional information because the positional encoding determines the position by counting from the beginning of a sequence. While for the latter case~\cite{narm,Liu18STAMP}, it neglects the positional information of all other interactions except for the last one, which merely represents the most recent intention. GNN-based methods~\cite{srgnn,gc-san,fgnn,bc,fgnnj,gag} firstly generate a session graph based on the relative position between interactions and further apply the self-attention to generate session representations. For example, for the session on the left in Fig.~\ref{fig:pos}, there will be a directed edge connecting $v_7$ to $v_3$. While for the session on the right in Fig.~\ref{fig:pos}, there will be edges connecting $v_7$ to both $v_5$ and $v_6$. In terms of the \textit{forward} and \textit{backward} positional information, the model cannot tell which item is on the first or the last position. Therefore, the positional information leveraged in the GNN model is rather limited as the constructed session graphs tend to neglect both the \textit{forward} and the \textit{backward} position information.

Specifically for the attention mechanism in sequence modeling, positional encoding is the most widely-used to capture the positional information, which is introduced to represent the absolute position of words in a sentence for natural language processing~\cite{attention}. It is expected to extract the positional information for words appearing in a specific position counted from the beginning of a sentence. However, in language modeling, the relative positions between words are more important than the absolute positions in a sentence, which makes the original positional encoding deprecated in recent language models~\cite{rpe,xlnet,trans-xl}. Recent recommender systems that use the attention mechanism usually involve a learnable version of the absolute positional encoding~\cite{pos-sbrs,sasrec,bert4rec,ali-seq}. Similar to the fixed positional encoding, a learnable one can only capture the \textit{forward} positional information as well.

In this paper, the positional information in SBRS is firstly formally defined in terms of the \textit{forward} and the \textit{backward} positional information. Besides, the abilities of models in capturing this information are further analyzed. \textit{Forward-awareness} and \textit{backward-awareness} are mainly investigated as the properties of existing position encoding schemes to represent the positional information. More importantly, based on the theoretical analysis, we propose a novel dual positional encoding scheme, which can capture the positional information with \textit{forward-awareness} and \textit{backward-awareness} in the session-based recommendation. In attention to the dual positional encoding scheme, a well-designed Position-aware Gated Graph Neural Network module is proposed to further incorporate the positional information in the session representation learning.

In summary, the main contributions of this paper are as follows:
\begin{itemize}
    \item A \textbf{theoretical framework} is developed to analyze the ability of different positional encoding schemes in representing the positional information for SBRS.
    \item A \textbf{Dual Positional Encoding (DPE)} scheme is proposed to represent the positional information for SBRS, which can be extended to a learnable version, denoted as LDPE.
    \item A \textbf{Positional Recommender model (PosRec)} is proposed based on (L)DPE, in which a Position-aware Gated Graph Neural Network (PGGNN) module is designed to further exploit the positional information in SBRS.
    \item \textbf{Extensive experiments} are conducted on two real-world benchmark SBRS datasets, \textit{Yoochoose}\footnote{\href{https://2015.recsyschallenge.com/challenge.html}{https://2015.recsyschallenge.com/challenge.html}} and \textit{Diginetica}\footnote{\href{http://cikm2016.cs.iupui.edu/cikm-cup/}{http://cikm2016.cs.iupui.edu/cikm-cup/}}. The empirical results demonstrate the superiority of the proposed PosRec andd (L)DPE compared with baselines.
\end{itemize}

This paper is structured as follows: in Section~\ref{sec:re-wo}, the related work about SBRS and the positional encoding is briefly reviewed. In Section~\ref{sec:theorems}, the theoretical framework is elaborated, followed by the explanation of the proposed (L)DPE and PosRec model in Section~\ref{sec:build-pos-rec}. In Section~\ref{sec:exp}, experiments are conducted to evaluate the effectiveness of our method.

\section{Related Work}
\label{sec:re-wo}
In this section, we review three main topics of previous research: the session-based recommendation, the positional encoding, and the graph neural networks.
\subsection{Session-based Recommendation}
{\bf Markov chain} is applied by many models~\cite{mdp,zimdars2001using} to learn the dependency of items in sequential data. Using probabilistic decision-tree models, Zimdars et al.~\cite{zimdars2001using} proposed to encode the state of the transition pattern of items. Shani et al.~\cite{mdp} made use of a Markov Decision Process (MDP) to compute item transition probabilities.

{\bf Deep learning models} become popular since the widely use of recurrent neural networks~\cite{gru4rec,narm,Liu18STAMP,gru4rec+,tois-poi,tois-poi1,LiLZSC19,LiDCPOI21}. There are three main branches of methods to perform the representation learning of the session, i.e., recurrent models~\cite{gru4rec,gru4rec+,narm}, attention models~\cite{Liu18STAMP,ssrm,history} and graph models~\cite{srgnn,fgnn,gc-san,fgnnj,gag,bc,hypersbrs}. (1) For recurrent models, e.g., GRU4REC~\cite{gru4rec,gru4rec+} and NARM~\cite{narm}, Gated Recurrent Unit~\cite{gru} and Long Short-Term Memory~\cite{lstm} are applied respectively and the positional information is implicitly modeled by the recurrent computing procedure. The recurrent structure includes a strong inductive bias that the relationship between items is linear along with the position. (2) For attention models, e.g., NARM (a self-attention layer is applied after the recurrent layer) and STAMP~\cite{Liu18STAMP} utilizes self-attention~\cite{attention} over the last item to capture the relationship between the last item and the rest in the session. These attention-based methods only consider the importance of the last position while neglecting other positions. (3) In graph modeling, e.g., SR-GNN~\cite{srgnn}, GC-SAN~\cite{gc-san}, FGNN~\cite{fgnn} and MGNN-SPred~\cite{bc}, a session is converted into a graph and Graph Neural Networks (GNN)~\cite{gcn,ggnn,gat} captures the connectivity of items. Afterward, a readout function is applied to compute a session representation with the processed item representations. For SR-GNN and GC-SAN, the readout function is similar to attention-based models by performing a self-attention over the last item. While FGNN uses a Set2Set~\cite{set2set} module and computes a descriptive vector, which is considered as a latent description of items. MGNN-SPred makes use of the mean feature of the whole sequence to represent the user modeling. Consequently, these GNN-based methods only capture the relative position for the connected items, which does not satisfy the \textit{forward-awareness} and \textit{backward-awareness}. The proposed PosRec falls into the category of graph-based model. To enhance the exploitation of the positional information of the graph representation learning, the (L)DPE is included in the embedding of the items and the graph neural network is redesigned to have a position-aware module.

{\bf Sequential recommendation} is a close research field to SBRS. In recent years, deep learning models are very popular~\cite{caser,sasrec,bert4rec,tois-seq,tois-seq1,safm,ddl}. Caser~\cite{caser} applies convolutional layers to process the embeddings of items in a sequence. SASRec~\cite{sasrec} and BERT4Rec~\cite{bert4rec} use the Transformer~\cite{attention} in a single direction style and a bidirection style respectively to model the sequential pattern in the interaction sequence.

\subsection{Positional Encoding}
{\bf Absolute positional encoding} is firstly introduced with the attention structure to provide the access of sequential information for the permutation invariant computation~\cite{attention}. It assigns a fixed vector to each position in a sequence. The vector is computed either in a sinusoidal way or a learned style. For example, the language model BERT~\cite{bert} and the recommendation model BERT4Rec~\cite{bert4rec}, they both use the learned positional encoding. {\bf Relative positional encoding} is later proposed to encode the relative position of two words, which is more meaningful for the natural language~\cite{rpe,xlnet,trans-xl,pos-iclr}. For example, the language model XLNet~\cite{xlnet} and Transformer-XL~\cite{trans-xl} propose different types of relative position encodings to represent the relative positional information between words in a sentence. {\bf Other positional encodings} include different positional encoding schemes that are suitable for data structures other than one-dimensional sequence. For example, to apply the attention to images, there are 2D positional encoding schemes~\cite{2dlatex,aacn,sasa,2dsa,e2eatt} that provide either the absolute or the relative encoding. For example, the attention augmented network~\cite{aacn} designs a 2D relative positional encoding to encode the positional information in the activation map. For tree structures, Shiv and Quirk~\cite{tree} proposed a specific scheme to encode the relationship between the root node and children nodes.

\subsection{Graph Neural Networks}
Recently, to enable neural networks to work on structured data (e.g., graph, point cloud, etc.), Graph Neural Networks (GNN) are widely investigated~\cite{gcn,gat,ggnn,pgnn}. Generally, the computation flow of GNN is called message passing, which is based on neighborhood aggregation. For example, GCN~\cite{gcn}, GAT~\cite{gat} and GGNN~\cite{ggnn} are majorly different in the aggregation method. However, these GNN models could easily fall into a lack of representative ability since the message passing is performed on a narrow scope of nodes. Thus, PGNN~\cite{pgnn} is proposed to include the information from randomly chosen anchor nodes to utilize extra structural information.

\section{Theoretical Framework for Positional Encoding}
\label{sec:theorems}
In this section, we build up the theoretical framework to analyze the property of different positional encoding schemes and what is needed to represent the positional information for SBRS.

\subsection{Positional Encoding}
The positional encoding (PE) is introduced by~\cite{attention} to enable the self-attention module to utilize the positional information of languages. Here, the sinusoidal positional encoding (SDE) $\boldsymbol{P}\in\mathbb{R}^{d\times 1}$ of a token at position $pos$ in the session of length $l$ is defined as:
\begin{equation}
\label{eq:spe}
\begin{aligned}
    P_{pos,2i}^l&=\sin(pos/f(i)), \\ P_{pos,2i+1}^l &=\cos(pos/f(i)),
\end{aligned}
\end{equation}
where $i \in \{0,1,\dots,d / 2-1\}$, $d$ is the dimension of the feature vector and $f(i)=10000^{2 i / d}$. In the following, all $pos\in \{0,1,\dots,l-1\}$ if not specified.

\subsection{Property of Positional Encoding}
\begin{definition}[Forward-awareness]
\label{def:for}
A positional encoding $\boldsymbol{P}$ is \textit{forward-aware} in positional information if $\forall p,q \in \mathbb{Z}^+$, $\exists A\subseteq \{0,1,\dots,d-1\}, A\neq \varnothing$, for two positions $pos_a$ and $pos_b$, if $pos_a = pos_b$, then $P_{pos_a,A}^p=P_{pos_b,A}^q$ and if $pos_a \neq pos_b$, then $P_{pos_a,A}^p \neq P_{pos_b,A}^q$.
\end{definition}

\begin{definition}[Backward-awareness]
\label{def:back}
A positional encoding $\boldsymbol{P}$ is \textit{backward-aware} in positional information if $\forall p,q \in \mathbb{Z}^+$, $\exists B\subseteq \{0,1,\dots,d-1\}, B\neq \varnothing$, for two positions $pos_a$ and $pos_b$, if $p-pos_a = q-pos_b$, then $P_{pos_a,B}^p=P_{pos_b,B}^q$ and if $p-pos_a \neq q-pos_b$, then $P_{pos_a,B}^p \neq P_{pos_b,B}^q$.
\end{definition}

To investigate the representation ability of a PE in a session, we define two features: \textit{forward-awareness} and \textit{backward-awareness}. If a PE is \textit{forward-aware}, the PE of the first token is the same for all sequences. Furthermore, if a position $pos$ exists in any sequence, the PE for $pos$ is the same across these sequences. For example, if we assign the position itself as the PE, i.e., $P_0^l=0,P_1^l=1\dots$, then it is \textit{forward-aware}. In contrast, if a PE is \textit{backward-aware}, the PE of the last token is the same for all sequences. Furthermore, if an $h$-th last position $pos$ exists in any sequence, the PE for $pos$ is the same across these sequences. For example, if we assign the reverse position as the PE, i.e., $P_{l-1}^l=0,P_{l-2}^l=1\dots$, then it is \textit{backward-aware}. A demonstration of \textit{forward-awareness} and \textit{backward-awareness} can be found in Fig.~\ref{fig:pos}.

\begin{prop}
\label{prop:for}
If a positional encoding $\boldsymbol{P}$ is \textit{forward-aware}, $\forall 0\leq pos_a\leq pos_b<\text{min}(p,q)$, $\exists f:\mathbb{R}^{d\times 1}\times\mathbb{R}^1\times\mathbb{R}^1\mapsto\mathbb{R}^{d\times 1}, \exists A$, s.t. $P_{pos_b,A}^p=f(P_{pos_a,A}^p, pos_a, pos_b)$, then $P_{pos_b,A}^q=f(P_{pos_a,A}^q, pos_a, pos_b)$.
\end{prop}

\begin{proof}
Following Definition~\ref{def:for}, because items are at the same position $pos_a$, $P_{pos_a,A}^p=P_{pos_a,A}^q$. Similarly for position $pos_b$, $P_{pos_b,A}^p=P_{pos_b,A}^q$. Then $P_{pos_b,A}^q=f(P_{pos_a,A}^q, pos_a, pos_b)$ holds for the function $f(\cdot,\cdot,\cdot)$.
\end{proof}

If there is a mapping between two PE in a session, the mapping also applies to other sessions that contain same positions. Similarly, the property of \textit{backward-aware} PE is as the following:

\begin{prop}
\label{prop:back}
If a positional encoding $\boldsymbol{P}$ is \textit{backward-aware}, $\forall 0\leq pos_a\leq pos_b<p,0\leq pos_c\leq pos_d<q$, $\exists f:\mathbb{R}^{d\times 1}\times\mathbb{R}^1\times\mathbb{R}^1\mapsto\mathbb{R}^{d\times 1}, \exists B$, s.t. $P_{p-pos_a,B}^p=f(P_{p-pos_b,B}^p, pos_a, pos_b)$, if $p-pos_a = q-pos_c$ and $p-pos_b = q-pos_d$, then $P_{q-pos_c,B}^q=f(P_{q-pos_d,B}^q, pos_a, pos_b)$.
\end{prop}

\begin{proof}
Following Definition~\ref{def:back}, because item at $pos_b$ for length $p$ and item at $pos_d$ for length $q$ are at the reverse position $p-pos_b$, $P_{pos_b,B}^p=P_{pos_d,B}^q$. Similarly for position $pos_a$ and $pos_c$, $P_{pos_a,B}^p=P_{pos_c,B}^q$. Then $P_{pos_c,B}^q=f(P_{pos_d,B}^q, pos_a, pos_b)$ holds for the function $f(\cdot,\cdot,\cdot)$.
\end{proof}

These Definitions and Properties together give another important Properties of an absolute PE.
\begin{prop}
\label{prop:uniq}
An absolute positional encoding is unique for each position.
\end{prop}
\begin{proof}
\label{prf:prop-uniq}
If there are duplicate PE for different positions, Definition~\ref{def:for} and~\ref{def:back} are violated.
\end{proof}

\subsection{Positional Information for Session-based Recommendation}
\label{sec:pos-pos}
\begin{definition}
\label{def:pos-sbrs}
A positional encoding that can represent the positional information in SBRS is both \textit{forward-aware} and \textit{backward-aware}.
\end{definition}

As discussed in the Introduction, the position in a session carries specific positional information in SBRS. The first item reflects the initial intention of the user while the last item is always considered more relevant to the latest preference of the user. And the items in-between usually represent the preference shift inside the session. For the \textit{forward-aware} requirement, following Definition~\ref{def:for}, two items at the same position of two different sessions always have the same slice of their PE. Following Property~\ref{prop:for}, the relationship between any position and the first position is the same across different sessions. As for the \textit{backward-aware} requirement, the position in \textit{forward-aware} requirement is changed into the reverse position following Definition~\ref{def:back} and Property~\ref{prop:back}.

\begin{thm}
\label{thm:spe}
The sinusoidal positional encoding cannot represent the positional information in SBRS because it is \textit{forward-aware} but not \textit{backward-aware}.
\end{thm}

\begin{proof}
We first prove that SPE is \textit{forward-aware} and then SPE is not \textit{backward-aware}. (1) According to Eq.\ (\ref{eq:spe}), SPE directly follows Definition~\ref{def:for} for \textit{forward-aware}. (2) Take the last item of two sessions w.r.t.\ length $1$ and $2$ as example. For length $1$ session, $P_{0,2i}^1=0$ and $P_{0,2i+1}^1=1$. If SPE is \textit{backward-aware}, for length $2$ session, there should be a slice of $P_{1,B}^2$ is the same as $P_{0,2i}^1$ and $P_{0,2i+1}^1$. For $2i$ dimension of SPE,  $P_{0,2i}^1=P_{1,2i}^2=\sin(1/10000^{2 i / d})$. It is clear that $1/10000^{2 i / d}\in\left[0.00001,1\right]$. Then $P_{0,2i}^1\neq P_{1,2i}^2$. Similarly, $P_{0,2i+1}^1\neq P_{1,2i+1}^2$. Therefore, SPE is not \textit{backward-aware}.
\end{proof}

This Theorem states that the sinusoidal positional encoding is not informative for positional information in session-based recommendation. As proved above, SPE is \textit{forward-aware} because it is exactly calculated based on the position. Using SPE in any SBRS can only indicate how far an item is from the user's initial intention. However, SPE is not \textit{backward-aware} as it simply cannot tell if an item is at the last position of a session. In SBRS, it is crucial to know the preference shift within the session~\cite{gru4rec,mdp}. Because SPE is not \textit{backward-aware}, if a model uses SPE, there is no information about the closeness between an item and the user's latest preference (i.e., the item at the last position).

\begin{cor}
\label{cor:rpe}
The relative positional encoding cannot represent the positional information in SBRS because it is neither \textit{forward-aware} nor \textit{backward-aware}.
\end{cor}

\begin{proof}
The RPE is explored by recent language models~\cite{rpe,xlnet,trans-xl}. During the attention score calculation, the absolute positional encoding $\boldsymbol{P}$, e.g., SPE, is included as:
\begin{equation}
    \boldsymbol{A}=(\boldsymbol{X}_i+\boldsymbol{P}_i) \boldsymbol{W}_{q r y} \boldsymbol{W}_{k e y}^{\top}(\boldsymbol{X}_j+\boldsymbol{P}_j)^{\top},
\end{equation}
where $\boldsymbol{X}$ is the input feature and $\boldsymbol{W}$ is trainable weights.

For RPE in different work, they basically follow a format:
\begin{equation}
    \boldsymbol{A}=\boldsymbol{X}_i\boldsymbol{W}_{q r y} \boldsymbol{W}_{k e y}^{\top}\boldsymbol{X}_j^{\top}+g(\boldsymbol{P}_{ij}),
\end{equation}
where $\boldsymbol{P}$ only represents the relative position between $i$ and $j$. Because $\boldsymbol{P}_{ij}$ does not provide any information about the absolute position of a token, for different center tokens $i_1$ and $i_2$, $\boldsymbol{P}_{i_1j}\neq\boldsymbol{P}_{i_2j}$. Because $j$ can be before or after $i$ in position, RPE simultaneously is not \textit{forward-aware} and \textit{backward-aware}.
\end{proof}

Relative positional encoding (RPE) is designed to relax the assumption in language models that a word in an absolute position has the same meaning. RPE focuses more on the meaning of the relative position between two words. As proved above, RPE is neither \textit{forward-aware} nor \textit{backward-aware}, thus failing to meet both requirements of SBRS.

Empirically, the closer an item is to the last item, the more accurate it can reflect the user's latest preference. As discussed in the Introduction, many methods consider the last item as the representation of the latest preference (usually referred to as short-term or local preference). Meanwhile, other items are treated with less importance (usually referred to as long-term or global preference). Following Theorem~\ref{thm:spe}, SPE only contains the \textit{forward-awareness}. Intuitively, we can modify the SPE to a reverse sinusoidal positional encoding (RSPE):
\begin{equation}
\label{eq:rspe}
\begin{aligned}
    P_{l-pos-1,2i}^l&=\sin(\left(l-pos-1\right)/f(i)), \\ P_{l-pos-1,2i+1}^l&=\cos(\left(l-pos-1\right)/f(i)).
\end{aligned}
\end{equation}

\begin{cor}
\label{cor:rspe}
The reverse sinusoidal positional encoding cannot represent the positional information in SBRS because it is \textit{backward-aware} but not \textit{forward-aware}.
\end{cor}

\begin{proof}
Similar to the proof of Theorem~\ref{thm:spe}, we firstly prove RSPE is \textit{backward-aware} and then is not \textit{forward-aware}. (1) According to Eq.\ (\ref{eq:rspe}), RSPE directly follows Definition~\ref{def:back} for \textit{backward-aware}. (2) Take the first item of two sessions w.r.t.\ length $1$ and $2$ as example. For length $1$ session, $P_{0,2i}^1=0$ and $P_{0,2i+1}^1=1$. If RSPE is \textit{forward-aware}, for length $2$ session, there should be a slice of $P_{0,A}^2$ is the same as $P_{0,2i}^1$ and $P_{0,2i+1}^1$. For $2i$ dimension of RSPE,  $P_{0,2i}^1=P_{0,2i}^2=\sin(1/10000^{2 i / d})$. It is clear that $1/10000^{2 i / d}\in\left[0.00001,1\right]$. Then $P_{0,2i}^1\neq P_{0,2i}^2$. Similarly, $P_{0,2i+1}^1\neq P_{0,2i+1}^2$. Therefore, RSPE is not \textit{forward-aware}.
\end{proof}

With RSPE rather than SPE, a model is theoretically able to utilize the positional information that can reflect how an item is different from the latest preference in the session. But obviously, RSPE neglects the positional information representing the relationship between an item and the initial intention.

\subsection{Beyond Single Directional Positional Encoding}
\label{sec:beyond}
In the content above, we focus on the positional encoding that only rolls out in a single direction. In the following, we will discuss the additional positional encoding and 2D positional encoding.

Direct addition of SPE and RSPE fails in this situation. Such an addition will create a symmetric positional encoding that for $pos_a$ and $pos_b$ in a length $l$ session, if $pos_a-0=l-pos_b-1$, $P_{pos_a}^l=P_{p-pos_b-1}^l$. This will break the Property~\ref{prop:uniq}.

If we exchange the $2i$ and $2i+1$ dimensions of RSPE in Eq.\ (\ref{eq:rspe}), and do the addition of SPE and RSPE, the resulted additional sinusoidal positional encoding (ASPE) follows Property~\ref{prop:uniq} for uniqueness, but it is inconsistent with Definition~\ref{def:for} and Definition~\ref{def:back}. Therefore, there is no guarantee on the positional information in SBRS according to Definition~\ref{def:pos-sbrs}.

An interesting property about this ASPE is that the pattern of uniqueness is insufficient so that the attention model cannot easily infer the positions but only to memorize. In SPE and RSPE, there is a linear combination property between two positions~\cite{attention,tree}. But the ASPE breaks this property, which leads to the attention model cannot learn the relationship between different positions. We prove this difference in {Appendix}~\ref{prf:prop-ape}.

In the literature of computer vision that utilizes attention mechanism, there is a type of encoding for images called 2D sinusoidal positional encoding (2DSPE)~\cite{2dlatex,aacn,sasa,2dsa,e2eatt}. If $pos$ and $l$ are considered as the height and width, then 2DSPE is similar to the ASPE that the encoding of each pair $(pos,l)$ is totally unique and thus, it does not follow the \textit{forward-aware} and \textit{backward-aware} requirements. Therefore, they are not eligible for SBRS. Detail of an example of 2DSPE is presented in {Appendix}~\ref{exp:2dspe}.

\section{Building Positional Recommender Model}
\label{sec:build-pos-rec}
In this section, we will derive a (learned) dual positional encoding ((L)DPE) to improve the representation ability of positional information and utilize (L)DPE to develop our Positional Recommender model for session-based recommendation.

\subsection{Problem Definition}
\label{sec:def}
In SBRS, an item is denoted as $v$ and there is a unique item set $\mathcal{V}=\{v_1,v_2,v_3,\ldots,v_m\}$, with $m$ being the number of items. A session sequence from an anonymous user is defined as an order list $\mathcal{S}=[v_{s,1},v_{s,2},v_{s,3},\ldots,v_{s,l}]$, $v_{s,*} \in \mathcal{V}$. $l$ is the length of the session $\mathcal{S}$. In this paper, a sequence has at least one item and $l\in \mathbb{Z}^+$. The goal of our model is to take an anonymous session $\mathcal{S}$ as input, and predict the next item $v_{s,l+1}$ that matches the current preference.

\subsection{Dual Positional Encoding}
\label{sec:dpe}
We propose a dual positional encoding (DPE) by concatenating of half of the SPE and half of the RSPE positional encoding. The DPE is defined as:
\begin{equation}
\label{eq:dpe}
\begin{aligned}
    P_{pos,2i}^l&=\sin(pos/f(i)),\\
    P_{pos,2i+1}^l&=\cos(pos/f(i)),\\
    P_{pos,2i+d/2}^l&=\sin((l-pos-1)/f(i)),\\
    P_{pos,2i+1+d/2}^l&=\cos((l-pos-1)/f(i)),
\end{aligned}
\end{equation}
where $i\in \{0,1,\dots,d/4\}$ and for clarity, we assume $d/4\in\mathbb{Z}$ and all our results can be easily generalized to other cases.
\begin{thm}
\label{thm:dpe}
Dual positional encoding can represent the positional information of SBRS because it is both forward-aware and backward-aware.
\end{thm}

\begin{proof}
(1) $\forall p,q\in\mathbb{Z}^+$, $\forall pos\leq \text{min}(p,q)$, $\exists A=\{0,1,\dots,d/2-1\}$, s.t. $P_{pos,A}^p=P_{pos,A}^q$. This follows Definition~\ref{def:for} and DPE is \textit{forward-aware}. (2) $\forall p,q\in\mathbb{Z}^+$, $\forall pos_a,pos_b=\{a,b|p-a=q-b,0\leq a<p,0\leq b<q\}$, $\exists B=\{d/2,d/2+1,\dots,d-1\}$, s.t. $P_{pos_a,B}^p=P_{pos_b,B}^q$. This follows Definition~\ref{def:back} and DPE is \textit{backward-aware}. Therefore, DPE can represent the positional information of SBRS.
\end{proof}

This theorem states that the proposed DPE can represent the positional information of SBRS. For the $A=\{0,1,\dots,d/2-1\}$ dimensions of DPE, it is \textit{forward-aware}. While for the rest $B=\{d/2,d/2+1,\dots d-1\}$ dimensions, it is \textit{backward-aware}. For example, for the first item of any session with length $p$ and $q$, $P_{0,A}^p=P_{0,A}^q$ is always true while no guarantee that $P_{0,B}^p=P_{0,B}^q$. But $P_{p-1,B}^p=P_{q-1,B}^q$ is always true for DPE.

The proposed DPE has met the requirements of SBRS. This encoding scheme is parameter-free. In this situation, positions in a session can be considered linear because they are ordered with a consistent interval. However, in the real world, the position of each item actually comes from the timestamp of the interaction. The time intervals between interactions are neither consistent nor linear. Therefore, we propose the following learned dual positional encoding (LDPE) to improve the inductive bias injected into a session-based recommendation model:
\begin{equation}
\label{eq:ldpe}
\begin{aligned}
    P_{pos,0:d/2-1}^l&=\text{Embed}[pos],\\
    P_{pos,d/2:d-1}^l&=\text{Embed}[l-pos-1],
\end{aligned}
\end{equation}
where $\text{Embed}\in \mathbb{R}^{max(l)\times d}$ stands for a learned embedding matrix and $[\cdot]$ is the same as the slice operation for a list in Python. Similar to DPE, LDPE follows the Definition~\ref{def:for} and~\ref{def:back}, and thus~\ref{def:pos-sbrs}.

\subsection{Positional Recommender Model}
\label{sec:pos-rec}
With DPE and LDPE ((L)DPE), we now build our Positional Recommender model (PosRec) based on the Position-aware Gated Graph Neural Network (PGGNN) and the bidirectional Transformer. Similar to GNN-based models~\cite{srgnn,fgnn,gc-san}, a session is firstly converted into a weighted and directed session graph. Then PGGNN is applied to calculate the position-aware item embedding as the input of the bidirectional Transformer layer. (L)DPE is incorporated into the bidirectional Transformer layer to enhance the positional information. In the end, a single vector is computed as the representation of the session and used to predict the user's next click.

\subsubsection{Session Graph}
\label{sec:sess-graph}
To utilize the neighboring information, a session is converted into a weighted directed session graph. Similar to~\cite{srgnn,fgnn,gc-san}, the conversion procedure basically abides by the following process. If an item $v_{s,t}$ is immediately followed by the next item $v_{s,t+1}$ in the session $S$, then a directed edge $(w_{s,t,t+1},v_{s,t},v_{s,t+1})$ from this item to the next item with the edge weight $w_{s,t,t+1}$ indicating the frequency of occurrence of such an edge in $S$, is added to the session graph $G_s(V_s,E_s)$. $V_s$ includes all items in the session $S$, and we refer to an item as a node in the following without specific indication. Each node feature $\mathbf{x}_{s,t}\in \mathbb{R}^{1\times d}$ is initialized by the corresponding ID of the item $v_{s,t}$ and a lookup embedding matrix. $E_s$ stands for all generated edges.

\subsubsection{Position-aware Gated Graph Neural Network}
\label{sec:pgnn}
The Position-aware Gated Graph Neural Network (PGGNN) is designed to process the session graph to obtain the updated item embedding. PGGNN consists of two node aggregation steps, a neighboring node aggregation based on Gated Graph Neural Network (GGNN)~\cite{ggnn} and an anchor node aggregation based on position-aware Graph Neural Network (PGNN)~\cite{pgnn}.

GGNN for the weighted and directed session graph is defined as:
\begin{equation}
\label{eq:ggnn}
    \hat{\mathbf{x}_{t'}}=\sum_{v_{t'}\in N_{\text{in}}(v_t)}w_{t',t}\mathbf{x}_{t'}\mathbf{W}_{\text{in}}||\sum_{v_{t'}\in N_{\text{out}}(v_t)}w_{t,t'}\mathbf{x}_{t'}\mathbf{W}_{\text{out}},
\end{equation}
\begin{equation}
    \mathbf{x}'_t=\text{GRU}(\mathbf{x}_t,\hat{\mathbf{x}_{t'}}),
\end{equation}
where $\hat{\mathbf{x}_{t'}}\in \mathbb{R}^{1\times 2d}$ is the message from all neighbors of $v_t$, $N_{\text{in}}(v_t)$ is the set of nodes targeting at $v_t$, $N_{\text{out}}(v_t)$ is the set of nodes targeted by $v_t$, $\mathbf{W}_{\text{in}},\mathbf{W}_{\text{out}}\in\mathbb{R}^{d\times d}$ are trainable weights and $||$ stands for the concatenation along the feature dimension. $\mathbf{x}'_t$ is the updated node feature of $v_t$.

\begin{figure}[t]
    \centering
    \includegraphics[width=1.\linewidth]{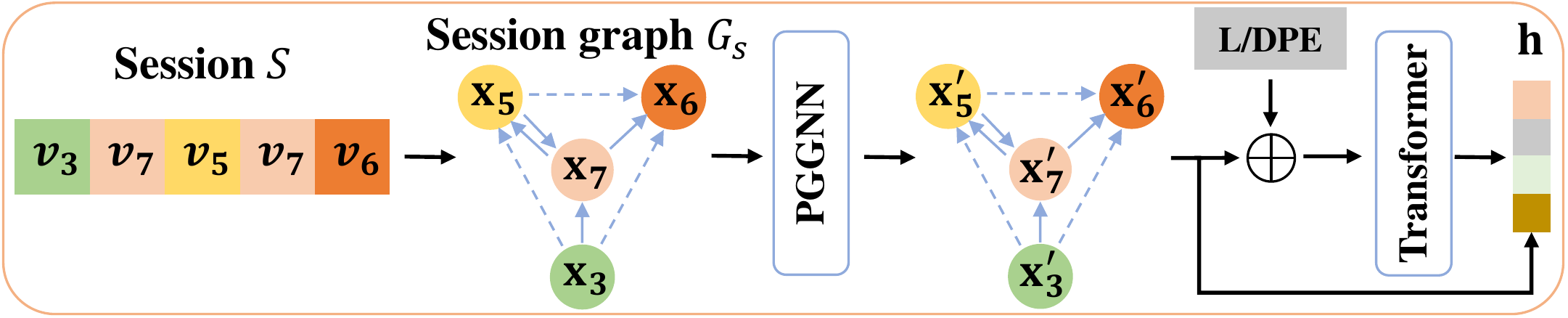}
    \caption{The pipeline of PosRec. A session $S$ is firstly converted into a session graph $G_s$. PGGNN aggregates neighboring (solid edges) and anchor (dashed edges) nodes to update node features in $G_s$. Bidirectional Transformer uses updated node features and (L)DPE to compute a session representation $\mathbf{h}$. $\oplus$ stands for element-wise addition.}
    \label{fig:whole}
\end{figure}

The anchor node is first introduced in PGNN by~\cite{pgnn} with random sampling on all nodes in a graph. In the context of the session graph, nodes with notable importance, e.g., the first item, the last item and re-appearing items, can be chosen as anchor nodes. Therefore, to improve the inductive bias in the GGNN, for each node $v_t$ in a session, we add the first item $v_0$ to $N_{\text{in}}(v_t)$ and the last item $v_{l-1}$ to $N_{\text{out}}(v_t)$. In addition, for all items $v_*$ appearing more than once in a session, $v_*$ are added to both $N_{\text{in}}(v_t)$ and $N_{\text{out}}(v_t)$. Therefore, we substitute $N_{\text{in}}(v_t)$ and $N_{\text{out}}(v_t)$ in Eq.\ (\ref{eq:ggnn}) with $N_{\text{in}}'(v_t)=\{N_{\text{in}}(v_t)+v_0+v_*\}$ and $N_{\text{out}}'(v_t)=\{N_{\text{out}}(v_t)+v_{l-1}+v_*\}$. The corresponding weights between any anchor node and other nodes are defined as the distance of these two nodes on an unweighted and undirected graph converted from $G_s$ by omitting the weight and direction associated with every edge. Examples of these edges are shown as dashed edges in Fig.~\ref{fig:whole}.

\subsubsection{Bidirectional Transformer Readout Function with (L)DPE}
\label{sec:trans}
With the updated item features and (L)DPE, the bidirectional Transformer layer serves as the readout function to generate a feature vector for the session. The bidirectional Transformer layer operates at the graph level rather than the sequence level of the session, which will lower the noisy signal of repetitive items. Consider $\mathbf{X}'\in\mathbb{R}^{n\times d}$, which includes all updated node features, where $n$ is the number of unique items in the session. Let $\mathbf{P}^{n\times d}$ represent (L)DPE of corresponding items of nodes in the session. The feature vector $\mathbf{h}\in\mathbb{R}^{1\times d}$ representing the session can be defined as:
\begin{equation}
    \mathbf{H}=\text{Transformer}(\mathbf{X}'+\mathbf{P}),
\end{equation}
\begin{equation}
\label{eq:h-session}
    \mathbf{h}=\lambda_0\mathbf{X}_{l-1}'+\lambda_1\mathbf{H}_{l-1}+\lambda_2\mathbf{H}_0,
\end{equation}
where $\mathbf{H}\in\mathbb{R}^{n\times d}$ carries all output states of the bidirectional Transformer and subscripts $l-1$ and $0$ represent the corresponding entries in $\mathbf{X}'$ and $\mathbf{H}$ of items $v_0$ and $v_{l-1}$. $\lambda_*$ are pre-defined weights.

Note that, if a session does not contain duplicated items, every interaction will have a unique (L)DPE. When there are repeated items, the node for such an item will adopt the \textit{forward-aware} part of (L)DPE of the earliest appearance and the \textit{backward-aware} part of (L)DPE of the latest appearance.

\subsubsection{Objective Function}
\label{sec:obj}
After having a representation $\mathbf{h}$ of a session, we can compare $\mathbf{h}$ with the whole item set to decide what to recommend to the user. Let $\mathbf{X}\in\mathbb{R}^{m\times d}$ be the initial embedding of the whole item set. The score of recommendation $\hat{\mathbf{y}}\in\mathbb{R}^{m\times 1}$ and the Cross-Entropy loss function are defined as:
\begin{equation}
\label{eq:score}
\hat{\mathbf{y}}=\text{Softmax}(\mathbf{X}\mathbf{h}^\top),
\end{equation}
\begin{equation}
L=-\sum_{i=1}^{M}\mathbf{y}_i^\top\log \left(\hat{\mathbf{y}}_i\right),
\end{equation}
where $\mathbf{y}_i\in\mathbb{R}^{m\times 1}$ is the one-hot label of training sample data $i$ and $M$ is the batch size.

\subsection{Discussion}
We will discuss the relationship and difference between the proposed PosRec model and existing recommendation methods as well as the position encoding schemes.

\subsubsection{Comparisons with Existing Recommendation Methods}
The proposed PosRec model makes use of a newly designed GNN module for the item representation learning and a Transformer module equipped with the (L)DPE for the session representation learning. Compared with previous GNN based methods, e.g., SR-GNN~\cite{srgnn}, GC-SAN~\cite{gc-san}, FGNN~\cite{fgnn,fgnnj}, and MGNN-SPred~\cite{bc}, PosRec has a newly designed GNN module, PGGNN, which is position aware. The position awareness is the main difference between the GNN modules. In these previous work, the main contributions are including the direction information and the behavioral information in the edges. PGGNN introduces the position information in the edges, which approaches the item representation learning from a novel perspective. In addition, as these research work indicates, there is a sparsity issue in the graph construction and propagation for sessions since there could be no repeated items in the same session to construct a graph rather than a simple link list of items. Different methods try to add more connections by self-loops~\cite{fgnn} and using the cross-session information~\cite{fgnnj}. In the proposed PGGNN method, the anchor nodes are chosen to be the first and the last item, which develops extra meaningful edges due to the requirement of calculating the relationship between normal nodes and the anchor nodes.

\subsubsection{Comparisons with Existing Position Encoding Schemes}
The proposed PosRec model incorporates both the \textit{forward} and \textit{backward} positional information with the (L)DPE. For existing recommendation models~\cite{sasrec,bert4rec}, the most popular position encoding scheme is the LPE introduced in Attention~\cite{attention}. As analyzed above, LPE is a \textit{forward-aware} position encoding scheme. These research work has also investigated the performance of SPE, which is shown to be inferior to the LPE. It could be due to the SPE can only represent the initial intention, a less important factor compared with the latest preference. While LPE still has a possibility of learning an implicit positional information with the learnable embeddings. There are other methods to incorporate the positional information in recommendation, for example, NARM~\cite{narm}, STAMP~\cite{Liu18STAMP}, SR-GNN~\cite{srgnn}, and GC-SAN~\cite{gc-san} have an attention module using the last item as the query to all other items in the session to emphasize on the last position. Such a strategy for the positional information can only consider the positional information of the last interaction, and neglect the positional information of all other interactions. While for the GNN-based methods~\cite{srgnn,gc-san,fgnn,fgnnj,bc}, the relative positional information is contained in the direction of edges. But the relative positional information cannot reflect the absolute positional information, e.g., the latest preference and the initial intention. The proposed (L)DPE simultaneously provides the property of being both \textit{forward} and \textit{backward-aware}, and the learnability for a more representative embedding scheme.

\section{Experiments}
\label{sec:exp}
In this section, we present how extensive experiments are conducted to evaluate the effectiveness of our proposed PosRec model, (L)DPE and the PGGNN module. We will answer the following research questions:
\begin{itemize}
    \item \textbf{RQ1}: How does PosRec perform in the session-based recommendation task? (Section~\ref{sec:overall})
    \item \textbf{RQ2}: How does (L)DPE perform compared with other positional encoding schemes? (Section~\ref{sec:diff-pos})
    \item \textbf{RQ3}: Does anchor node aggregation in PGGNN improve the recommendation? (Section~\ref{sec:exp-anc})
    \item \textbf{RQ4}: What is the visualization of DPE? (Section~\ref{sec:vis})
    \item \textbf{RQ5}: How sensitive is PosRec w.r.t.\ the hyper-parameters? (Section~\ref{sec:exp-sen})
\end{itemize}

\subsection{Setup}
\label{sec:exp-setup}
In this section, we will describe the experimental setup in terms of datasets (Section~\ref{sec:dataset}), the preprocessing procedure (Section~\ref{sec:prepro}), baselines (Section~\ref{sec:baselines}), evaluation metrics (Section~\ref{sec:metric}) and the implementation (Section~\ref{sec:imple}).

\subsubsection{Dataset}
\label{sec:dataset}
Experiments are conducted on two benchmark datasets \textit{Yoochoose} and \textit{Diginetica}, which is consistent with previous methods~\cite{narm,Liu18STAMP,srgnn,fgnn}.

\begin{itemize}
    \item \textit{Yoochoose} is used as a challenge dataset for RecSys Challenge 2015. It is obtained by recording click-streams from an e-commerce website within 6 months. Since \textit{Yoochoose} is a huge dataset, we follow previous methods~\cite{narm,Liu18STAMP,srgnn,fgnn} to further divide this dataset into two subsets according to the timestamp. \textit{Yoo.~1/64} stands for the most recent $\frac{1}{64}$ of the whole dataset and \textit{Yoo.~1/4} for $\frac{1}{4}$ correspondingly.
    \item \textit{Diginetica} is used as a challenge dataset for CIKM cup 2016. It contains the transaction data which is suitable for session-based recommendation.
\end{itemize}

The detailed statistics of each dataset can be found in Table~\ref{tab:datasets}.

\begin{table}[t]
    \centering
    \caption{Statistic of datasets. (Yoo.\ is short for Yoochoose.)}
    \begin{tabular}{lccccc}
         \toprule
         Dataset&Clicks&$\sharp$ Train&$\sharp$ Test&Items&Avg. length\\
         \midrule
         Yoo.\ 1/64&557248&369859&55898&16766&6.16\\
         Yoo.\ 1/4&8326407&5917746&55898&29618&5.71\\
         Diginetica&982961&719470&60858&43097&5.12\\
         \bottomrule
    \end{tabular}
    \label{tab:datasets}
\end{table}

\subsubsection{Preprocessing}
\label{sec:prepro}
For the fairness and the convenience of comparison, we follow~\cite{narm,Liu18STAMP,srgnn,fgnn} to filter out sessions of length 1 and items which occur less than 5 times in each dataset respectively. After the preprocessing step, there are 7,981,580 sessions and 37,483 items remaining in \textit{Yoochoose} dataset, while 204,771 sessions and 43097 items in \textit{Diginetica} dataset. Similar to~\cite{improved}, we split a session of length $n$ into $n-1$ partial sessions of length ranging from $2$ to $n$ to augment the datasets. For the partial session of length $i$ in the session $S$, it is defined as $[v_{s,0},\ldots,v_{s,i-1}]$ with the last item $v_{s,i-1}$ as $v_{label}$. Following~\cite{narm,Liu18STAMP,srgnn,fgnn}, for \textit{Yoochoose} dataset, the most recent portions $1/64$ and $1/4$ of the training sequence are used as two split datasets respectively.

\subsubsection{Baselines}
\label{sec:baselines}
In order to demonstrate the advantage of the proposed PosRec model, we compare it with the following representative methods:
\begin{itemize}
    \item \textbf{POP} is a popularity-based method that always recommends the most popular items in the whole training set, which serves as a strong baseline in some situations although it is simple.
    \item \textbf{S-POP} is a popularity-based method that always recommends the most popular items for the individual session.
    \item \textbf{Item-KNN}~\cite{item-knn} computes the similarity of items by the cosine distance of two item vectors in sessions. Regularization is also introduced to avoid the rare high similarities for unvisited items.
    \item\textbf{BPR-MF}~\cite{bprmf} proposes a BPR objective function which utilizes a pairwise ranking loss to train the ranking model. Following~\cite{narm}, Matrix Factorization is modified to session-based recommendation by using mean latent vectors of items in a session.
    \item \textbf{FPMC}~\cite{fpmc} is a hybrid model for the next-basket recommendation and it achieves state-of-the-art results. For anonymous session-based recommendation, following~\cite{narm}, we omit the user feature directly because of the unavailability.
    \item \textbf{GRU4REC}~\cite{gru4rec} stacks multiple GRU layers to encode the session sequence into a final state. It also applies a ranking loss to train the model.
    \item \textbf{NARM}~\cite{narm} extends to use an attention layer to combine all of the encoded states of RNN, which enables the model to explicitly emphasize on the more important parts of the input.
    \item \textbf{STAMP}~\cite{Liu18STAMP} uses attention layers to replace all RNN encoders in previous work to even make the model more powerful by fully relying on the self-attention of the last item in a sequence. STAMP does not use any kind of positional encoding.
    \item \textbf{SR-GNN}~\cite{srgnn} applies a gated graph convolutional layer~\cite{ggnn} to obtain item embeddings, followed by a self-attention of the last item as STAMP does to compute the sequence level embeddings.
    \item \textbf{FGNN}~\cite{fgnn} is also a graph-based recommender system, which uses the attention mechanism~\cite{attention} in both the item representation learning and the item order learning.
    \item \textbf{GC-SAN}~\cite{gc-san} substitutes the simple attention in the graph embedding learning of SR-GNN with multi-layer Transformers~\cite{attention}.
\end{itemize}

Although SASRec~\cite{sasrec} is originally used in the sequential recommendation task rather than the SBRS task, we can still make this state-of-the-art method adapt to our experiment.
\begin{itemize}
    \item \textbf{SASRec} is highly similar to STAMP that stacks attention layers and use the last hidden layer to predict a user's preference. SASRec makes use of LPE in its original model.
\end{itemize}

\afterpage{
\begin{center}
\begin{longtable}{c|cccc}
    \caption{Overall performance.}
    \label{tab:baseline}\\
    \toprule
    \multirow{2}*{Method}&\multicolumn{4}{c}{\textit{Yoo.\ 1/64}}\\
    \cline{2-5}
    &R@5&R@10&M@5&M@10\\
    \midrule
    POP&2.37&4.56&0.56&1.13\\
    S-POP&9.96&20.18&15.25&17.96\\
    Item-KNN&28.35$\pm$0.13&41.82$\pm$0.08&19.37$\pm$0.12&21.24$\pm$0.07\\
    BPR-MF&7.64$\pm$0.18&20.47$\pm$0.14&8.58$\pm$0.15&11.65$\pm$0.11\\
    FPMC&22.93$\pm$0.09&35.38$\pm$0.19&11.83$\pm$0.17&14.57$\pm$0.10\\
    GRU4REC&37.81$\pm$0.08&50.30$\pm$0.13&20.13$\pm$0.09&22.81$\pm$0.05\\
    NARM&44.69$\pm$0.12&57.56$\pm$0.09&25.43$\pm$0.05&27.16$\pm$0.08\\
    STAMP&46.42$\pm$0.13&58.67$\pm$0.07&28.05$\pm$0.08&29.66$\pm$0.10\\
    SR-GNN&47.33$\pm$0.07&60.04$\pm$0.11&28.32$\pm$0.10&30.08$\pm$0.12\\
    FGNN&47.12$\pm$0.10&60.13$\pm$0.08&28.45$\pm$0.06&30.17$\pm$0.09\\
    GC-SAN&46.48$\pm$0.08&59.47$\pm$0.11&27.58$\pm$0.18&29.33$\pm$0.06\\
    SASRec&46.65$\pm$0.16&58.98$\pm$0.10&28.13$\pm$0.12&29.87$\pm$0.13\\
    \midrule
    PosRec&$\underline{\bm{47.96}}\pm$0.15&$\underline{\bm{60.90}}\pm$0.07&$\underline{\bm{28.83}}\pm$0.12&$\underline{\bm{30.57}}\pm$0.09\\
    \bottomrule
    \toprule
    \multirow{2}*{Method}&\multicolumn{4}{c}{\textit{Yoo.\ 1/4}}\\
    \cline{2-5}
    &R@5&R@10&M@5&M@10\\
    \midrule
    POP&0.76&0.98&0.09&0.15\\
    S-POP&8.69&18.57&14.84&16.87\\
    Item-KNN&30.16$\pm$0.15&41.86$\pm$0.13&18.54$\pm$0.09&20.29$\pm$0.14\\
    BPR-MF&1.15$\pm$0.11&2.61$\pm$0.06&0.79$\pm$0.05&1.13$\pm$0.07\\
    FPMC&-&-&-&-\\
    GRU4REC&36.80$\pm$0.10&49.60$\pm$0.12&19.71$\pm$0.08&21.43$\pm$0.05\\
    NARM&44.95$\pm$0.08&57.73$\pm$0.07&25.60$\pm$0.09&27.39$\pm$0.11\\
    STAMP&45.38$\pm$0.12&58.03$\pm$0.10&27.38$\pm$0.12&29.08$\pm$0.09\\
    SR-GNN&47.71$\pm$0.09&60.64$\pm$0.09&28.33$\pm$0.07&30.24$\pm$0.06\\
    FGNN&47.63$\pm$0.14&60.68$\pm$0.13&28.43$\pm$0.10&30.19$\pm$0.08\\
    GC-SAN&46.97$\pm$0.07&59.86$\pm$0.14&28.12$\pm$0.08&29.72$\pm$0.09\\
    SASRec&45.21$\pm$0.18&57.88$\pm$0.108&27.46$\pm$0.17&29.23$\pm$0.15\\
    \midrule
    PosRec&$\underline{\bm{47.97}}\pm$0.12&$\underline{\bm{60.92}}\pm$0.18&$\underline{\bm{29.29}}\pm$0.07&$\underline{\bm{31.03}}\pm$0.08\\
    \bottomrule
    \toprule
    \multirow{2}*{Method}&\multicolumn{4}{c}{\textit{Diginetica}}\\
    \cline{2-5}
    &R@5&R@10&M@5&M@10\\
    \midrule
    POP&0.34&0.68&0.06&0.13\\
    S-POP&2.67&9.95&9.32&11.79\\
    Item-KNN&11.75$\pm$0.09&24.09$\pm$0.08&8.34$\pm$0.12&10.63$\pm$0.18\\
    BPR-MF&0.78$\pm$0.10&2.28$\pm$0.14&0.07$\pm$0.03&0.83$\pm$0.06\\
    FPMC&4.49$\pm$0.07&15.71$\pm$0.012&2.86$\pm$0.21&5.17$\pm$0.13\\
    GRU4REC&21.87$\pm$0.06&32.67$\pm$0.12&11.83$\pm$0.09&13.26$\pm$0.07\\
    NARM&24.43$\pm$0.07&36.09$\pm$0.09&12.48$\pm$0.07&14.02$\pm$0.10\\
    STAMP&25.85$\pm$0.07&37.46$\pm$0.10&14.42$\pm$0.08&15.96$\pm$0.12\\
    SR-GNN&26.53$\pm$0.09&37.87$\pm$0.14&14.99$\pm$0.05&16.49$\pm$0.06\\
    FGNN&26.46$\pm$0.13&37.82$\pm$0.08&15.06$\pm$0.12&16.55$\pm$0.13\\
    GC-SAN&26.29$\pm$0.08&37.75$\pm$0.07&14.73$\pm$0.12&16.23$\pm$0.09\\
    SASRec&25.87$\pm$0.13&37.53$\pm$0.09&14.37$\pm$0.16&16.12$\pm$0.11\\
    \midrule
    PosRec&$\underline{\bm{27.30}}\pm$0.12&$\underline{\bm{38.87}}\pm$0.16&$\underline{\bm{15.31}}\pm$0.09&$\underline{\bm{16.85}}\pm$0.08\\
    \bottomrule
\end{longtable}
\end{center}}

\subsubsection{Evaluation metrics}
\label{sec:metric}
For each time step, a recommender system should give out a full ranking over the whole item set. According to~\cite{metric}, such a ranking result will lead to a fairer comparison than sampling-based ranking methods for different models. Additionally to keep the same setting as previous baselines, we mainly choose to use two metrics, Recall and Mean Reciprocal Ranking. For both of them, we use top-5 and top-10 result to make comparisons.

\begin{itemize}
    \item \textbf{R@K} (Recall calculated over top-K items). The R@K score is the metric that calculates the proportion of test cases which recommends the correct items in a top K position in a ranking list,
    \begin{equation}
        \text{R@K}=\frac{n_{hit}}{N},
    \end{equation}
    where $N$ represents the number of test sequences $S_{test}$ in the dataset and $n_{hit}$ counts the number that the desired items are in the top K position in the ranking list, which is named the $hit$. R@K is also known as the hit ratio.
    \item \textbf{M@K} (Mean Reciprocal Rank calculated over top-K items). The reciprocal is set to $0$ when the desired items are not in the top K position and the calculation is as follows,
    \begin{equation}
            \text{M@K}=\frac{1}{N}\sum\limits_{v_{label}\in S_{test}}\frac{1}{Rank(v_{label})}.
    \end{equation}
    The Mean Reciprocal Rank is a normalized ranking of $hit$, the higher the score, the better the quality of the recommendation because it indicates a higher ranking position of the desired item.
\end{itemize}

\subsubsection{Implementation}
\label{sec:imple}
We apply one layer of PGGNN and one layer of the attention module for our PosRec. Unless indicated otherwise, we use Adam~\cite{adam} to train our model with an initial learning rate $0.001$ that decreases at the rate $0.1$ for every $3$ epochs. The batch size and the embedding size are set to $100$. To reduce the overfitting, we apply an $l_2$ regularization for all parameters and early stop at the end of the 4-th epoch. For weights in Eq.\ (\ref{eq:h-session}), $\lambda_0$ and $\lambda_1$ are both set to $1$. For \textit{Yoochoose}, $\lambda_2$ is set to $0.5$ while for \textit{Diginetica}, it is set to $1$. And the default position encoding is set to LDPE if there is no further indication. We use one Nvidia GeForce RTX 2080 ti GPU for training.

\subsection{Overall Performance}
\label{sec:overall}
The overall recommendation performance is demonstrated in Table~\ref{tab:baseline}. We compare the PosRec with the following baselines: (1) shallow methods: POP, S-POP, Item-KNN~\cite{item-knn}, BPR-MF~\cite{bprmf} and FPMC~\cite{fpmc}; (2) GRU-based methods: GRU4REC~\cite{gru4rec} and NARM~\cite{narm}; (3) Attention-based method: STAMP~\cite{Liu18STAMP}; (4) GNN-based methods: SR-GNN~\cite{srgnn}, FGNN~\cite{fgnn} and GC-SAN~\cite{gc-san} and (5) adapted sequential methods: SASRec~\cite{sasrec}.

Our PosRec achieves the best performance compared with all baselines across all the datasets by four metrics. Compared with the previous state-of-the-art methods, the improvement is consistent. Our proposed PosRec method effectively exploits the positional information in the positional encoding module and the PGGNN module.

The traditional methods generally cannot compete with the current deep learning models. The popularity-based methods, POP and S-POP, simply recommend items based on the frequencies of appearance in the whole dataset and the current session respectively. These popularity-based approaches tend to recommend fixed items in a general situation, which is not able to learn to recommend the proper items. Although the S-POP model can make use of the session information to improve the performance compared with the POP, it still lacks the ability to learn the session pattern. PosRec performs much better compared with both of them. For shallow learning-based methods, BPR-MF and FPMC, they do not consider the session information in the recommendation, they do not have a competitive performance neither scale up well in the larger dataset \textit{Yoochoose $1/4$}. Item-KNN outperforms all the traditional methods. Item-KNN only considers the closeness between items while avoiding sequential information. Generally, these traditional methods cannot compete with the neural network-based recommender systems.

GRU4REC is the first session-based model that uses a recurrent structure to capture sequential information. GRU4REC outperforms the shallow and popularity-based methods by a large margin, which can be viewed as a strong baseline. The sequential information is encoded explicitly by the recurrent calculation procedure. But the positional information of a session is only included implicitly along with the recurrent calculation. A big issue of the recurrent-based methods is that there is a catastrophic forgetting of the early information. For the attention-based baselines, NARM and STAMP, they apply the self-attention mechanism mainly over the last item, which considers the last item as the pivot item to represent the latest intention. This structure is a better inductive bias in positional information than linearly rolling out as RNN structure by outperforming the GRU4REC. This situation is also proved by the better performance of STAMP compared with NARM. STAMP completely excludes the recurrent structure and relies heavily on the last item. SASRec achieves a slightly higher result than STAMP since they both perform attention while SASRec includes a learned absolute positional encoding. Recent GNN-based methods, SR-GNN, FGNN and GC-SAN, have made improvements by utilizing the connectivity of items. These methods use a session graph to represent the session sequence by linking the interactions according to their chronological order. The relative positional information between interactions is included in the edge and in the graph readout calculation step, it resembles the attention mechanism to focus more on the last item. The positional information is only partially considered with the relative positional information and the \textit{backward-awareness}.

\begin{figure}[t]
    \centering
    \subfigure{
    \includegraphics[width=0.23\linewidth]{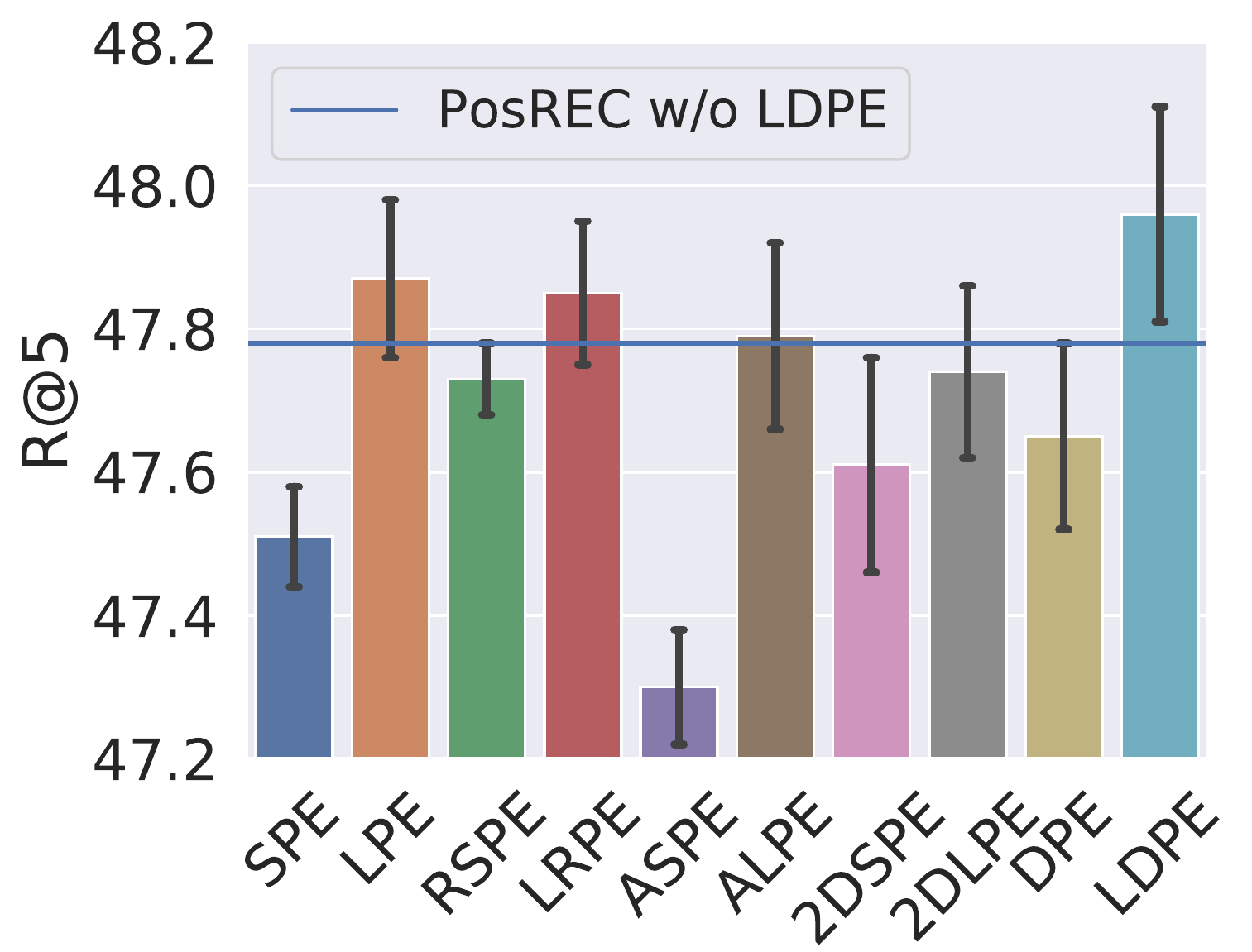}
    }
    \subfigure{
    \includegraphics[width=0.23\linewidth]{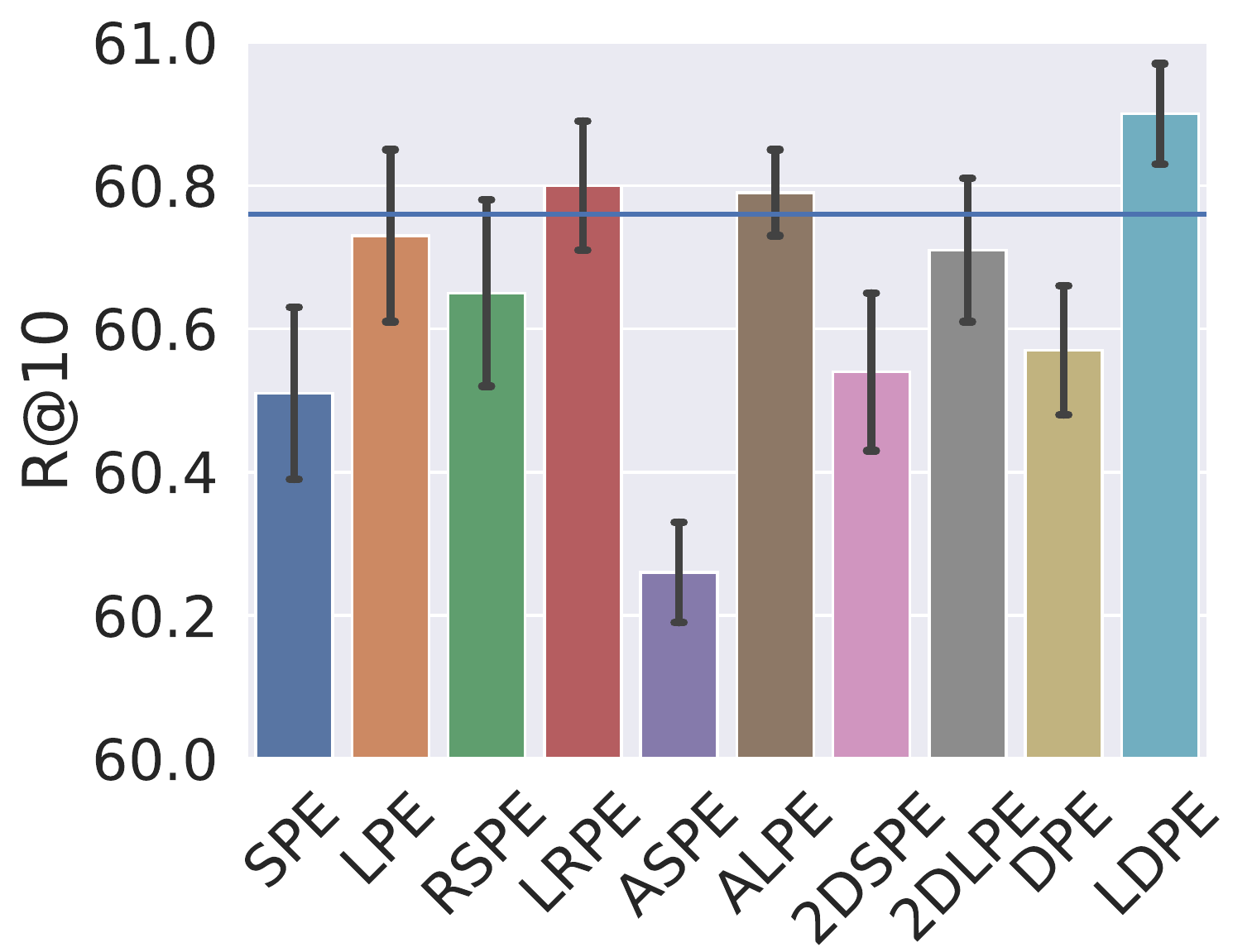}
    }
    \subfigure{
    \includegraphics[width=0.23\linewidth]{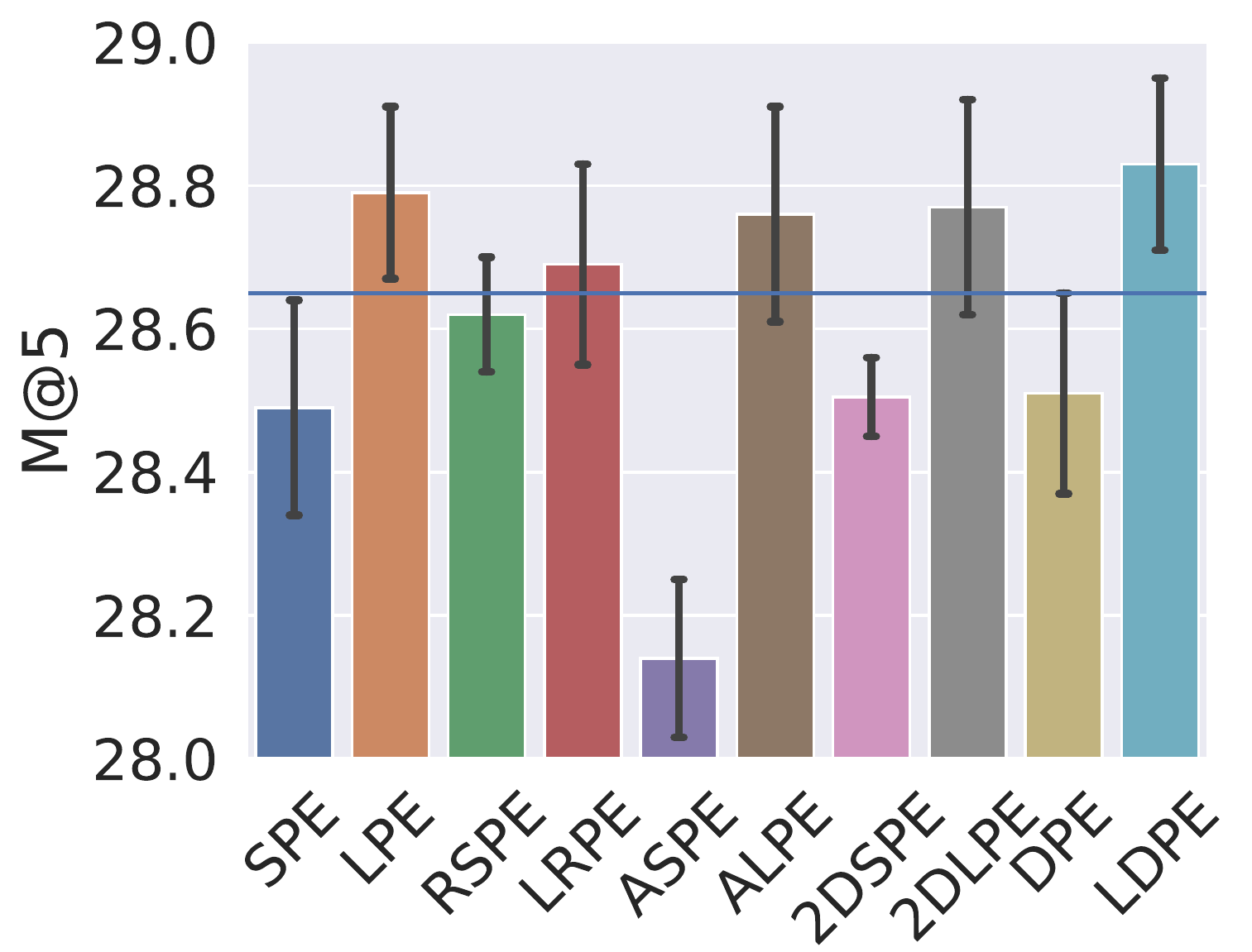}
    }
    \subfigure{
    \includegraphics[width=0.23\linewidth]{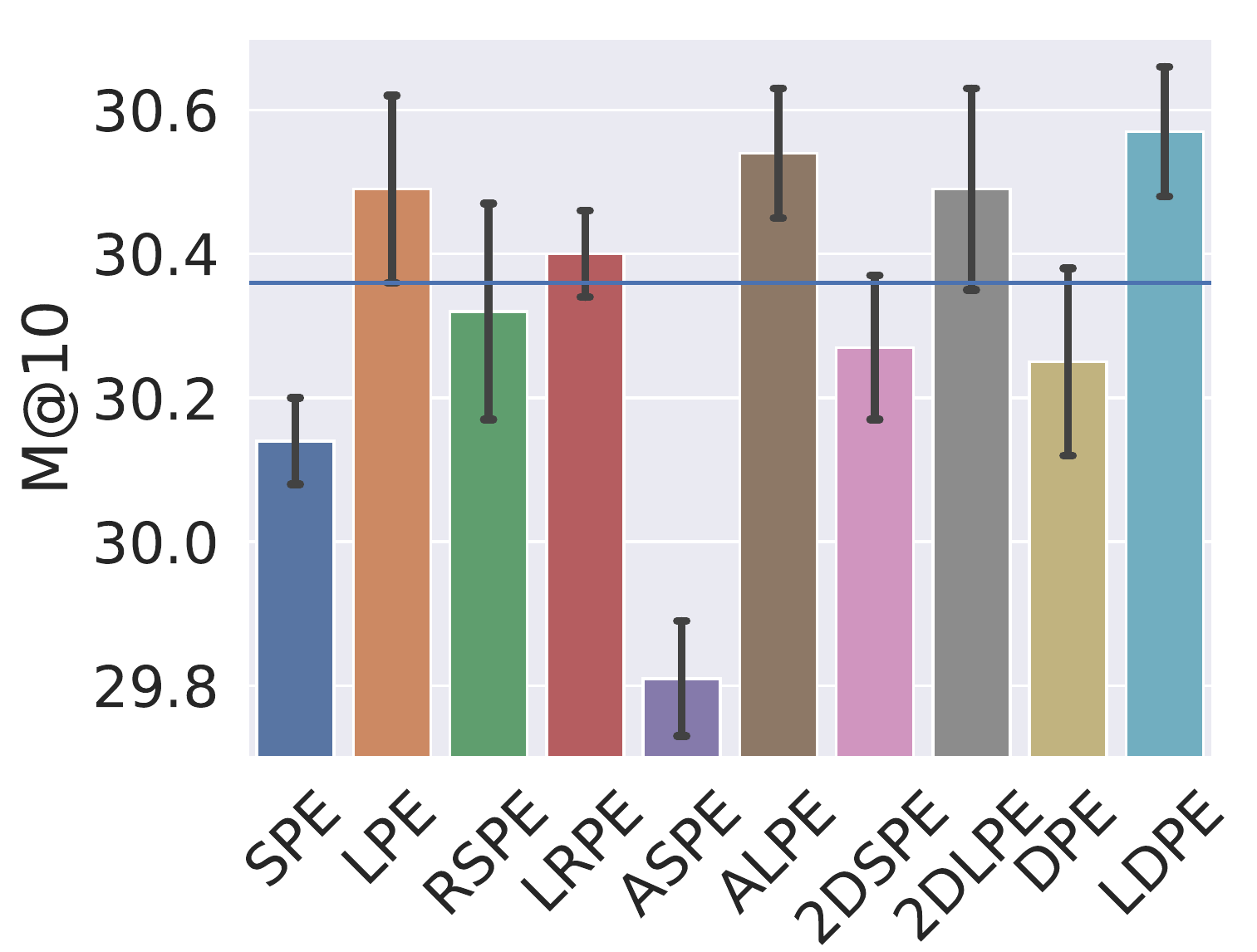}
    }
    \caption{Performance of different positional encoding schemes on \textit{Yoo. 1/64}.}
\label{fig:pe_y}
\end{figure}

\begin{figure}[t]
    \centering
    \subfigure{
    \includegraphics[width=0.23\linewidth]{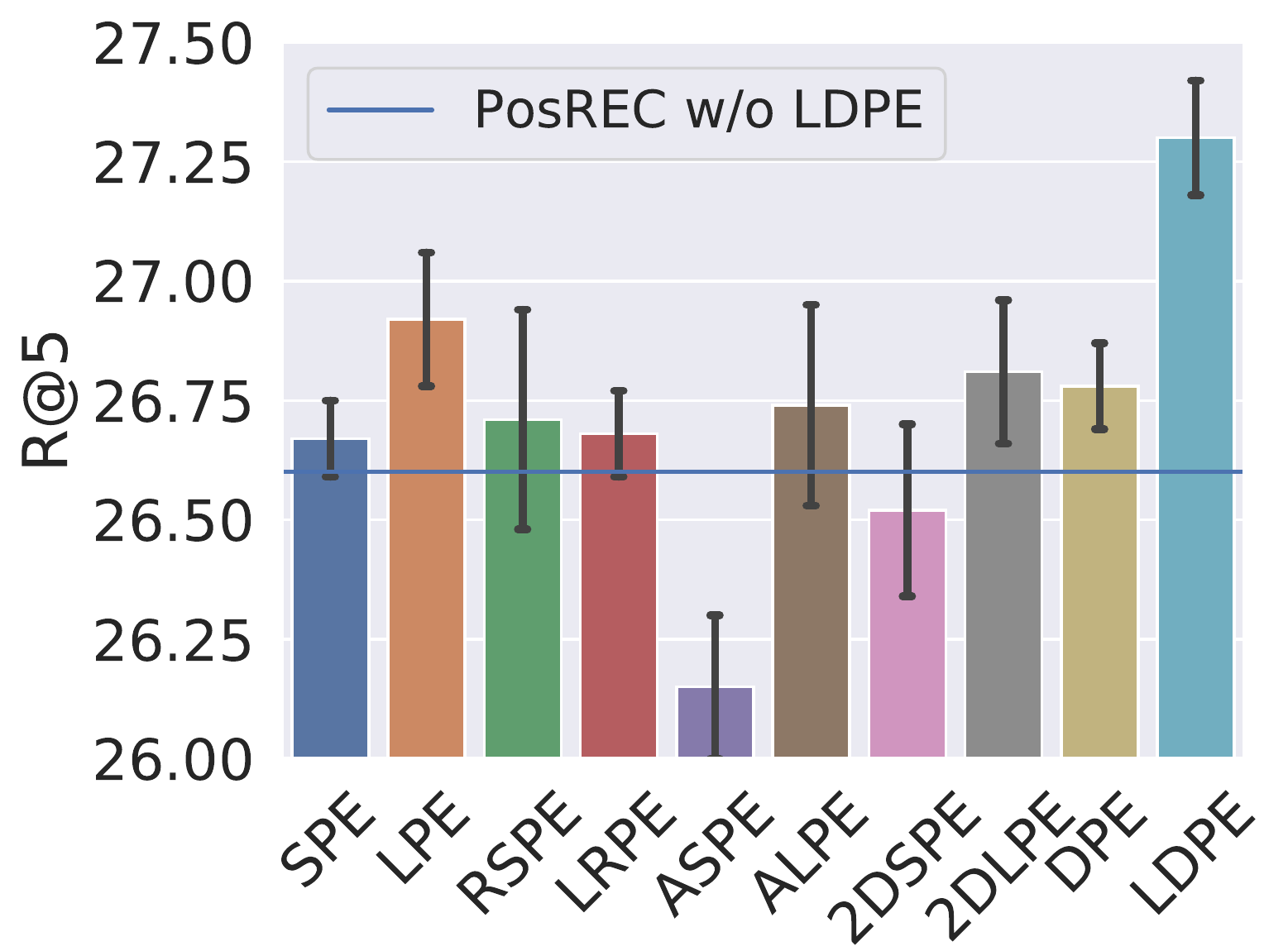}
    }
    \subfigure{
    \includegraphics[width=0.23\linewidth]{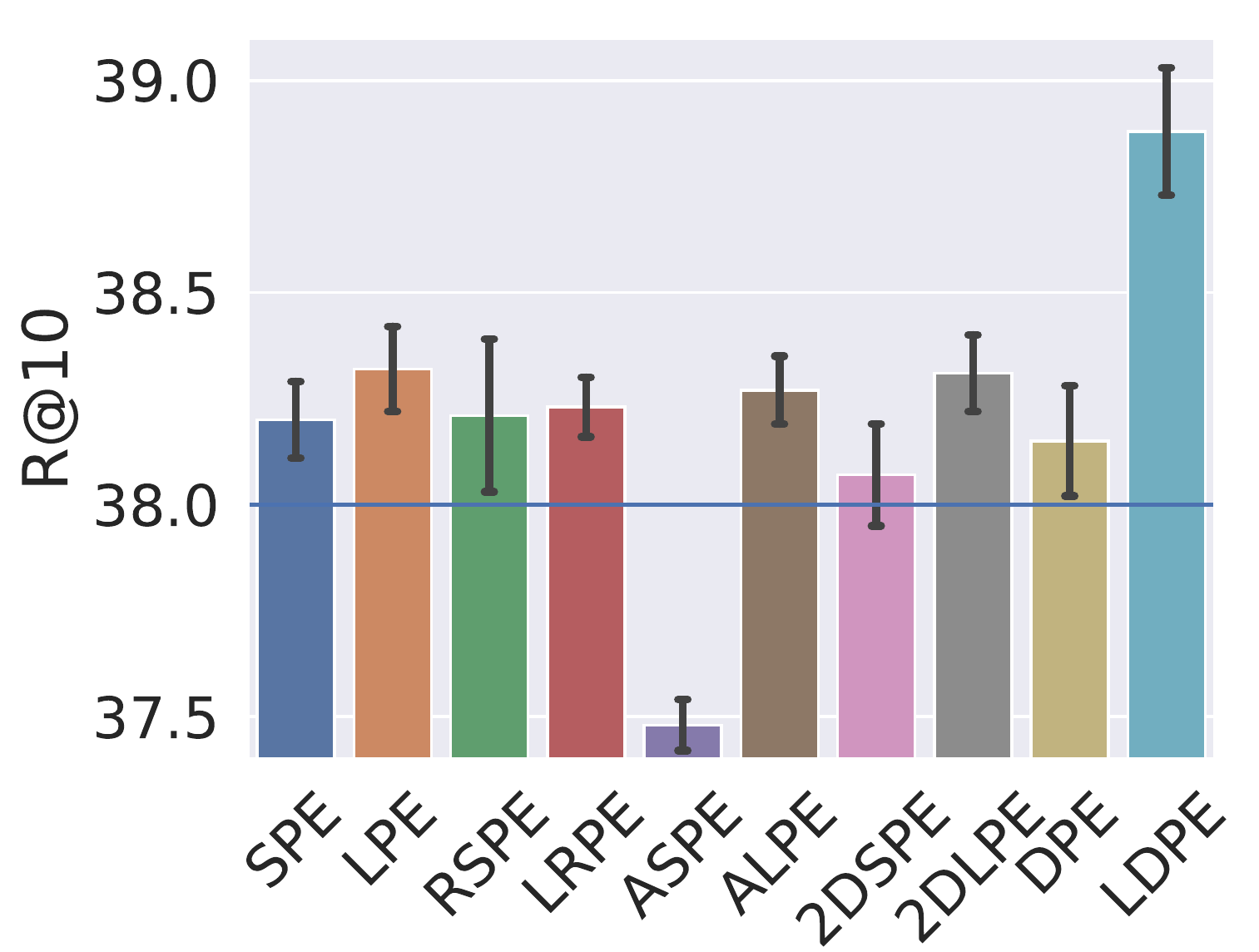}
    }
    \subfigure{
    \includegraphics[width=0.23\linewidth]{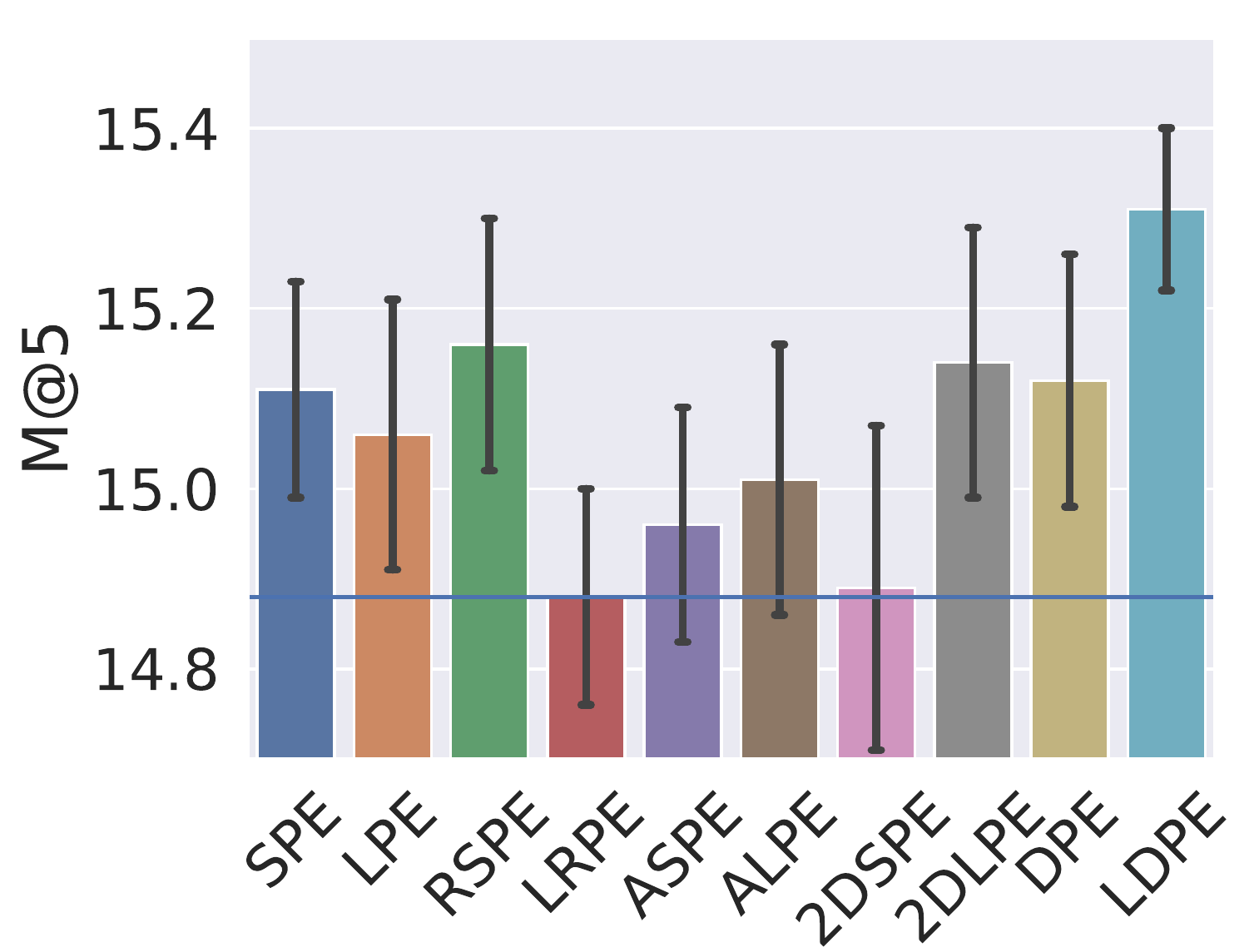}
    }
    \subfigure{
    \includegraphics[width=0.23\linewidth]{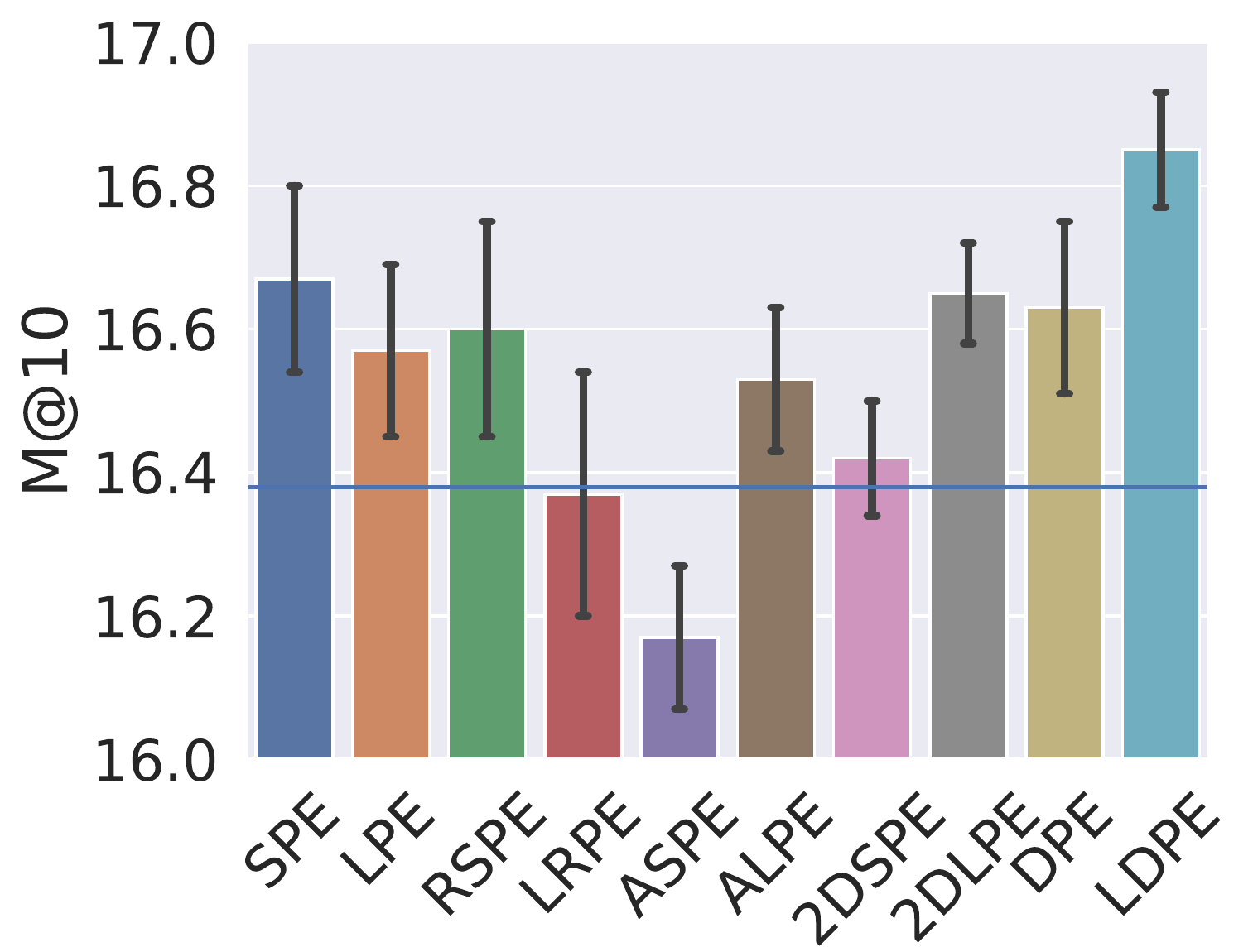}
    }
    \caption{Performance of different positional encoding schemes on \textit{Diginetica}.}
\label{fig:pe_d}
\end{figure}

\subsection{Different Positional Encoding Schemes}
\label{sec:diff-pos}
To evaluate the effect of (L)DPE, we substitute (L)DPE with the following encoding schemes in the bidirectional Transformer module in PosRec: SPE (sinusoidal positional encoding), LPE (learned positional encoding), RSPE (relative sinusoidal positional encoding), LRPE (learned relative positional encoding), ASPE (additional sinusoidal positional encoding), ALPE (additional learned positional encoding), 2DSPE (2D sinusoidal positional encoding) and 2DLPE (2D learned positional encoding), where S indicates a sinusoidal scheme and L indicates a learned scheme. The result is presented in Fig.~\ref{fig:pe_y} and~\ref{fig:pe_d}. To further verify the ubiquitous efficacy of LDPE, we integrate LDPE into the attention-based models, STAMP and SASRec. The result is shown in Table~\ref{tab:ldpe-att}.

Among Fig.~\ref{fig:pe_y} and~\ref{fig:pe_d}, the blue line in each chart indicates the base model that does not include any positional encoding in the PosRec model. Compared with other positional encoding schemes, LDPE can consistently improve the recommendation performance in all situations. For the learnable scheme, LPE and LRPE both achieve relatively good performance because they represent parts of the positional information. The LPE scheme is adopted by SASRec~\cite{sasrec} and BERT4Rec~\cite{bert4rec} for the sequential recommendation as well. There is a small gap for LPE to outperform LRPE. It could be because LPE provides a further \textit{forward-aware} information while the readout function has already included the last item as the \textit{backward-aware} information. ALPE and 2DLPE are not theoretically considered as suitable for SBRS. But they still have a comparable result because their entries of the PE are unique so that the model can still learn the pattern of the encoding but in a harder way. As for the parameter-free scheme, LPE and RSPE achieve comparable results with DPE because they provide reasonable positional information to the model. ASPE has the worst result because the \textit{forward-aware} and \textit{backward-aware} information is entangled in the representation. In contrast, 2DSPE has a clearer connection between entries than ASPE.

\begin{table}[t]
    \centering
    \begin{tabular}{c|cccc|cccc}
         \toprule
         \multirow{2}*{Method}&\multicolumn{4}{c}{\textit{Yoo.\ 1/64}}&\multicolumn{4}{c}{\textit{Diginetica}}\\
    \cline{2-9}
    &R@5&R@10&M@5&M@10&R@5&R@10&M@5&M@10\\
         \midrule
         STAMP&46.42&58.67&28.05&29.66&25.85&37.46&14.42&15.96\\
         STAMP+LDPE&46.78&58.82&28.31&29.94&26.23&37.81&14.79&16.34\\
         \midrule
         SASRec&46.65&58.98&28.13&29.87&25.87&37.53&14.37&16.12\\
         SASRec+LDPE&46.89&59.17&28.37&30.11&26.26&37.90&14.71&16.53\\
         \bottomrule
    \end{tabular}
    \caption{Performance of LDPE on attention-based models.}
    \label{tab:ldpe-att}
\end{table}

In Table~\ref{tab:ldpe-att}, it is shown that LDPE can consistently improve the performance of attention-based models. STAMP does not originally use any positional encoding. While SASRec has already utilized an LPE in its basic method. For both \textit{Yoochoose} and \textit{Diginetica} dataset, LDPE can improve the performance of LPE by a large margin with the same number of parameters, which verifies that LDPE can be integrated into different models.

\begin{figure}[t]
    \centering
    \subfigure{
    \includegraphics[width=0.23\linewidth]{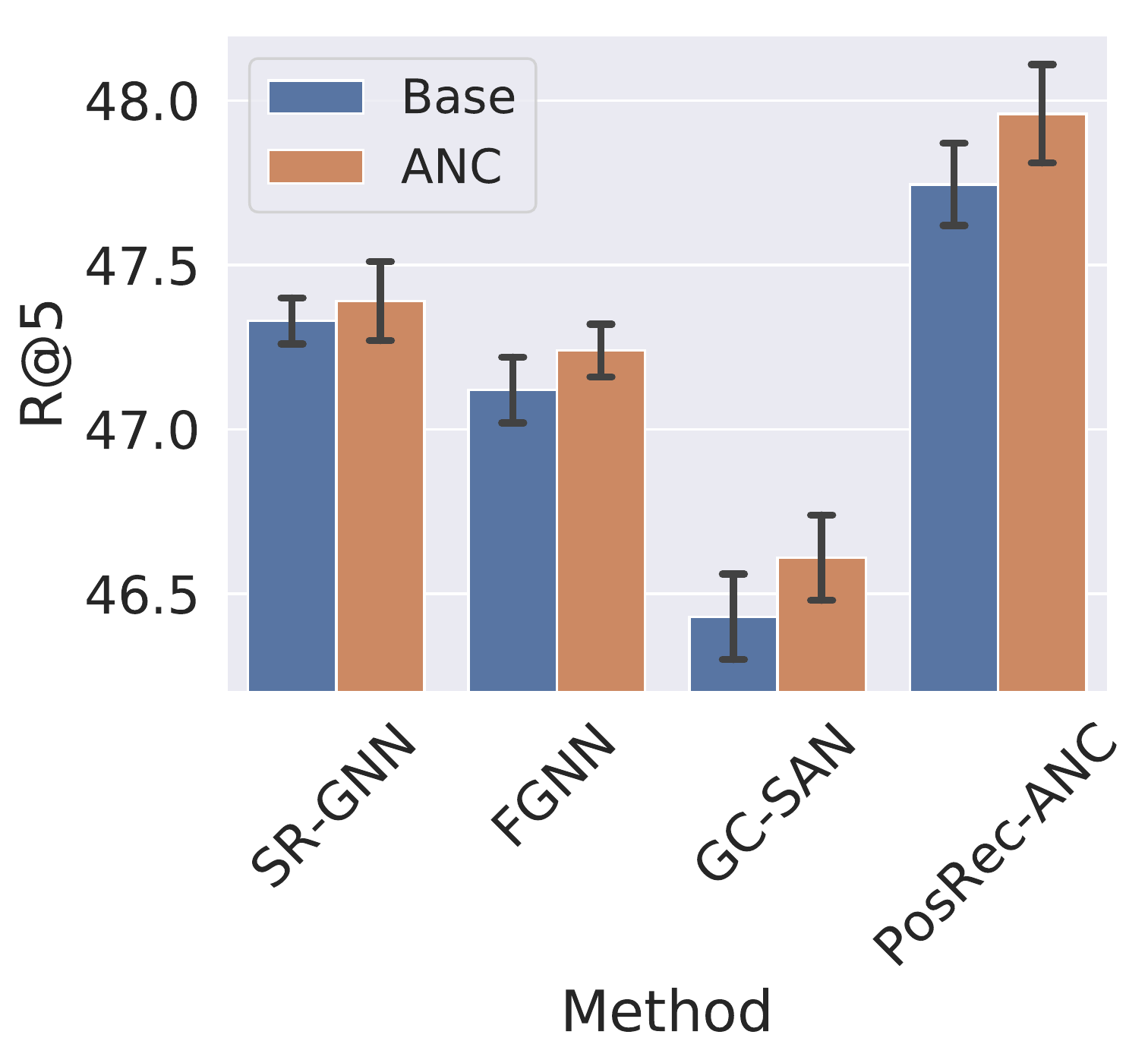}
    }
    \subfigure{
    \includegraphics[width=0.23\linewidth]{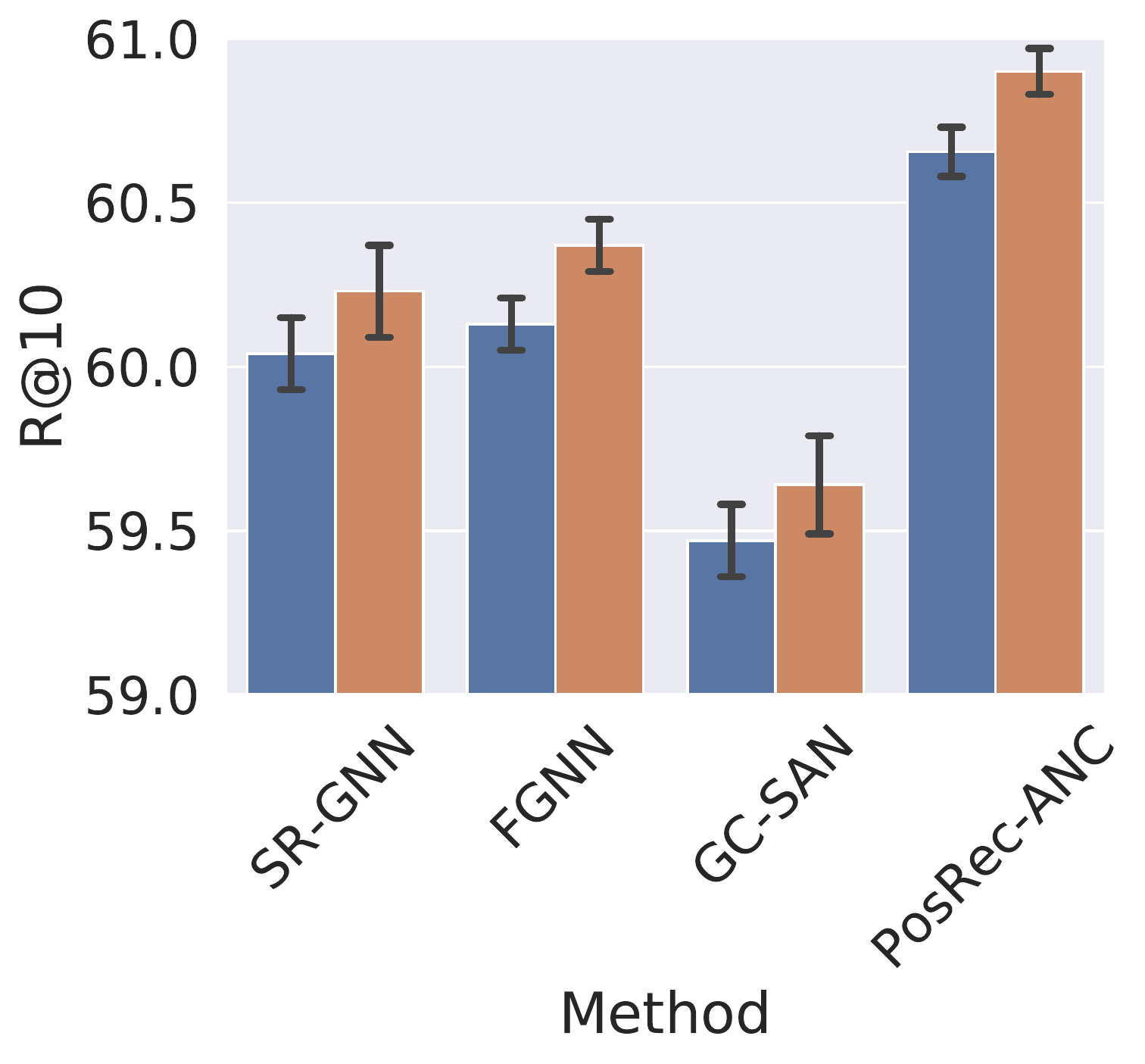}
    }
    \subfigure{
    \includegraphics[width=0.23\linewidth]{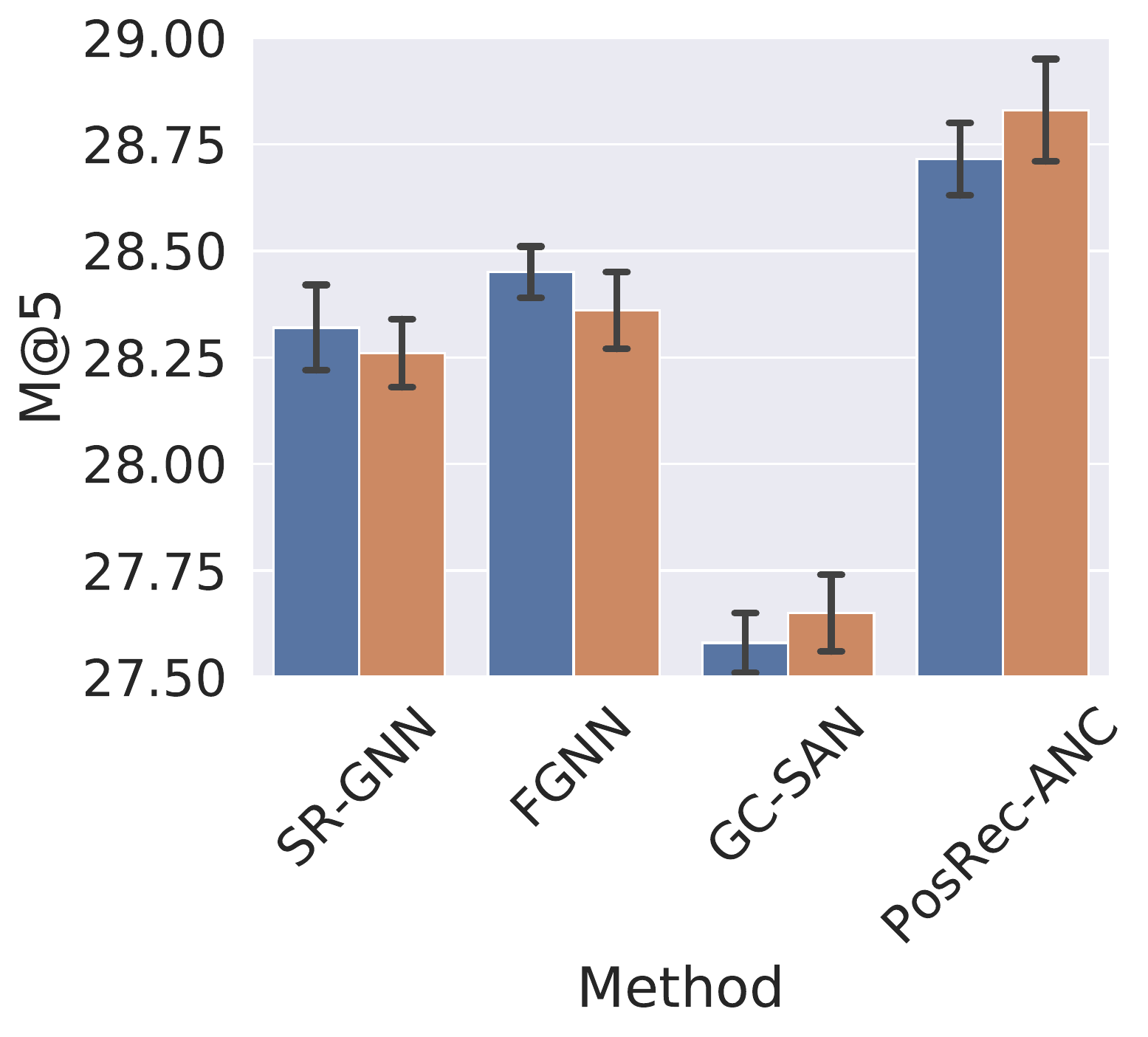}
    }
    \subfigure{
    \includegraphics[width=0.23\linewidth]{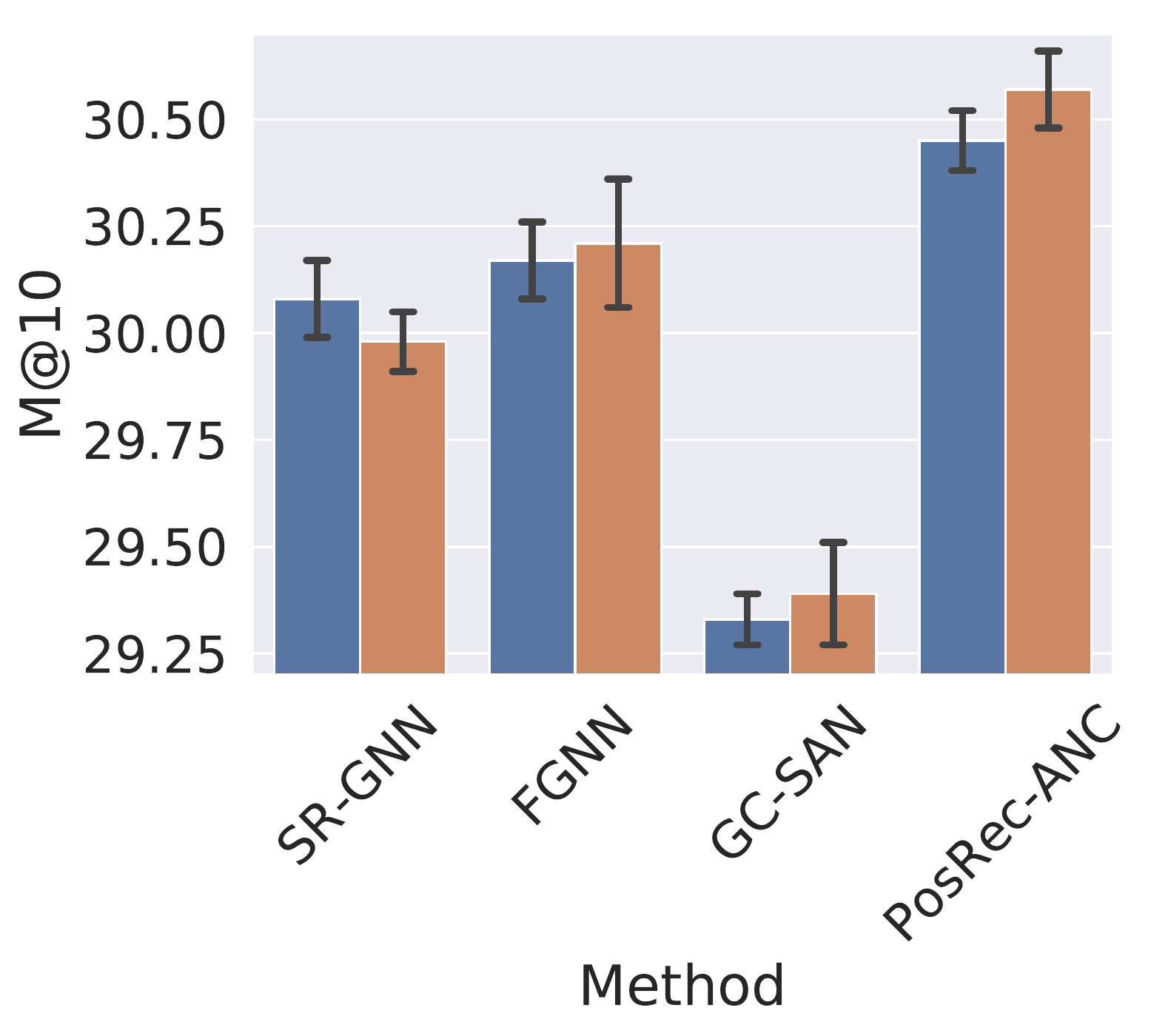}
    }
    \caption{Performance of anchor node aggregation on \textit{Yoo. 1/64}.}
\label{fig:anc_y}
\end{figure}

\begin{figure}[t]
    \centering
    \subfigure{
    \includegraphics[width=0.23\linewidth]{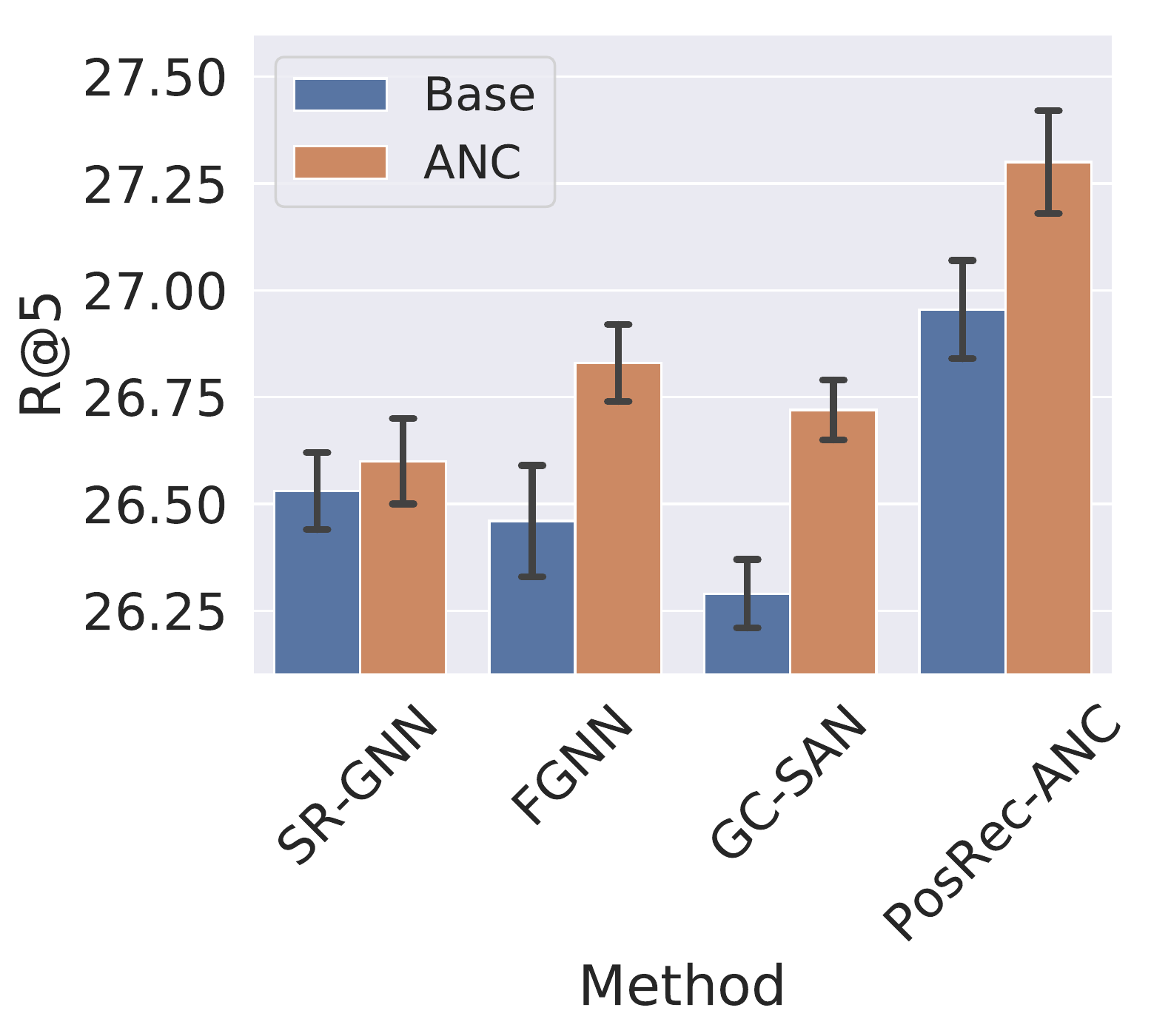}
    }
    \subfigure{
    \includegraphics[width=0.23\linewidth]{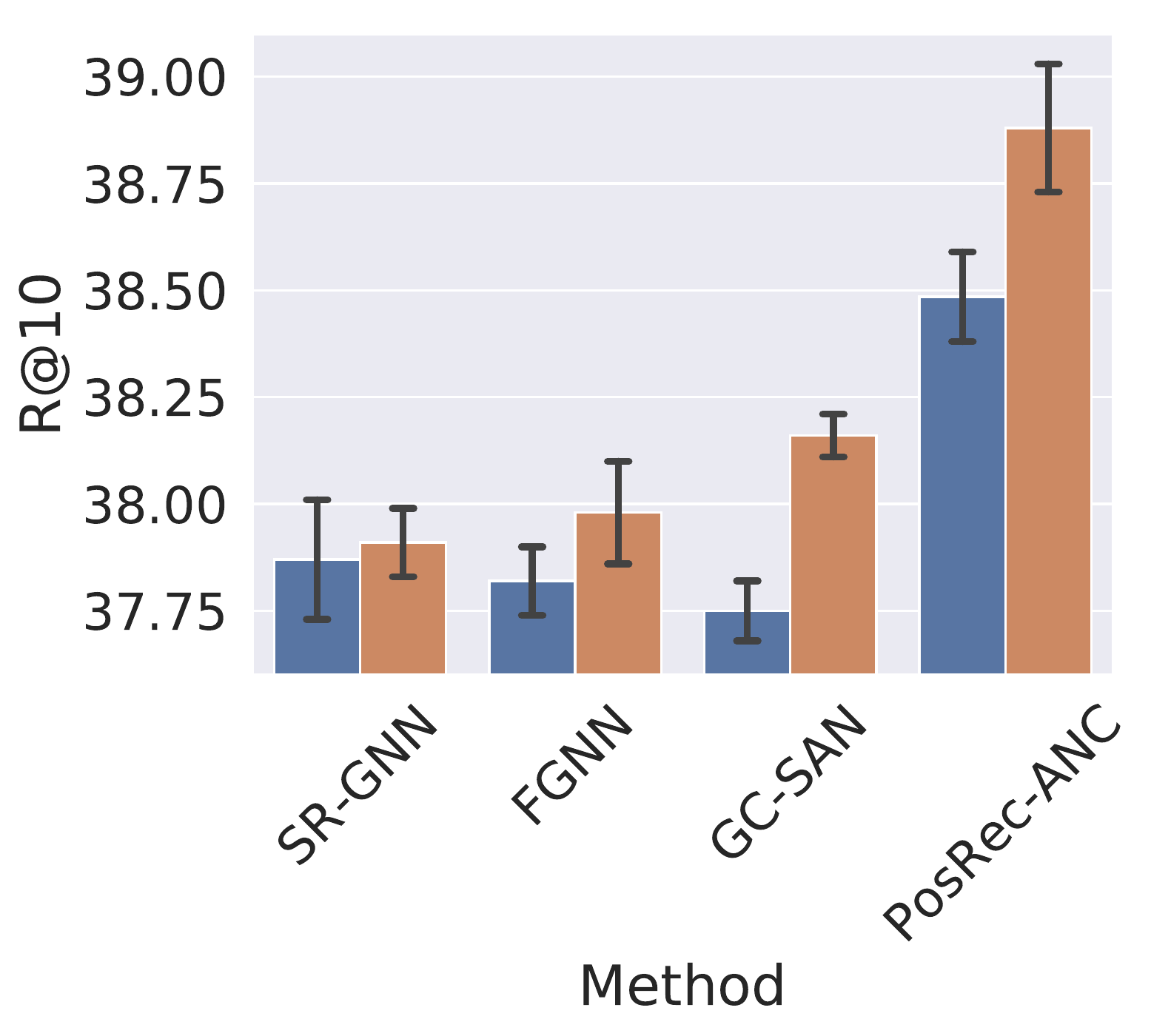}
    }
    \subfigure{
    \includegraphics[width=0.23\linewidth]{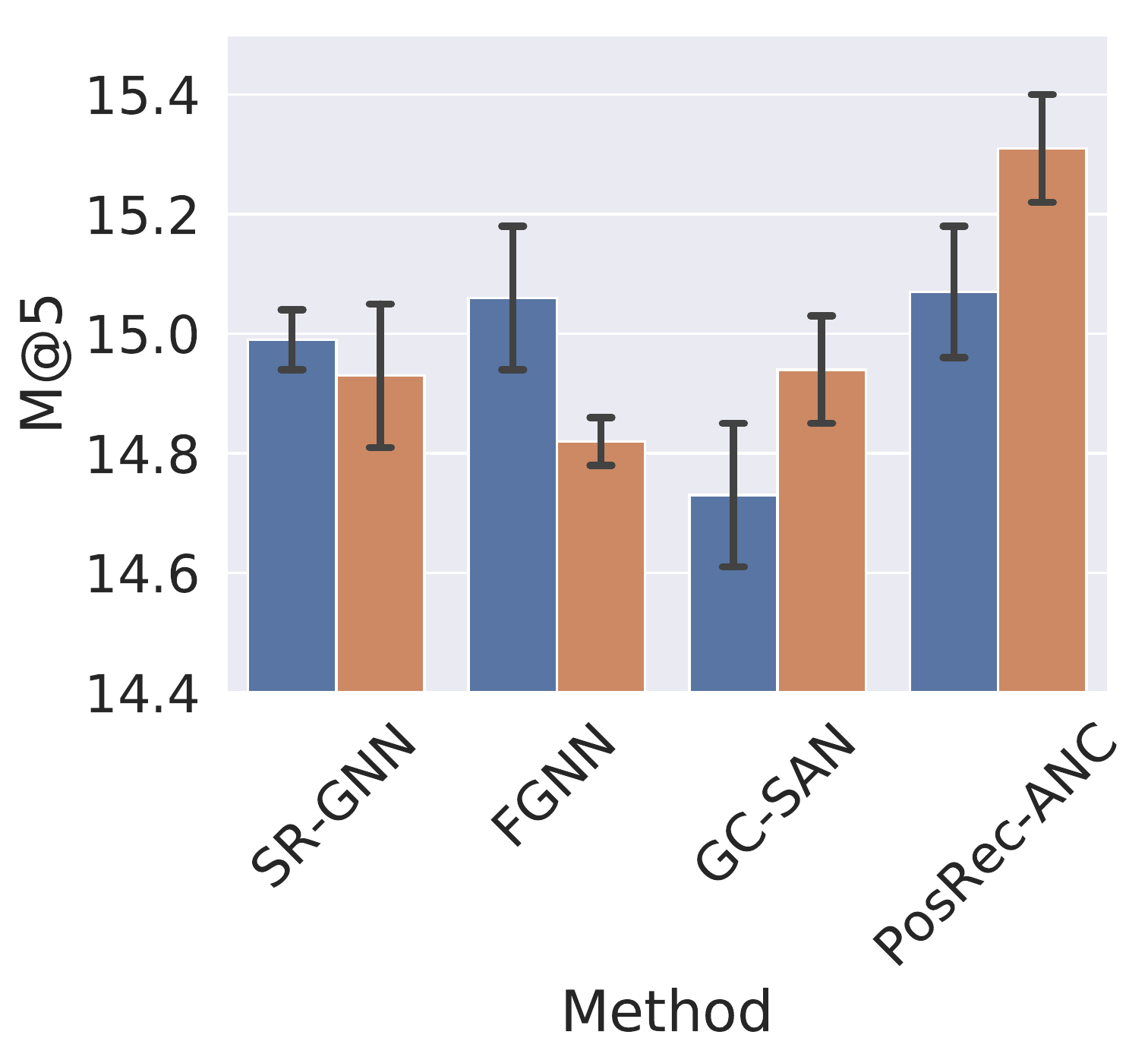}
    }
    \subfigure{
    \includegraphics[width=0.23\linewidth]{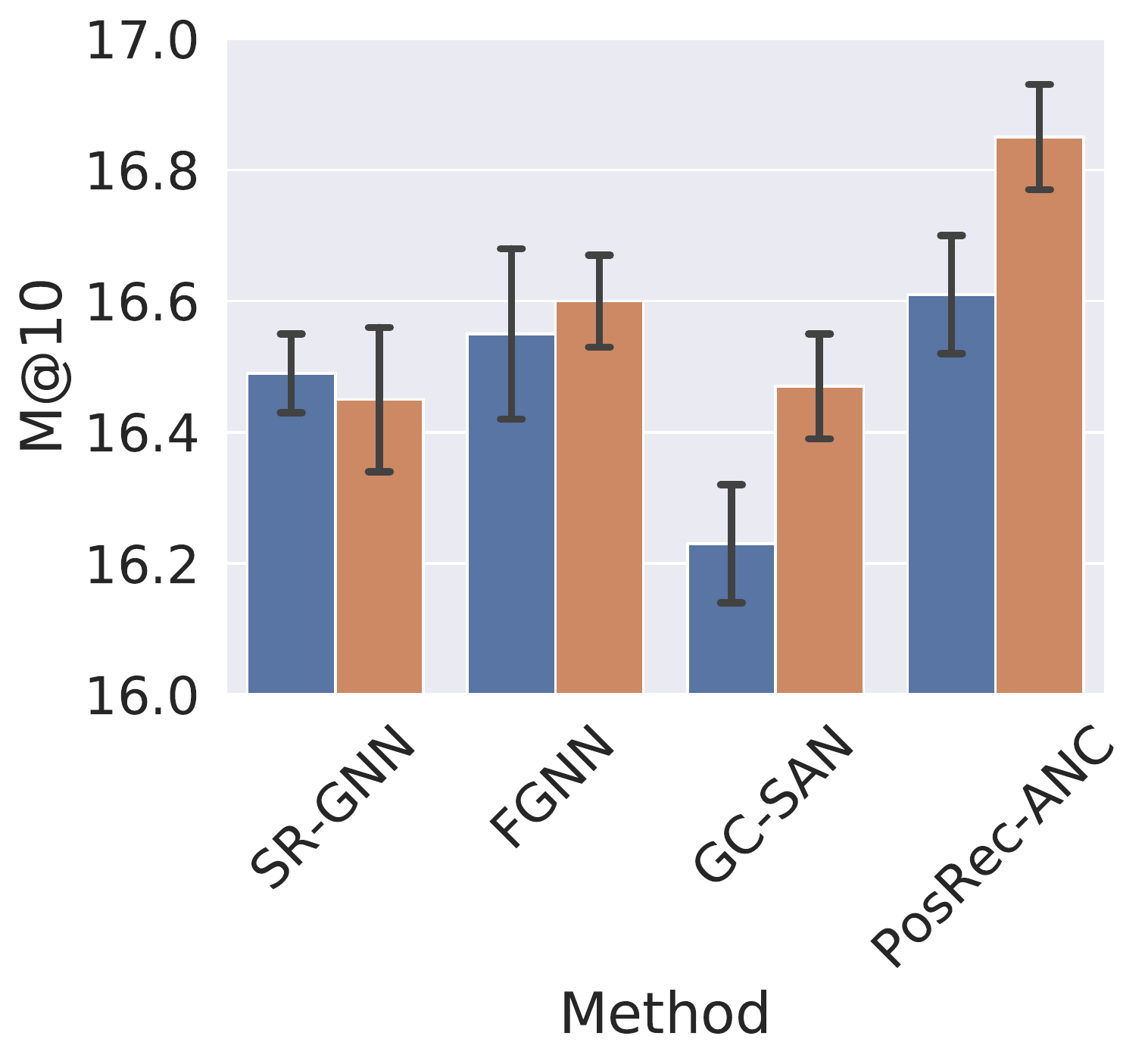}
    }
    \caption{Performance of anchor node aggregation on \textit{Diginetica}.}
\label{fig:anc_d}
\end{figure}

\subsection{Effect of Anchor Node Aggregation}
\label{sec:exp-anc}
To evaluate the effect of the anchor node aggregation in PGGNN, we add this module to GNN-based methods: SR-GNN, FGNN, GC-SAN and PosRec-ANC (indicating the PosRec model removing the anchor node aggregation). The result is presented in Fig.~\ref{fig:anc_y} an~\ref{fig:anc_d}.

It is shown that with the anchor node aggregation, the performance of all GNN-based methods has an increase for Recall. While for the MRR metric, only GC-SAN and PosRec gain a steady improvement. Anchor nodes improve the model performance: (1) injecting the inductive bias of the distance between normal items and important items to the model; (2) increasing the connectivity of the session graph. On one hand, the \textit{position-aware} node feature learning is proved by~\cite{pgnn} that it can provide the positional information in addition to the structure information. On the other hand, session data is originally sparse in the aspect of connectivity. A large portion of session data is even a simple sequence without any repetitive items. After including the aggregation of anchor nodes, the connectivity is increased because there are additional edges between normal nodes and anchor nodes. GNN layers are more suitable to process such data.

\subsection{Visualization}
\label{sec:vis}
In this experiment, we visualize how SPE, DPE, LPE, and LDPE look like and their characteristics. We firstly show the SPE and DPE themselves in Fig.~\ref{fig:pe10}. We examine the nearest neighbor relationship between encoding of different positions in Fig.~\ref{fig:dot}. Example session lengths are set to 10 and 20 and the embedding size is set to 100. Fig.~\ref{fig:spe10} and~\ref{fig:dpe10} show the SPE and DPE for length 10 respectively. We further present the heatmap of LPE and LDPE of length 10 and 20 in Fig.~\ref{fig:lpe_dot_two} and Fig.~\ref{fig:ldpe_dot_two}, which are learned based on the experiment on \textit{Yoo.\ 1/64}.

\begin{figure}[t]
    \centering
    \subfigure[SPE of len.\ 10.]{
    \label{fig:spe10}
    \includegraphics[width=0.475\linewidth]{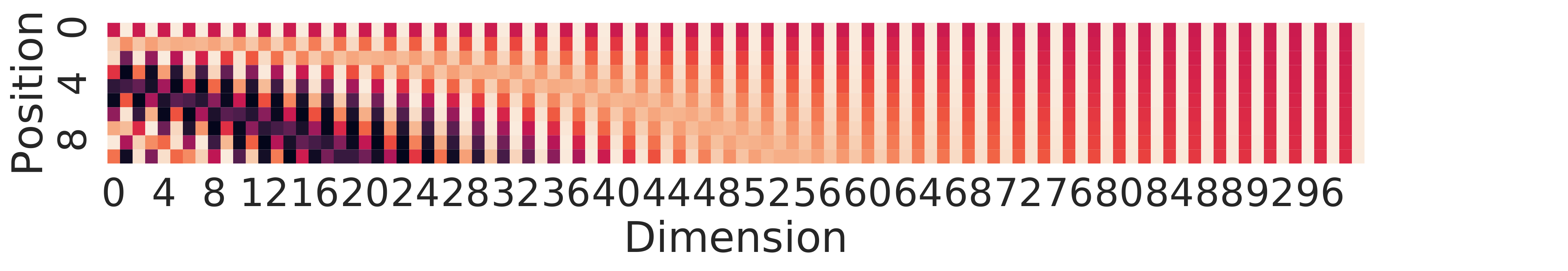}
    }
    \subfigure[DPE of len.\ 10.]{
    \label{fig:dpe10}
    \includegraphics[width=0.475\linewidth]{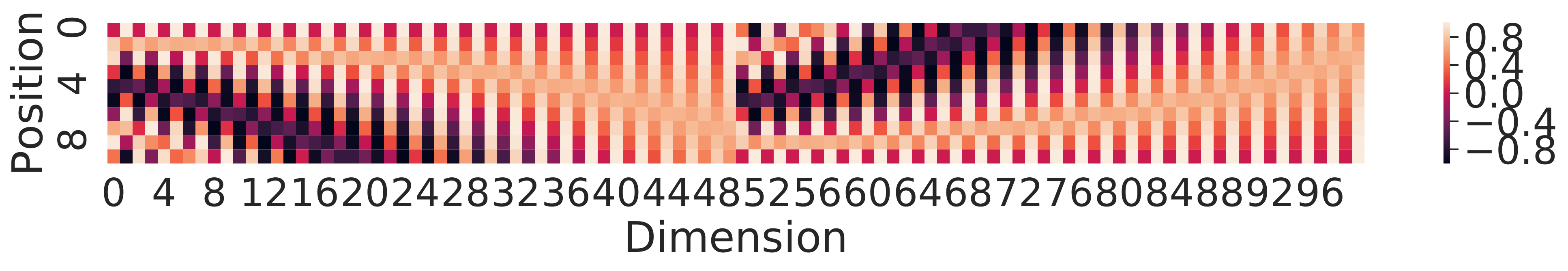}
    }
    \caption{Positional encoding visualization for the session length of 10.}
\label{fig:pe10}
\end{figure}

\begin{figure}[t]
    \centering
    \subfigure[SPE of len.\ 20.]{
    \label{fig:spe20}
    \includegraphics[width=0.475\linewidth]{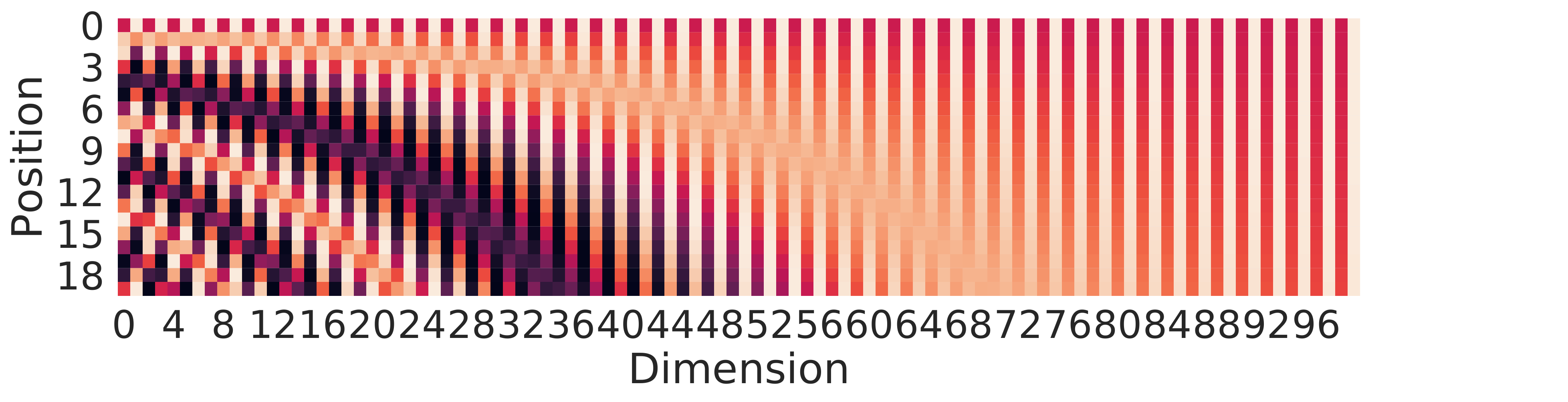}
    }
    \subfigure[DPE of len.\ 20.]{
    \label{fig:dpe20}
    \includegraphics[width=0.475\linewidth]{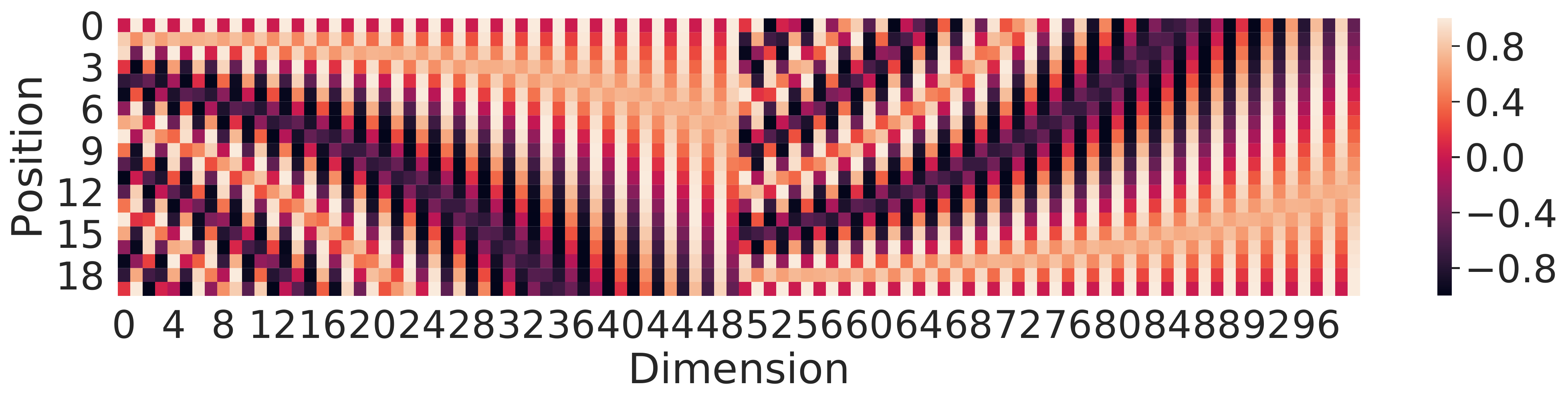}
    }
    \caption{Positional encoding visualization for the session length of 20.}
\label{fig:pe20}
\end{figure}

\begin{figure}[t]
    \centering
    \subfigure[SPE of len.\ 10.]{
    \label{fig:spe_dot10}
    \includegraphics[width=0.2\linewidth]{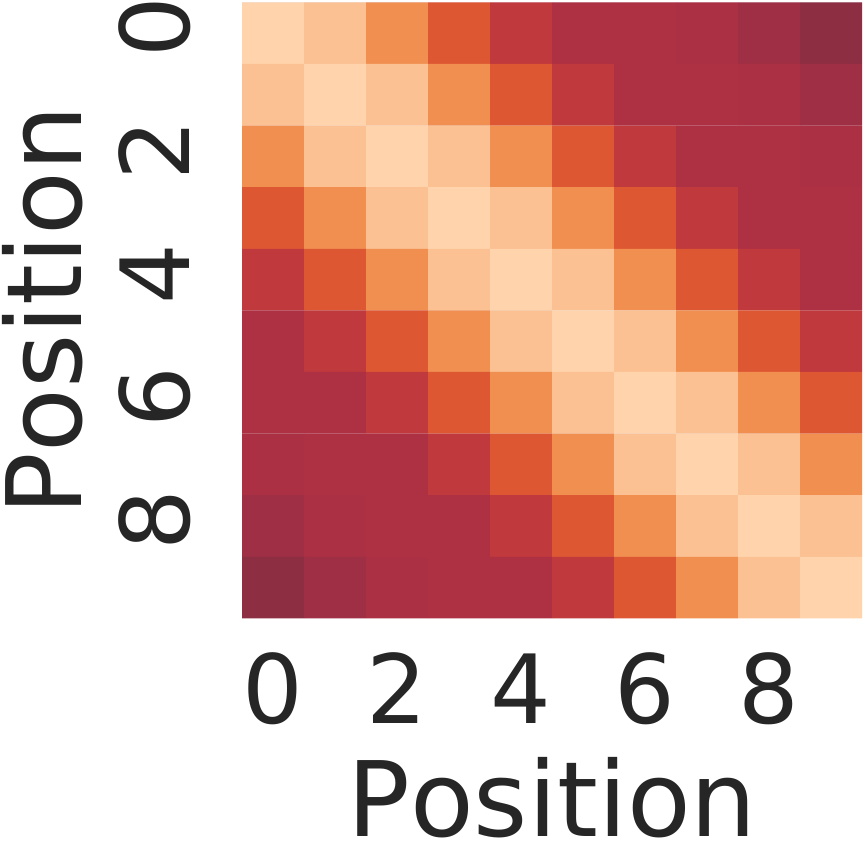}
    }
    \subfigure[SPE of len.\ 20.]{
    \label{fig:spe_dot20}
    \includegraphics[width=0.2\linewidth]{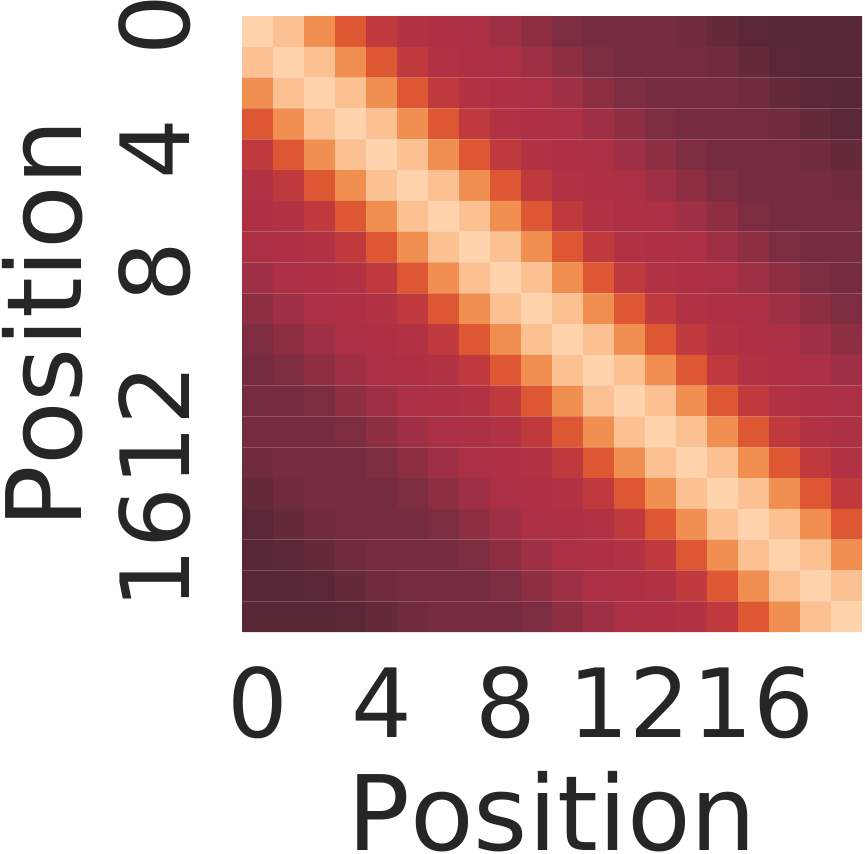}
    }
    \subfigure[DPE of len.\ 10.]{
    \label{fig:dpe_dot10}
    \includegraphics[width=0.2\linewidth]{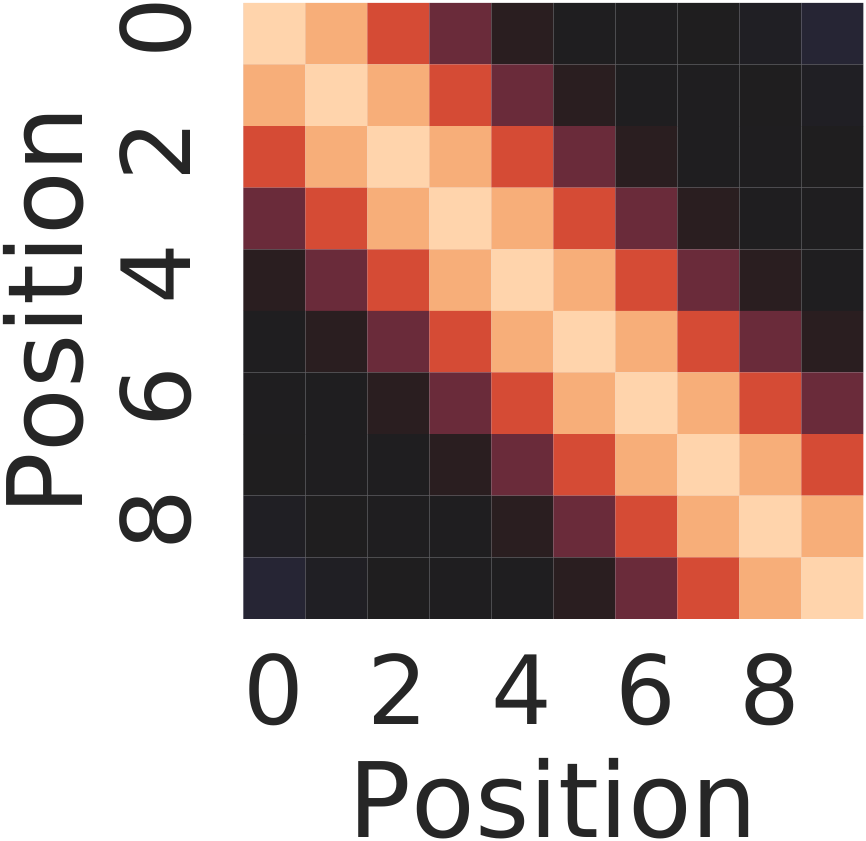}
    }
    \subfigure[DPE of len.\ 20.]{
    \label{fig:dpe_dot20}
    \includegraphics[width=0.251\linewidth]{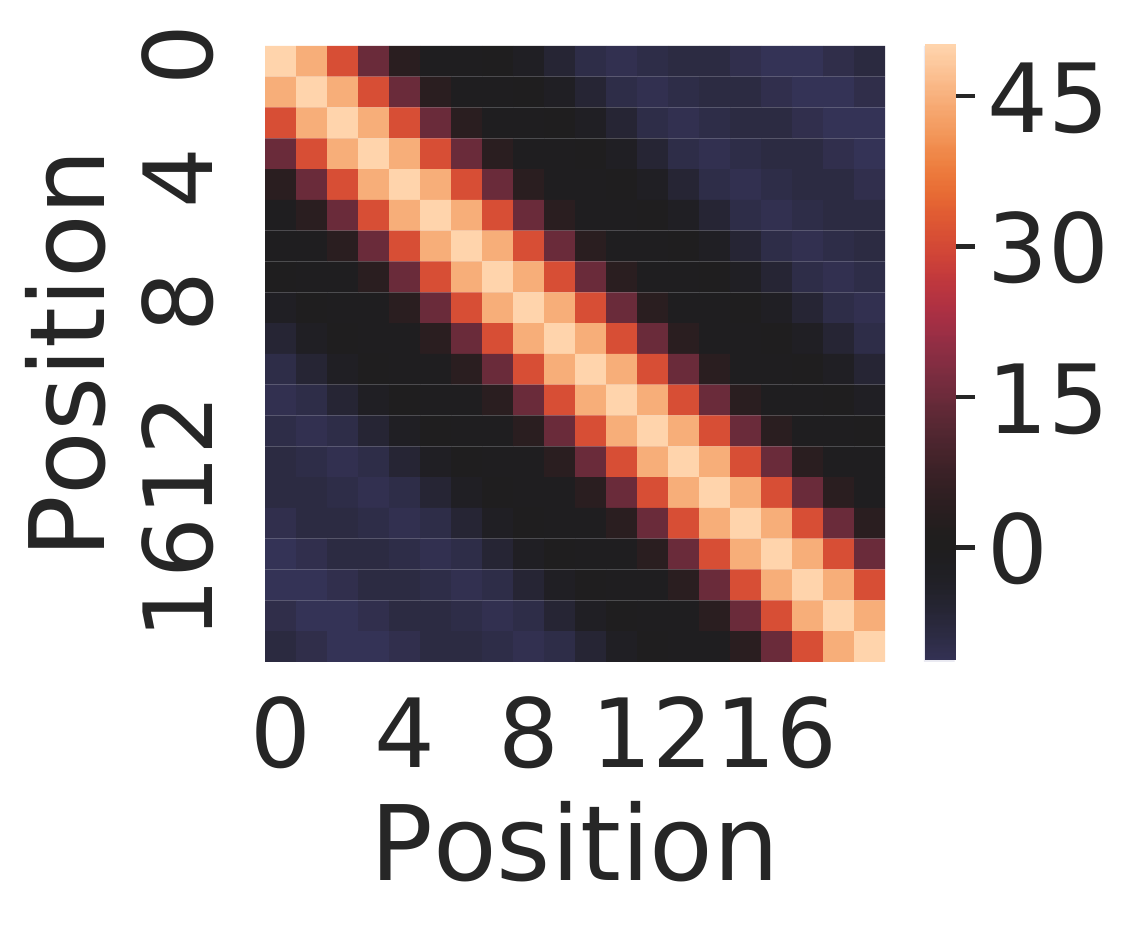}
    }
    \caption{Nearest neighbor heatmap of the same length.}
\label{fig:dot}
\end{figure}

\begin{figure}[t]
    \centering
    \subfigure[SPE of len.\ 10 and 20.]{
    \label{fig:spe_dot_two10}
    \includegraphics[width=0.41\linewidth]{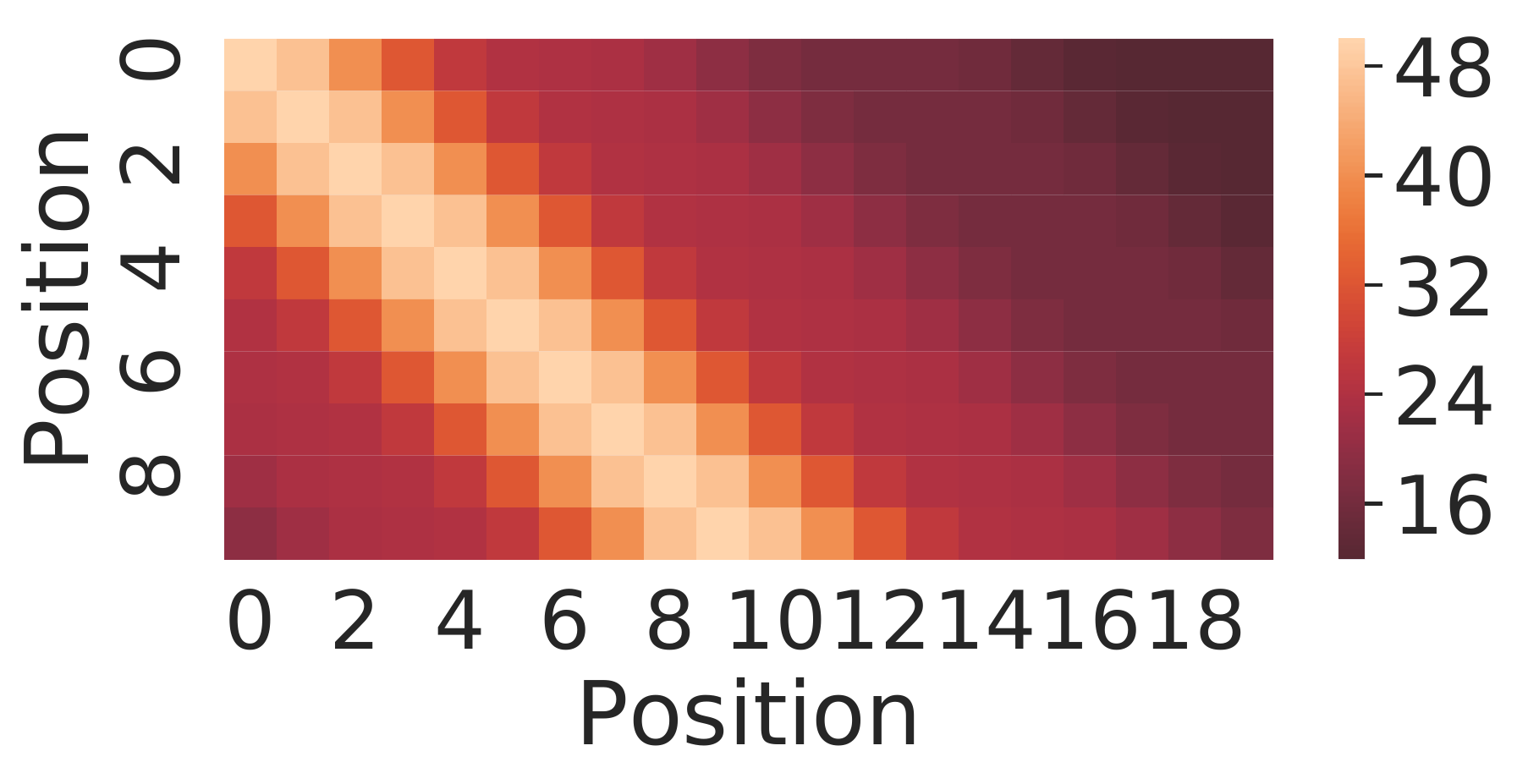}
    }
    \subfigure[DPE of len.\ 10 and 20.]{
    \label{fig:dpe_dot_two10}
    \includegraphics[width=0.41\linewidth]{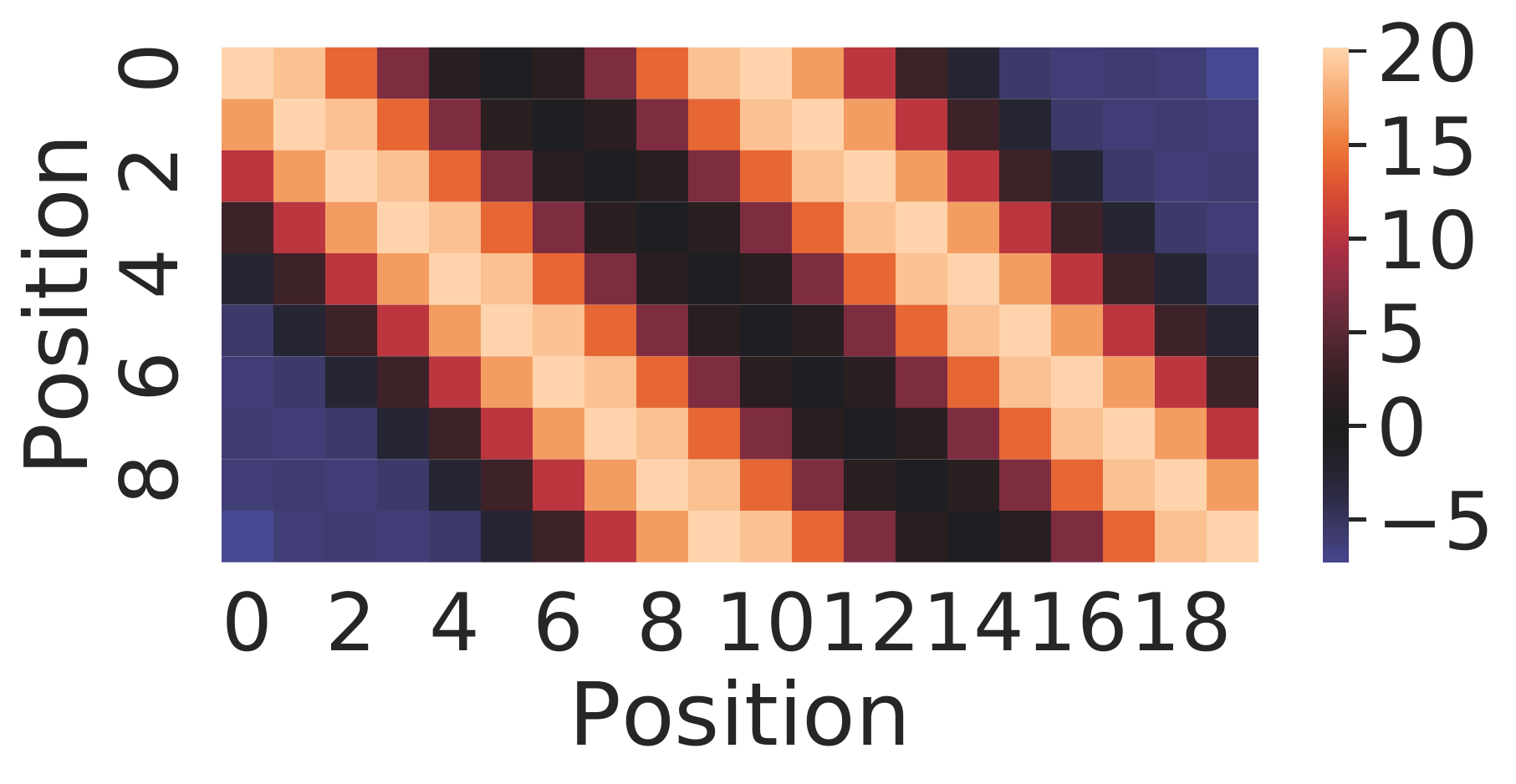}
    }
   \caption{Nearest neighbor heatmap of length 10 and 20.}
\label{fig:dot_two}
\end{figure}

\begin{figure}[t]
    \centering
    \subfigure[LPE of len.\ 10.]{
    \label{fig:lpe_dot_10}
    \includegraphics[width=0.2\linewidth]{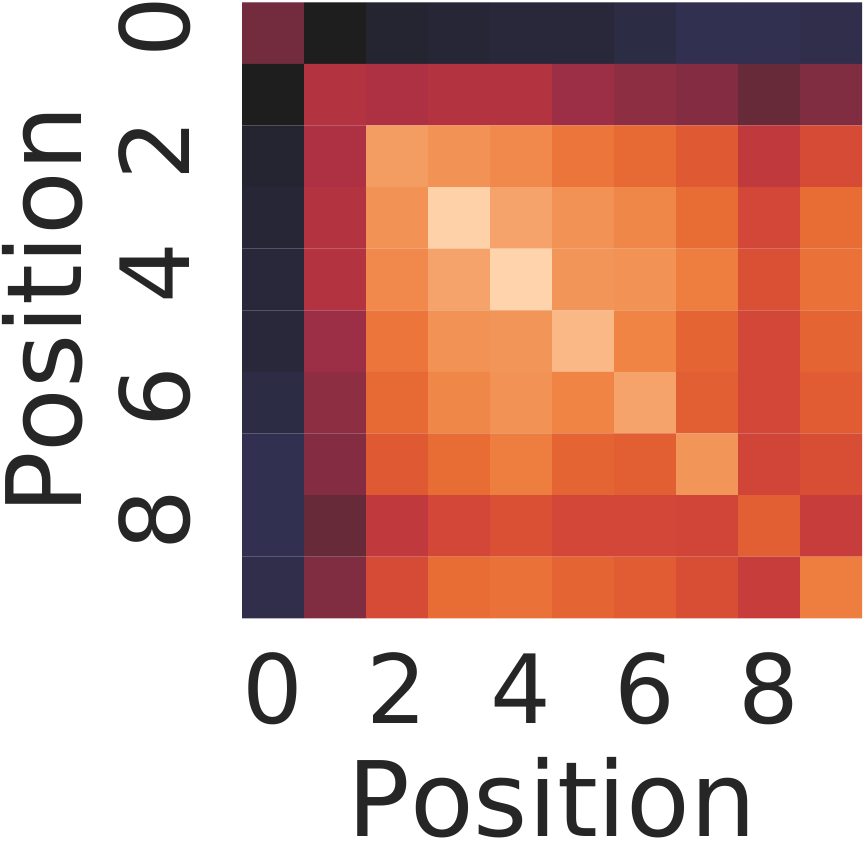}
    }
    \subfigure[LPE of len.\ 20.]{
    \label{fig:lpe_dot_20}
    \includegraphics[width=0.2\linewidth]{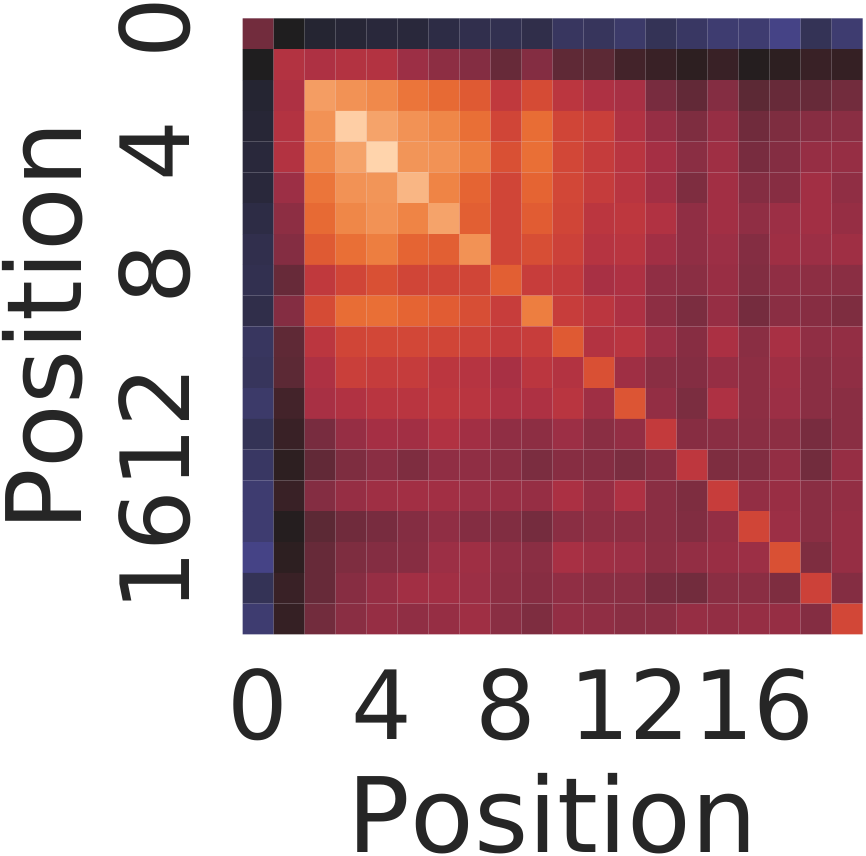}
    }
    \subfigure[LPE of len.\ 10 and 20.]{
    \label{fig:lpe_dot_two10}
    \includegraphics[width=0.43\linewidth]{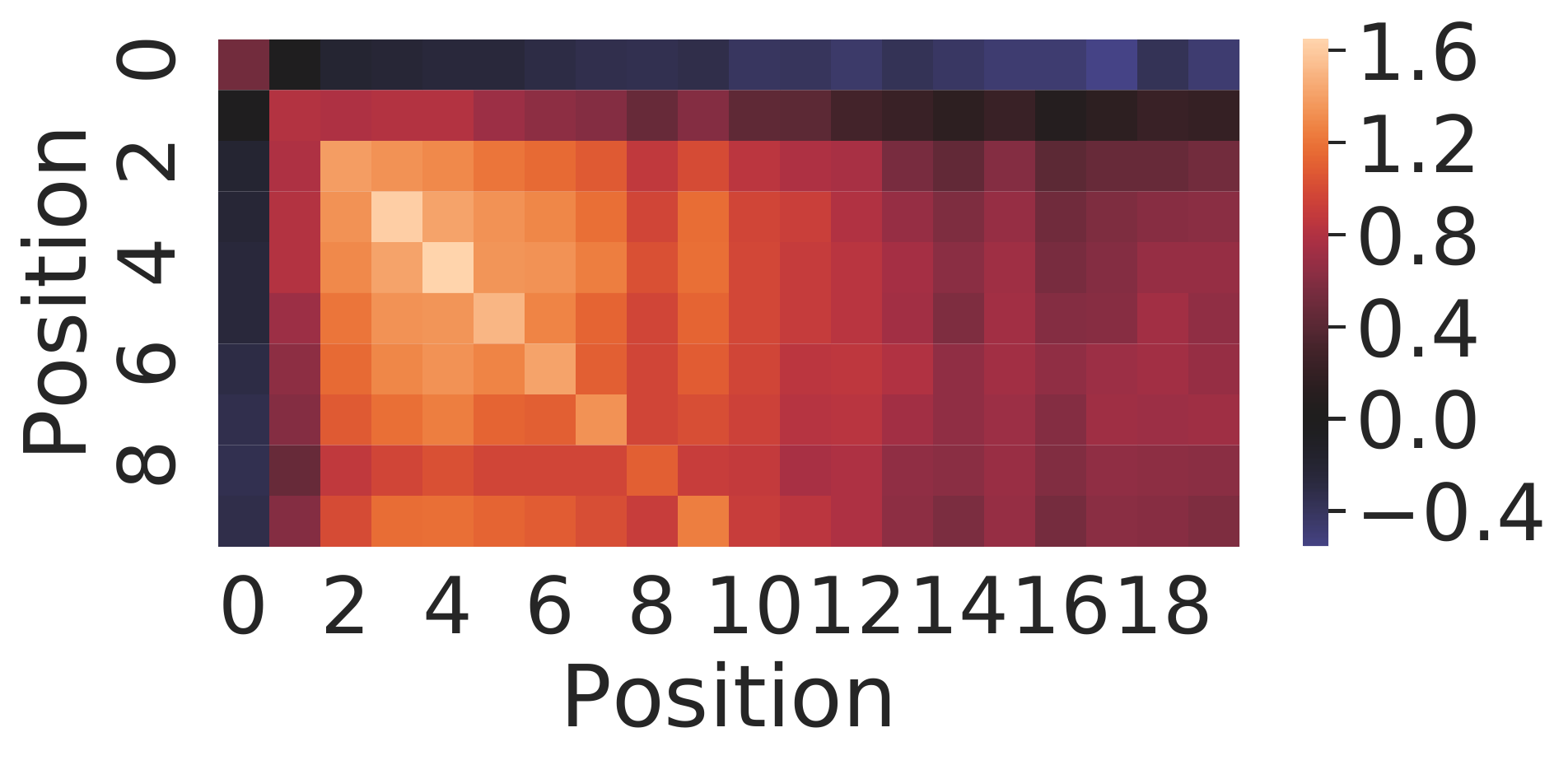}
    }
   \caption{Nearest neighbor heatmap of length 10 and 20 of LPE.}
\label{fig:lpe_dot_two}
\end{figure}

\begin{figure}[t]
    \centering
    \subfigure[LDPE of len.\ 10.]{
    \label{fig:ldpe_dot_10}
    \includegraphics[width=0.2\linewidth]{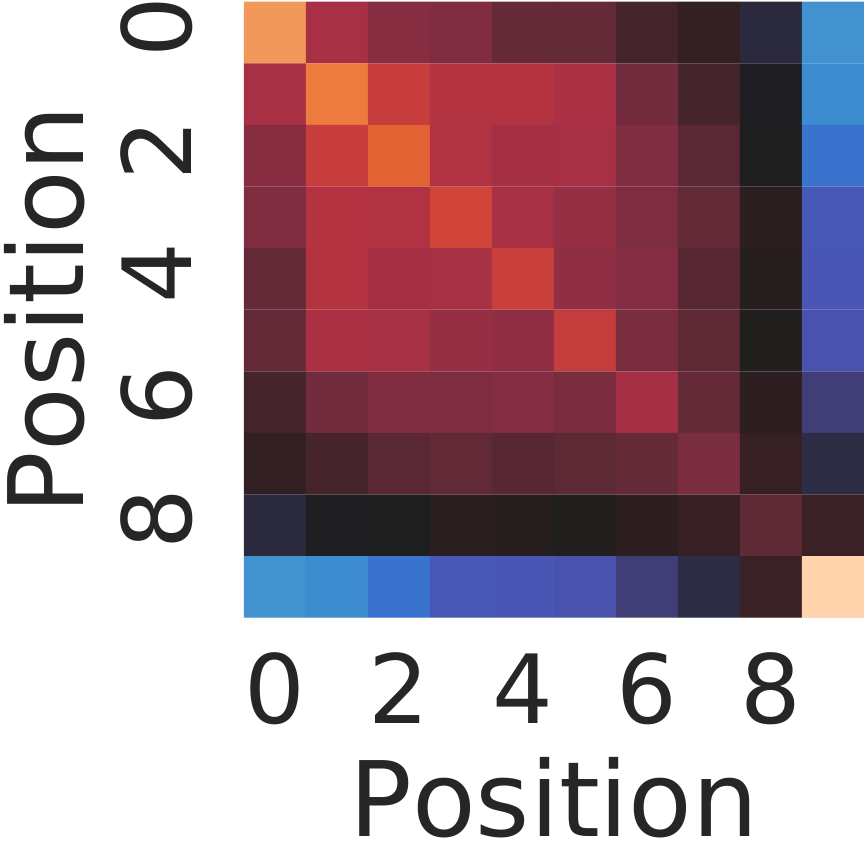}
    }
    \subfigure[LDPE of len.\ 20.]{
    \label{fig:ldpe_dot_20}
    \includegraphics[width=0.2\linewidth]{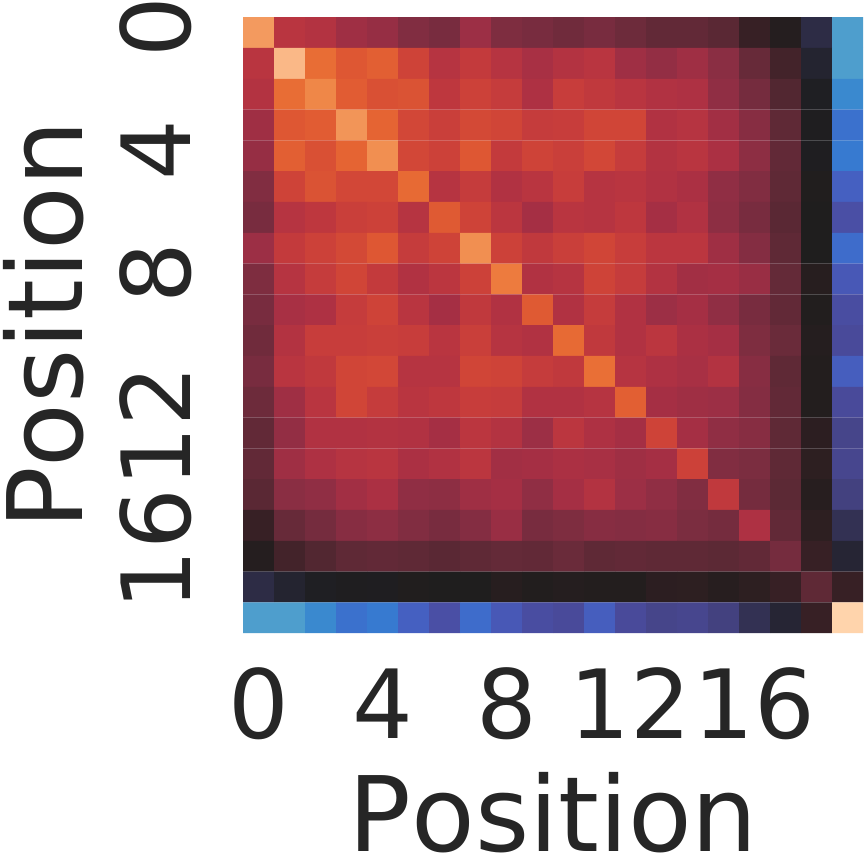}
    }
    \subfigure[LDPE of len.\ 10 and 20.]{
    \label{fig:ldpe_dot_two10}
    \includegraphics[width=0.43\linewidth]{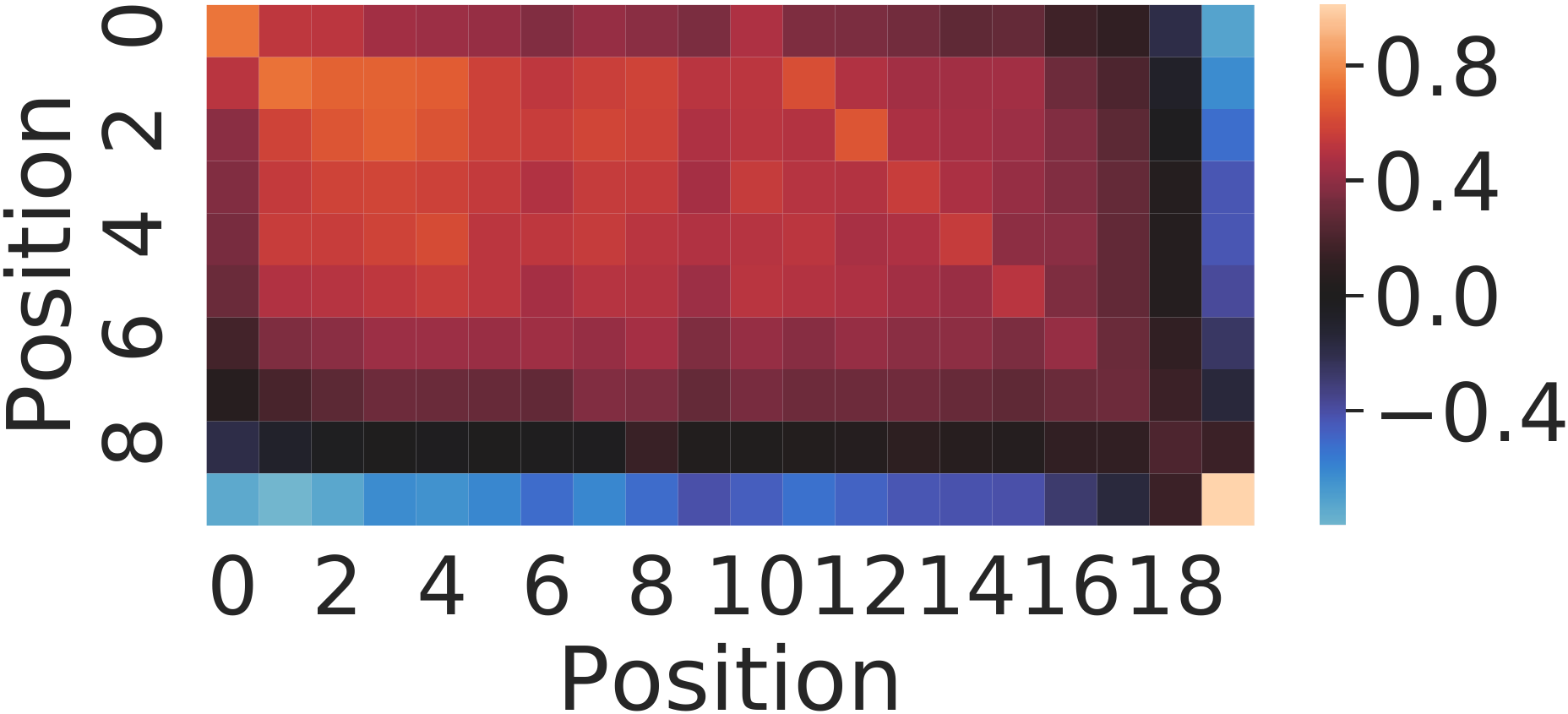}
    }
   \caption{Nearest neighbor heatmap of length 10 and 20 of LDPE.}
\label{fig:ldpe_dot_two}
\end{figure}

The nearest neighbor heatmaps shown in Fig.~\ref{fig:dot} represent the dot product of the positional encoding of the same session length. SPE and DPE of lengths 10 and 20 are demonstrated respectively. By definition, the heatmap of RSPE is the same as SPE. We can see that within the same length, these positional encoding schemes do have a strong correlation to the same position of the same length.

To examine the correlation of the positional encoding from different session lengths, we demonstrate the results from the dot product between lengths of 10 and 20 in Fig.~\ref{fig:dot_two}. For the case of SPE in Fig.~\ref{fig:spe_dot_two10},there is one diagonal strip that a strong correlation is only in between the same position but not the same reverse position. For example, there is no strong correlation between the last position of length 10 and length 20 (0.16 in the lower right corner). While for the case of DPE in Fig.~\ref{fig:dpe_dot_two10}, two diagonal strips indicate a strong correlations in both direction of positions.

From Fig.~\ref{fig:spe10} and~\ref{fig:spe20}, it is clear to see the \textit{forward-awareness} of SPE. For example, no matter what length is, the encoding for the first position will always be the same (the first row). While for the last position of lengths 10 and 20, as the max length and the embedding size vary, there will not be a shared same part, which is not \textit{backward-aware}. For RSPE, the case is just the reverse, i.e., the heat map would be upside down of SPE. From Fig.~\ref{fig:dpe10} and~\ref{fig:dpe20}, both \textit{forward-awareness} and \textit{backward-awareness} are demonstrated. The same positions and reverse positions will always share the same half of the embedding respectively.

The nearest neighbor heatmaps of LPE and LDPE are shown in Fig.~\ref{fig:lpe_dot_two} and Fig.~\ref{fig:ldpe_dot_two} respectively. From the heatmaps between the same length for 10 and 20 in Fig.~\ref{fig:lpe_dot_10} and Fig.~\ref{fig:lpe_dot_20}, the LPE focuses mostly on the initial intention with the most related encoding being the first position. From the heatmap between the length 10 and 20 in Fig.~\ref{fig:lpe_dot_two10}, it is obvious that LPE can only capture the \textit{forward} positional information of the initial intention of sessions with different lengths. The position encodings close to the end of sessions are not carrying \textit{forward} or \textit{backward} positional information, which is shown at the bottom right corner of Fig.~\ref{fig:lpe_dot_two10} with a similar closeness between different positions. Compared with LPE, the heatmaps  of LDPE are shown in Fig.~\ref{fig:ldpe_dot_10} and Fig.~\ref{fig:ldpe_dot_20}, from which the diagonal cells indicate that the LDPE can capture the positional information of distinct positions of the same length. While for Fig.~\ref{fig:ldpe_dot_two10}, it can be seen that the LDPE can capture the \textit{backward} positional information with the light cell at the bottom right corner. For the \textit{forward} positional information, the heatmap at the upper left corner shows a diagonal pattern. While the diagonal pattern is more obvious on the right part in Fig.~\ref{fig:ldpe_dot_two10}, which indicates that the \textit{backward} positional information is more important and useful for SBRS.

\subsection{Parameter Sensitivity}
\label{sec:exp-sen}

In this experiment, we evaluate the effect of $\lambda_0,\lambda_1$ and $\lambda_2$ in Eq.~\ref{eq:h-session}. We evaluate the relative scale between these three hyper-parameters. We fix $\lambda_0$ and $\lambda_1$ to 1 because they both represent the strength about the information from the last position. And we change $\lambda_2$ from $\{0.25,0.5,1,2\}$. In Fig.~\ref{fig:lambda2}, we demonstrate the result of the sensitivity of $\lambda_2$ is shown. For the \textit{Yoochoose} dataset in Fig.~\ref{fig:lambda2_yoo}, the best performance is achieved when $\lambda_2=0.5$. Compared with $\lambda_2=0.25$, we can see that inadequate initial intent will do harm to the overall performance. When it comes to $\lambda_2=1\text{ or }2$, we can also imply that the initial intent is not as important as the latest preference. For the \textit{Diginetica} dataset in Fig.~\ref{fig:lambda2_digi}, the best performance is achieved when $\lambda_2=1$. Compared with $\lambda_2=0.25\text{ or }0.5$, we can see that more initial intent will benefit the overall performance.

\begin{figure}[t]
    \centering
    \subfigure[\textit{Yoo.\ 1/64.}]{
    \label{fig:lambda2_yoo}
    \includegraphics[width=0.23\linewidth]{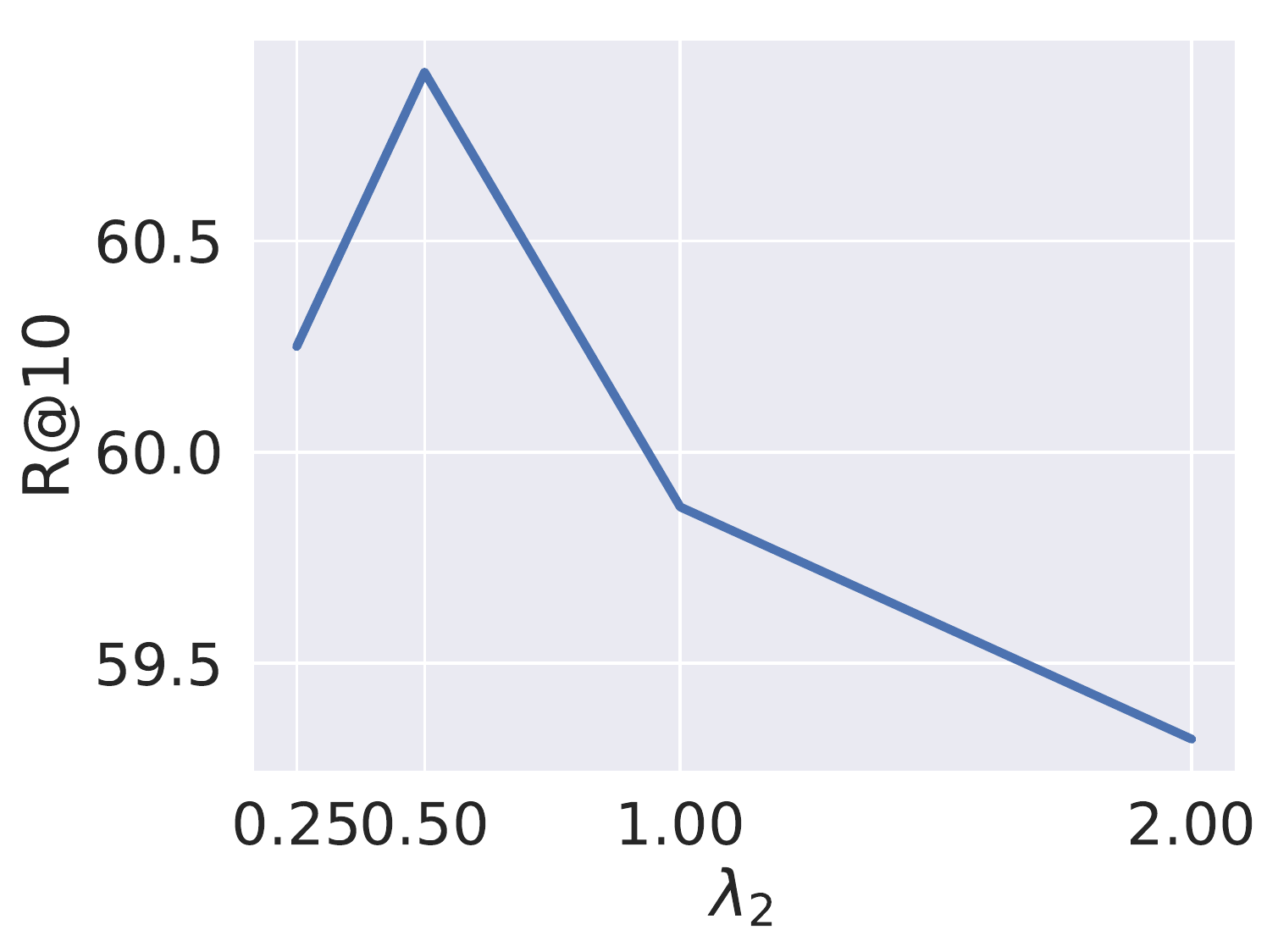}
    }
    \subfigure[\textit{Diginetica.}]{
    \label{fig:lambda2_digi}
    \includegraphics[width=0.23\linewidth]{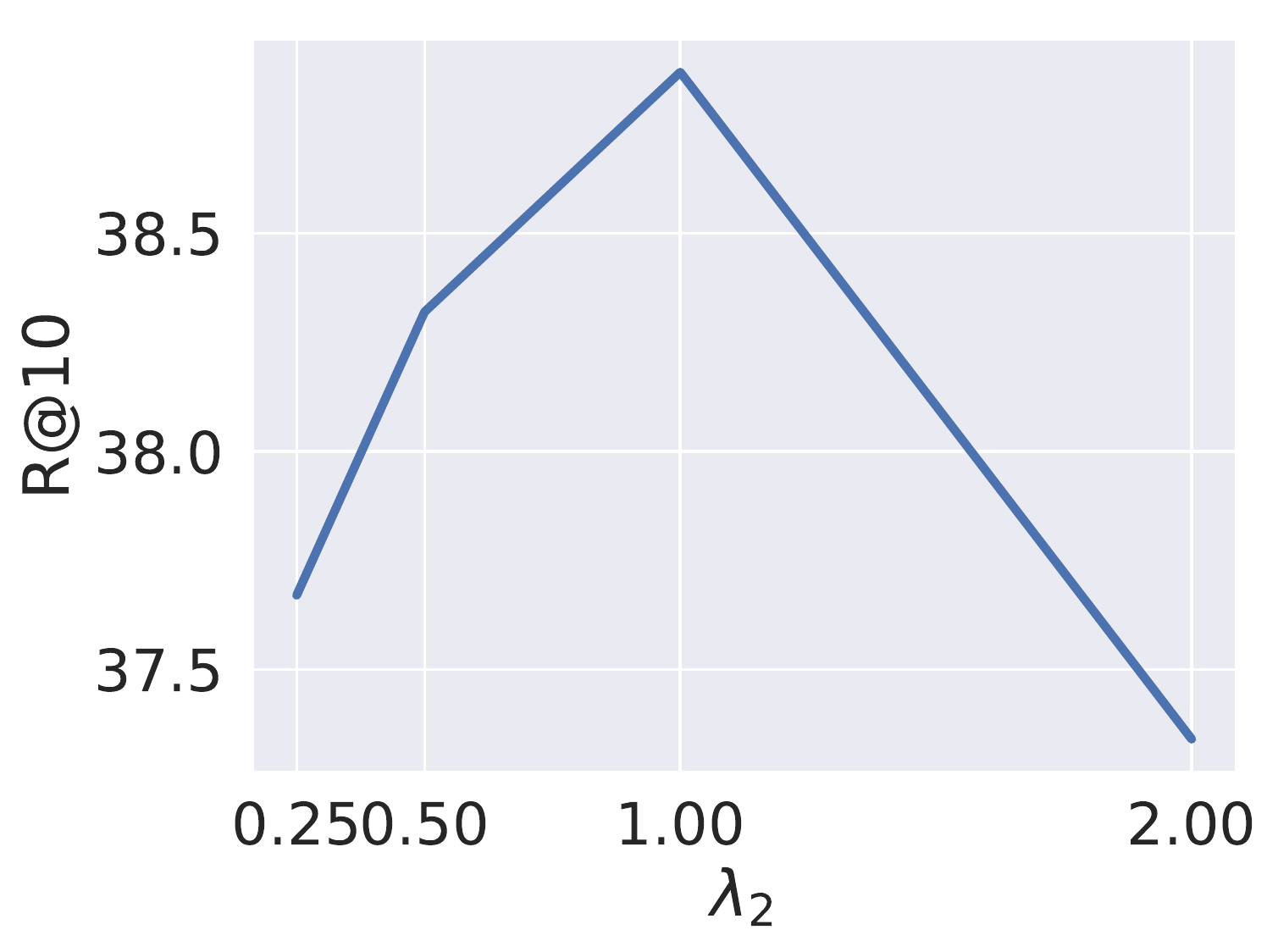}
    }
   \caption{Sensitivity of $\lambda_2$.}
\label{fig:lambda2}
\end{figure}

\section{Conclusion}
\label{sec:con}
In this paper, we investigate how the positional information can be exploited in the session-based recommendation. We find that there are two types of positional information, \textit{forward} and the \textit{backward} to represent the initial and the latest intentions in a session. A theoretical framework is proposed to analyze the representation ability of positional encoding schemes for the positional information in the session-based recommendation task. Specifically, the \textit{forward-awareness} and the \textit{backward-awareness} are defined for the evaluation of these schemes. Conventionally, a positional encoding is designed for language models with only the \textit{forward-awareness}. However, session-based recommendation requires both of the the \textit{forward-awareness} and the \textit{backward-awareness}. Therefore, a novel (learned) dual positional encoding scheme ((L)DPE) is proposed to fully capture both of the \textit{forward} and the \textit{backward} positional information. Besides, we design a PGGNN module to enhance the representation ability for graph neural networks to make use of the positional information. Combining both (L)DPE and the PGGNN, we build a PosRec model to perform the session-based recommendation with effective exploitation of the positional information. Extensive experiments are conducted on two benchmark datasets, which demonstrate that our proposal can achieve state-of-the-art performance and effectively exploit the positional information in SBRS.

\bibliographystyle{ACM-Reference-Format}
\bibliography{TOIS-2021-0053}


\begin{thebibliography}{55}


\ifx \showCODEN    \undefined \def \showCODEN     #1{\unskip}     \fi
\ifx \showDOI      \undefined \def \showDOI       #1{#1}\fi
\ifx \showISBNx    \undefined \def \showISBNx     #1{\unskip}     \fi
\ifx \showISBNxiii \undefined \def \showISBNxiii  #1{\unskip}     \fi
\ifx \showISSN     \undefined \def \showISSN      #1{\unskip}     \fi
\ifx \showLCCN     \undefined \def \showLCCN      #1{\unskip}     \fi
\ifx \shownote     \undefined \def \shownote      #1{#1}          \fi
\ifx \showarticletitle \undefined \def \showarticletitle #1{#1}   \fi
\ifx \showURL      \undefined \def \showURL       {\relax}        \fi
\providecommand\bibfield[2]{#2}
\providecommand\bibinfo[2]{#2}
\providecommand\natexlab[1]{#1}
\providecommand\showeprint[2][]{arXiv:#2}

\bibitem[\protect\citeauthoryear{Aliannejadi and Crestani}{Aliannejadi and
  Crestani}{2018}]%
        {tois-poi}
\bibfield{author}{\bibinfo{person}{Mohammad Aliannejadi} {and}
  \bibinfo{person}{Fabio Crestani}.} \bibinfo{year}{2018}\natexlab{}.
\newblock \showarticletitle{Personalized Context-Aware Point of Interest
  Recommendation}.
\newblock \bibinfo{journal}{\emph{{ACM} Trans. Inf. Syst.}}
  \bibinfo{volume}{36} (\bibinfo{year}{2018}).
\newblock


\bibitem[\protect\citeauthoryear{Bello, Zoph, Le, Vaswani, and Shlens}{Bello
  et~al\mbox{.}}{2019}]%
        {aacn}
\bibfield{author}{\bibinfo{person}{Irwan Bello}, \bibinfo{person}{Barret Zoph},
  \bibinfo{person}{Quoc Le}, \bibinfo{person}{Ashish Vaswani}, {and}
  \bibinfo{person}{Jonathon Shlens}.} \bibinfo{year}{2019}\natexlab{}.
\newblock \showarticletitle{Attention Augmented Convolutional Networks}. In
  \bibinfo{booktitle}{\emph{ICCV}}.
\newblock


\bibitem[\protect\citeauthoryear{Beutel, Covington, Jain, Xu, Li, Gatto, and
  Chi}{Beutel et~al\mbox{.}}{2018}]%
        {latent-cross}
\bibfield{author}{\bibinfo{person}{Alex Beutel}, \bibinfo{person}{Paul
  Covington}, \bibinfo{person}{Sagar Jain}, \bibinfo{person}{Can Xu},
  \bibinfo{person}{Jia Li}, \bibinfo{person}{Vince Gatto}, {and}
  \bibinfo{person}{Ed~H. Chi}.} \bibinfo{year}{2018}\natexlab{}.
\newblock \showarticletitle{Latent Cross: Making Use of Context in Recurrent
  Recommender Systems}. In \bibinfo{booktitle}{\emph{WSDM}}.
\newblock


\bibitem[\protect\citeauthoryear{Cao, Zhang, Song, Pan, and Xu}{Cao
  et~al\mbox{.}}{2020}]%
        {pos-sbrs}
\bibfield{author}{\bibinfo{person}{Yi Cao}, \bibinfo{person}{Weifeng Zhang},
  \bibinfo{person}{Bo Song}, \bibinfo{person}{Weike Pan}, {and}
  \bibinfo{person}{Congfu Xu}.} \bibinfo{year}{2020}\natexlab{}.
\newblock \showarticletitle{Position-aware context attention for session-based
  recommendation}.
\newblock \bibinfo{journal}{\emph{Neurocomputing}}  \bibinfo{volume}{376}
  (\bibinfo{year}{2020}).
\newblock


\bibitem[\protect\citeauthoryear{Carion, Massa, Synnaeve, Usunier, Kirillov,
  and Zagoruyko}{Carion et~al\mbox{.}}{2020}]%
        {e2eatt}
\bibfield{author}{\bibinfo{person}{Nicolas Carion}, \bibinfo{person}{Francisco
  Massa}, \bibinfo{person}{Gabriel Synnaeve}, \bibinfo{person}{Nicolas
  Usunier}, \bibinfo{person}{Alexander Kirillov}, {and} \bibinfo{person}{Sergey
  Zagoruyko}.} \bibinfo{year}{2020}\natexlab{}.
\newblock \showarticletitle{End-to-End Object Detection with Transformers}.
\newblock \bibinfo{journal}{\emph{CoRR}}  \bibinfo{volume}{abs/2005.12872}
  (\bibinfo{year}{2020}).
\newblock


\bibitem[\protect\citeauthoryear{Chen, Zhao, Li, Huang, and Ou}{Chen
  et~al\mbox{.}}{2019}]%
        {ali-seq}
\bibfield{author}{\bibinfo{person}{Qiwei Chen}, \bibinfo{person}{Huan Zhao},
  \bibinfo{person}{Wei Li}, \bibinfo{person}{Pipei Huang}, {and}
  \bibinfo{person}{Wenwu Ou}.} \bibinfo{year}{2019}\natexlab{}.
\newblock \showarticletitle{Behavior Sequence Transformer for E-commerce
  Recommendation in Alibaba}.
\newblock \bibinfo{journal}{\emph{CoRR}}  \bibinfo{volume}{abs/1905.06874}
  (\bibinfo{year}{2019}).
\newblock


\bibitem[\protect\citeauthoryear{Chen, Yin, Nguyen, Peng, Li, and Zhou}{Chen
  et~al\mbox{.}}{2020}]%
        {safm}
\bibfield{author}{\bibinfo{person}{Tong Chen}, \bibinfo{person}{Hongzhi Yin},
  \bibinfo{person}{Quoc Viet~Hung Nguyen}, \bibinfo{person}{Wen{-}Chih Peng},
  \bibinfo{person}{Xue Li}, {and} \bibinfo{person}{Xiaofang Zhou}.}
  \bibinfo{year}{2020}\natexlab{}.
\newblock \showarticletitle{Sequence-Aware Factorization Machines for Temporal
  Predictive Analytics}. In \bibinfo{booktitle}{\emph{ICDE}}.
\newblock


\bibitem[\protect\citeauthoryear{Chung, Gulcehre, Cho, and Bengio}{Chung
  et~al\mbox{.}}{2014}]%
        {gru}
\bibfield{author}{\bibinfo{person}{Junyoung Chung}, \bibinfo{person}{Caglar
  Gulcehre}, \bibinfo{person}{Kyunghyun Cho}, {and} \bibinfo{person}{Yoshua
  Bengio}.} \bibinfo{year}{2014}\natexlab{}.
\newblock \showarticletitle{Empirical evaluation of gated recurrent neural
  networks on sequence modeling}. In \bibinfo{booktitle}{\emph{NIPS}}.
\newblock


\bibitem[\protect\citeauthoryear{Dai, Yang, Yang, Carbonell, Le, and
  Salakhutdinov}{Dai et~al\mbox{.}}{2019}]%
        {trans-xl}
\bibfield{author}{\bibinfo{person}{Zihang Dai}, \bibinfo{person}{Zhilin Yang},
  \bibinfo{person}{Yiming Yang}, \bibinfo{person}{Jaime~G. Carbonell},
  \bibinfo{person}{Quoc~Viet Le}, {and} \bibinfo{person}{Ruslan
  Salakhutdinov}.} \bibinfo{year}{2019}\natexlab{}.
\newblock \showarticletitle{Transformer-XL: Attentive Language Models beyond a
  Fixed-Length Context}. In \bibinfo{booktitle}{\emph{ACL}}.
\newblock


\bibitem[\protect\citeauthoryear{Devlin, Chang, Lee, and Toutanova}{Devlin
  et~al\mbox{.}}{2019}]%
        {bert}
\bibfield{author}{\bibinfo{person}{Jacob Devlin}, \bibinfo{person}{Ming{-}Wei
  Chang}, \bibinfo{person}{Kenton Lee}, {and} \bibinfo{person}{Kristina
  Toutanova}.} \bibinfo{year}{2019}\natexlab{}.
\newblock \showarticletitle{{BERT:} Pre-training of Deep Bidirectional
  Transformers for Language Understanding}. In
  \bibinfo{booktitle}{\emph{NAACL-HLT}}.
\newblock


\bibitem[\protect\citeauthoryear{Fang, Zhang, Shu, and Guo}{Fang
  et~al\mbox{.}}{2020}]%
        {tois-seq}
\bibfield{author}{\bibinfo{person}{Hui Fang}, \bibinfo{person}{Danning Zhang},
  \bibinfo{person}{Yiheng Shu}, {and} \bibinfo{person}{Guibing Guo}.}
  \bibinfo{year}{2020}\natexlab{}.
\newblock \showarticletitle{Deep Learning for Sequential Recommendation:
  Algorithms, Influential Factors, and Evaluations}.
\newblock \bibinfo{journal}{\emph{{ACM} Trans. Inf. Syst.}}
  \bibinfo{volume}{39} (\bibinfo{year}{2020}).
\newblock


\bibitem[\protect\citeauthoryear{Guo, Yin, Wang, Chen, Zhou, and Hung}{Guo
  et~al\mbox{.}}{2019}]%
        {ssrm}
\bibfield{author}{\bibinfo{person}{Lei Guo}, \bibinfo{person}{Hongzhi Yin},
  \bibinfo{person}{Qinyong Wang}, \bibinfo{person}{Tong Chen},
  \bibinfo{person}{Alexander Zhou}, {and} \bibinfo{person}{Nguyen Quoc~Viet
  Hung}.} \bibinfo{year}{2019}\natexlab{}.
\newblock \showarticletitle{Streaming Session-based Recommendation}. In
  \bibinfo{booktitle}{\emph{SIGKDD}}.
\newblock


\bibitem[\protect\citeauthoryear{Hidasi and Karatzoglou}{Hidasi and
  Karatzoglou}{2018}]%
        {gru4rec+}
\bibfield{author}{\bibinfo{person}{Bal{\'{a}}zs Hidasi} {and}
  \bibinfo{person}{Alexandros Karatzoglou}.} \bibinfo{year}{2018}\natexlab{}.
\newblock \showarticletitle{Recurrent Neural Networks with Top-k Gains for
  Session-based Recommendations}. In \bibinfo{booktitle}{\emph{CIKM}}.
\newblock


\bibitem[\protect\citeauthoryear{Hidasi, Karatzoglou, Baltrunas, and
  Tikk}{Hidasi et~al\mbox{.}}{2016}]%
        {gru4rec}
\bibfield{author}{\bibinfo{person}{Bal{\'{a}}zs Hidasi},
  \bibinfo{person}{Alexandros Karatzoglou}, \bibinfo{person}{Linas Baltrunas},
  {and} \bibinfo{person}{Domonkos Tikk}.} \bibinfo{year}{2016}\natexlab{}.
\newblock \showarticletitle{Session-based Recommendations with Recurrent Neural
  Networks}. In \bibinfo{booktitle}{\emph{ICLR}}.
\newblock


\bibitem[\protect\citeauthoryear{Hidasi and Tikk}{Hidasi and Tikk}{2016}]%
        {sum}
\bibfield{author}{\bibinfo{person}{Bal{\'{a}}zs Hidasi} {and}
  \bibinfo{person}{Domonkos Tikk}.} \bibinfo{year}{2016}\natexlab{}.
\newblock \showarticletitle{General factorization framework for context-aware
  recommendations}.
\newblock \bibinfo{journal}{\emph{Data Min. Knowl. Discov.}}
  \bibinfo{volume}{30} (\bibinfo{year}{2016}).
\newblock


\bibitem[\protect\citeauthoryear{Hochreiter and Schmidhuber}{Hochreiter and
  Schmidhuber}{1997}]%
        {lstm}
\bibfield{author}{\bibinfo{person}{Sepp Hochreiter} {and}
  \bibinfo{person}{J{\"{u}}rgen Schmidhuber}.} \bibinfo{year}{1997}\natexlab{}.
\newblock \showarticletitle{Long Short-Term Memory}.
\newblock \bibinfo{journal}{\emph{Neural Computation}}  \bibinfo{volume}{9}
  (\bibinfo{year}{1997}).
\newblock


\bibitem[\protect\citeauthoryear{Kang and McAuley}{Kang and McAuley}{2018}]%
        {sasrec}
\bibfield{author}{\bibinfo{person}{Wang{-}Cheng Kang} {and}
  \bibinfo{person}{Julian~J. McAuley}.} \bibinfo{year}{2018}\natexlab{}.
\newblock \showarticletitle{Self-Attentive Sequential Recommendation}. In
  \bibinfo{booktitle}{\emph{ICDM}}.
\newblock


\bibitem[\protect\citeauthoryear{Kingma and Ba}{Kingma and Ba}{2015}]%
        {adam}
\bibfield{author}{\bibinfo{person}{Diederik~P. Kingma} {and}
  \bibinfo{person}{Jimmy Ba}.} \bibinfo{year}{2015}\natexlab{}.
\newblock \showarticletitle{Adam: {A} Method for Stochastic Optimization}. In
  \bibinfo{booktitle}{\emph{ICLR}}, \bibfield{editor}{\bibinfo{person}{Yoshua
  Bengio} {and} \bibinfo{person}{Yann LeCun}} (Eds.).
\newblock


\bibitem[\protect\citeauthoryear{Kipf and Welling}{Kipf and Welling}{2017}]%
        {gcn}
\bibfield{author}{\bibinfo{person}{Thomas~N. Kipf} {and} \bibinfo{person}{Max
  Welling}.} \bibinfo{year}{2017}\natexlab{}.
\newblock \showarticletitle{Semi-Supervised Classification with Graph
  Convolutional Networks}. In \bibinfo{booktitle}{\emph{ICLR}}.
\newblock


\bibitem[\protect\citeauthoryear{Krichene and Rendle}{Krichene and
  Rendle}{2020}]%
        {metric}
\bibfield{author}{\bibinfo{person}{Walid Krichene} {and}
  \bibinfo{person}{Steffen Rendle}.} \bibinfo{year}{2020}\natexlab{}.
\newblock \showarticletitle{On Sampled Metrics for Item Recommendation}. In
  \bibinfo{booktitle}{\emph{SIGKDD}}.
\newblock


\bibitem[\protect\citeauthoryear{Lee, Park, Baek, Oh, Kim, and Lee}{Lee
  et~al\mbox{.}}{2019}]%
        {2dsa}
\bibfield{author}{\bibinfo{person}{Junyeop Lee}, \bibinfo{person}{Sungrae
  Park}, \bibinfo{person}{Jeonghun Baek}, \bibinfo{person}{Seong~Joon Oh},
  \bibinfo{person}{Seonghyeon Kim}, {and} \bibinfo{person}{Hwalsuk Lee}.}
  \bibinfo{year}{2019}\natexlab{}.
\newblock \showarticletitle{On Recognizing Texts of Arbitrary Shapes with 2D
  Self-Attention}.
\newblock \bibinfo{journal}{\emph{CoRR}}  \bibinfo{volume}{abs/1910.04396}
  (\bibinfo{year}{2019}).
\newblock


\bibitem[\protect\citeauthoryear{Li, Ren, Chen, Ren, Lian, and Ma}{Li
  et~al\mbox{.}}{2017}]%
        {narm}
\bibfield{author}{\bibinfo{person}{Jing Li}, \bibinfo{person}{Pengjie Ren},
  \bibinfo{person}{Zhumin Chen}, \bibinfo{person}{Zhaochun Ren},
  \bibinfo{person}{Tao Lian}, {and} \bibinfo{person}{Jun Ma}.}
  \bibinfo{year}{2017}\natexlab{}.
\newblock \showarticletitle{Neural Attentive Session-based Recommendation}. In
  \bibinfo{booktitle}{\emph{CIKM}}.
\newblock


\bibitem[\protect\citeauthoryear{Li, Chen, Luo, Yin, and Huang}{Li
  et~al\mbox{.}}{2021}]%
        {LiDCPOI21}
\bibfield{author}{\bibinfo{person}{Yang Li}, \bibinfo{person}{Tong Chen},
  \bibinfo{person}{Yadan Luo}, \bibinfo{person}{Hongzhi Yin}, {and}
  \bibinfo{person}{Zi Huang}.} \bibinfo{year}{2021}\natexlab{}.
\newblock \showarticletitle{Discovering Collaborative Signals for Next POI
  Recommendation with Iterative Seq2Graph Augmentation}.
\newblock \bibinfo{journal}{\emph{CoRR}}  \bibinfo{volume}{abs/2106.15814}
  (\bibinfo{year}{2021}).
\newblock


\bibitem[\protect\citeauthoryear{Li, Luo, Zhang, Sadiq, and Cui}{Li
  et~al\mbox{.}}{2019}]%
        {LiLZSC19}
\bibfield{author}{\bibinfo{person}{Yang Li}, \bibinfo{person}{Yadan Luo},
  \bibinfo{person}{Zheng Zhang}, \bibinfo{person}{Shazia~W. Sadiq}, {and}
  \bibinfo{person}{Peng Cui}.} \bibinfo{year}{2019}\natexlab{}.
\newblock \showarticletitle{Context-Aware Attention-Based Data Augmentation for
  {POI} Recommendation}. In \bibinfo{booktitle}{\emph{{ICDE} Workshops}}.
\newblock


\bibitem[\protect\citeauthoryear{Li, Tarlow, Brockschmidt, and Zemel}{Li
  et~al\mbox{.}}{2016}]%
        {ggnn}
\bibfield{author}{\bibinfo{person}{Yujia Li}, \bibinfo{person}{Daniel Tarlow},
  \bibinfo{person}{Marc Brockschmidt}, {and} \bibinfo{person}{Richard~S.
  Zemel}.} \bibinfo{year}{2016}\natexlab{}.
\newblock \showarticletitle{Gated Graph Sequence Neural Networks}. In
  \bibinfo{booktitle}{\emph{ICLR}}.
\newblock


\bibitem[\protect\citeauthoryear{Liu, Zeng, Mokhosi, and Zhang}{Liu
  et~al\mbox{.}}{2018}]%
        {Liu18STAMP}
\bibfield{author}{\bibinfo{person}{Qiao Liu}, \bibinfo{person}{Yifu Zeng},
  \bibinfo{person}{Refuoe Mokhosi}, {and} \bibinfo{person}{Haibin Zhang}.}
  \bibinfo{year}{2018}\natexlab{}.
\newblock \showarticletitle{{STAMP:} Short-Term Attention/Memory Priority Model
  for Session-based Recommendation}. In \bibinfo{booktitle}{\emph{SIGKDD}}.
\newblock


\bibitem[\protect\citeauthoryear{Parmar, Ramachandran, Vaswani, Bello,
  Levskaya, and Shlens}{Parmar et~al\mbox{.}}{2019}]%
        {sasa}
\bibfield{author}{\bibinfo{person}{Niki Parmar}, \bibinfo{person}{Prajit
  Ramachandran}, \bibinfo{person}{Ashish Vaswani}, \bibinfo{person}{Irwan
  Bello}, \bibinfo{person}{Anselm Levskaya}, {and} \bibinfo{person}{Jon
  Shlens}.} \bibinfo{year}{2019}\natexlab{}.
\newblock \showarticletitle{Stand-Alone Self-Attention in Vision Models}. In
  \bibinfo{booktitle}{\emph{NeurIPS}}.
\newblock


\bibitem[\protect\citeauthoryear{Qian, Liu, Nguyen, and Yin}{Qian
  et~al\mbox{.}}{2019}]%
        {tois-poi1}
\bibfield{author}{\bibinfo{person}{Tieyun Qian}, \bibinfo{person}{Bei Liu},
  \bibinfo{person}{Quoc Viet~Hung Nguyen}, {and} \bibinfo{person}{Hongzhi
  Yin}.} \bibinfo{year}{2019}\natexlab{}.
\newblock \showarticletitle{Spatiotemporal Representation Learning for
  Translation-Based {POI} Recommendation}.
\newblock \bibinfo{journal}{\emph{{ACM} Trans. Inf. Syst.}}
  \bibinfo{volume}{37} (\bibinfo{year}{2019}).
\newblock


\bibitem[\protect\citeauthoryear{Qiu, Huang, Li, and Yin}{Qiu
  et~al\mbox{.}}{2020a}]%
        {fgnnj}
\bibfield{author}{\bibinfo{person}{Ruihong Qiu}, \bibinfo{person}{Zi Huang},
  \bibinfo{person}{Jingjing Li}, {and} \bibinfo{person}{Hongzhi Yin}.}
  \bibinfo{year}{2020}\natexlab{a}.
\newblock \showarticletitle{Exploiting Cross-session Information for
  Session-based Recommendation with Graph Neural Networks}.
\newblock \bibinfo{journal}{\emph{{ACM} Trans. Inf. Syst.}}
  \bibinfo{volume}{38} (\bibinfo{year}{2020}).
\newblock


\bibitem[\protect\citeauthoryear{Qiu, Li, Huang, and Yin}{Qiu
  et~al\mbox{.}}{2019}]%
        {fgnn}
\bibfield{author}{\bibinfo{person}{Ruihong Qiu}, \bibinfo{person}{Jingjing Li},
  \bibinfo{person}{Zi Huang}, {and} \bibinfo{person}{Hongzhi Yin}.}
  \bibinfo{year}{2019}\natexlab{}.
\newblock \showarticletitle{Rethinking the Item Order in Session-based
  Recommendation with Graph Neural Networks}. In
  \bibinfo{booktitle}{\emph{CIKM}}.
\newblock


\bibitem[\protect\citeauthoryear{Qiu, Yin, Huang, and Chen}{Qiu
  et~al\mbox{.}}{2020b}]%
        {gag}
\bibfield{author}{\bibinfo{person}{Ruihong Qiu}, \bibinfo{person}{Hongzhi Yin},
  \bibinfo{person}{Zi Huang}, {and} \bibinfo{person}{Tong Chen}.}
  \bibinfo{year}{2020}\natexlab{b}.
\newblock \showarticletitle{{GAG:} Global Attributed Graph Neural Network for
  Streaming Session-based Recommendation}. In
  \bibinfo{booktitle}{\emph{SIGIR}}.
\newblock


\bibitem[\protect\citeauthoryear{Rendle, Freudenthaler, Gantner, and
  Schmidt{-}Thieme}{Rendle et~al\mbox{.}}{2009}]%
        {bprmf}
\bibfield{author}{\bibinfo{person}{Steffen Rendle}, \bibinfo{person}{Christoph
  Freudenthaler}, \bibinfo{person}{Zeno Gantner}, {and} \bibinfo{person}{Lars
  Schmidt{-}Thieme}.} \bibinfo{year}{2009}\natexlab{}.
\newblock \showarticletitle{{BPR:} Bayesian Personalized Ranking from Implicit
  Feedback}. In \bibinfo{booktitle}{\emph{UAI}}.
\newblock


\bibitem[\protect\citeauthoryear{Rendle, Freudenthaler, and
  Schmidt{-}Thieme}{Rendle et~al\mbox{.}}{2010}]%
        {fpmc}
\bibfield{author}{\bibinfo{person}{Steffen Rendle}, \bibinfo{person}{Christoph
  Freudenthaler}, {and} \bibinfo{person}{Lars Schmidt{-}Thieme}.}
  \bibinfo{year}{2010}\natexlab{}.
\newblock \showarticletitle{Factorizing personalized Markov chains for
  next-basket recommendation}. In \bibinfo{booktitle}{\emph{WWW}}.
\newblock


\bibitem[\protect\citeauthoryear{Sarwar, Karypis, Konstan, and Riedl}{Sarwar
  et~al\mbox{.}}{2001}]%
        {item-knn}
\bibfield{author}{\bibinfo{person}{Badrul~Munir Sarwar},
  \bibinfo{person}{George Karypis}, \bibinfo{person}{Joseph~A. Konstan}, {and}
  \bibinfo{person}{John Riedl}.} \bibinfo{year}{2001}\natexlab{}.
\newblock \showarticletitle{Item-based collaborative filtering recommendation
  algorithms}. In \bibinfo{booktitle}{\emph{WWW}}.
\newblock


\bibitem[\protect\citeauthoryear{Shani, Heckerman, and Brafman}{Shani
  et~al\mbox{.}}{2005}]%
        {mdp}
\bibfield{author}{\bibinfo{person}{Guy Shani}, \bibinfo{person}{David
  Heckerman}, {and} \bibinfo{person}{Ronen~I. Brafman}.}
  \bibinfo{year}{2005}\natexlab{}.
\newblock \showarticletitle{An MDP-Based Recommender System}.
\newblock \bibinfo{journal}{\emph{J. Mach. Learn. Res.}}  \bibinfo{volume}{6}
  (\bibinfo{year}{2005}).
\newblock


\bibitem[\protect\citeauthoryear{Shaw, Uszkoreit, and Vaswani}{Shaw
  et~al\mbox{.}}{2018}]%
        {rpe}
\bibfield{author}{\bibinfo{person}{Peter Shaw}, \bibinfo{person}{Jakob
  Uszkoreit}, {and} \bibinfo{person}{Ashish Vaswani}.}
  \bibinfo{year}{2018}\natexlab{}.
\newblock \showarticletitle{Self-Attention with Relative Position
  Representations}. In \bibinfo{booktitle}{\emph{NAACL}}.
\newblock


\bibitem[\protect\citeauthoryear{Shiv and Quirk}{Shiv and Quirk}{2019}]%
        {tree}
\bibfield{author}{\bibinfo{person}{Vighnesh~Leonardo Shiv} {and}
  \bibinfo{person}{Chris Quirk}.} \bibinfo{year}{2019}\natexlab{}.
\newblock \showarticletitle{Novel positional encodings to enable tree-based
  transformers}. In \bibinfo{booktitle}{\emph{NeurIPS}}.
\newblock


\bibitem[\protect\citeauthoryear{Sun, Liu, Wu, Pei, Lin, Ou, and Jiang}{Sun
  et~al\mbox{.}}{2019a}]%
        {bert4rec}
\bibfield{author}{\bibinfo{person}{Fei Sun}, \bibinfo{person}{Jun Liu},
  \bibinfo{person}{Jian Wu}, \bibinfo{person}{Changhua Pei},
  \bibinfo{person}{Xiao Lin}, \bibinfo{person}{Wenwu Ou}, {and}
  \bibinfo{person}{Peng Jiang}.} \bibinfo{year}{2019}\natexlab{a}.
\newblock \showarticletitle{BERT4Rec: Sequential Recommendation with
  Bidirectional Encoder Representations from Transformer}. In
  \bibinfo{booktitle}{\emph{CIKM}}.
\newblock


\bibitem[\protect\citeauthoryear{Sun, Qian, Yin, Chen, Chen, and Chen}{Sun
  et~al\mbox{.}}{2019b}]%
        {history}
\bibfield{author}{\bibinfo{person}{Ke Sun}, \bibinfo{person}{Tieyun Qian},
  \bibinfo{person}{Hongzhi Yin}, \bibinfo{person}{Tong Chen},
  \bibinfo{person}{Yiqi Chen}, {and} \bibinfo{person}{Ling Chen}.}
  \bibinfo{year}{2019}\natexlab{b}.
\newblock \showarticletitle{What Can History Tell Us?}. In
  \bibinfo{booktitle}{\emph{CIKM}}.
\newblock


\bibitem[\protect\citeauthoryear{Tan, Xu, and Liu}{Tan et~al\mbox{.}}{2016}]%
        {improved}
\bibfield{author}{\bibinfo{person}{Yong~Kiam Tan}, \bibinfo{person}{Xinxing
  Xu}, {and} \bibinfo{person}{Yong Liu}.} \bibinfo{year}{2016}\natexlab{}.
\newblock \showarticletitle{Improved Recurrent Neural Networks for
  Session-based Recommendations}. In \bibinfo{booktitle}{\emph{DLRS@RecSys}}.
\newblock


\bibitem[\protect\citeauthoryear{Tang and Wang}{Tang and Wang}{2018}]%
        {caser}
\bibfield{author}{\bibinfo{person}{Jiaxi Tang} {and} \bibinfo{person}{Ke
  Wang}.} \bibinfo{year}{2018}\natexlab{}.
\newblock \showarticletitle{Personalized Top-N Sequential Recommendation via
  Convolutional Sequence Embedding}. In \bibinfo{booktitle}{\emph{WSDM}}.
\newblock


\bibitem[\protect\citeauthoryear{Vaswani, Shazeer, Parmar, Uszkoreit, Jones,
  Gomez, Kaiser, and Polosukhin}{Vaswani et~al\mbox{.}}{2017}]%
        {attention}
\bibfield{author}{\bibinfo{person}{Ashish Vaswani}, \bibinfo{person}{Noam
  Shazeer}, \bibinfo{person}{Niki Parmar}, \bibinfo{person}{Jakob Uszkoreit},
  \bibinfo{person}{Llion Jones}, \bibinfo{person}{Aidan~N. Gomez},
  \bibinfo{person}{Lukasz Kaiser}, {and} \bibinfo{person}{Illia Polosukhin}.}
  \bibinfo{year}{2017}\natexlab{}.
\newblock \showarticletitle{Attention is All you Need}. In
  \bibinfo{booktitle}{\emph{NIPS}}.
\newblock


\bibitem[\protect\citeauthoryear{Velickovic, Cucurull, Casanova, Romero,
  Li{\`{o}}, and Bengio}{Velickovic et~al\mbox{.}}{2018}]%
        {gat}
\bibfield{author}{\bibinfo{person}{Petar Velickovic}, \bibinfo{person}{Guillem
  Cucurull}, \bibinfo{person}{Arantxa Casanova}, \bibinfo{person}{Adriana
  Romero}, \bibinfo{person}{Pietro Li{\`{o}}}, {and} \bibinfo{person}{Yoshua
  Bengio}.} \bibinfo{year}{2018}\natexlab{}.
\newblock \showarticletitle{Graph Attention Networks}. In
  \bibinfo{booktitle}{\emph{ICLR}}.
\newblock


\bibitem[\protect\citeauthoryear{Vinyals, Bengio, and Kudlur}{Vinyals
  et~al\mbox{.}}{2016}]%
        {set2set}
\bibfield{author}{\bibinfo{person}{Oriol Vinyals}, \bibinfo{person}{Samy
  Bengio}, {and} \bibinfo{person}{Manjunath Kudlur}.}
  \bibinfo{year}{2016}\natexlab{}.
\newblock \showarticletitle{Order Matters: Sequence to sequence for sets}. In
  \bibinfo{booktitle}{\emph{ICLR}}.
\newblock


\bibitem[\protect\citeauthoryear{Wang, Shang, Lioma, Jiang, Yang, Liu, and
  Simonsen}{Wang et~al\mbox{.}}{2021b}]%
        {pos-iclr}
\bibfield{author}{\bibinfo{person}{Benyou Wang}, \bibinfo{person}{Lifeng
  Shang}, \bibinfo{person}{Christina Lioma}, \bibinfo{person}{Xin Jiang},
  \bibinfo{person}{Hao Yang}, \bibinfo{person}{Qun Liu}, {and}
  \bibinfo{person}{Jakob~Grue Simonsen}.} \bibinfo{year}{2021}\natexlab{b}.
\newblock \showarticletitle{On Position Embeddings in BERT}. In
  \bibinfo{booktitle}{\emph{ICLR}}.
\newblock


\bibitem[\protect\citeauthoryear{Wang, Ma, Zhang, Chen, Liu, and Ma}{Wang
  et~al\mbox{.}}{2021a}]%
        {tois-seq1}
\bibfield{author}{\bibinfo{person}{Chenyang Wang}, \bibinfo{person}{Weizhi Ma},
  \bibinfo{person}{Min Zhang}, \bibinfo{person}{Chong Chen},
  \bibinfo{person}{Yiqun Liu}, {and} \bibinfo{person}{Shaoping Ma}.}
  \bibinfo{year}{2021}\natexlab{a}.
\newblock \showarticletitle{Toward Dynamic User Intention: Temporal
  Evolutionary Effects of Item Relations in Sequential Recommendation}.
\newblock \bibinfo{journal}{\emph{ACM Trans. Inf. Syst.}}  \bibinfo{volume}{39}
  (\bibinfo{year}{2021}).
\newblock


\bibitem[\protect\citeauthoryear{Wang, Zhang, Liu, Liu, Zhang, Lin, and
  Zha}{Wang et~al\mbox{.}}{2020}]%
        {bc}
\bibfield{author}{\bibinfo{person}{Wen Wang}, \bibinfo{person}{Wei Zhang},
  \bibinfo{person}{Shukai Liu}, \bibinfo{person}{Qi Liu}, \bibinfo{person}{Bo
  Zhang}, \bibinfo{person}{Leyu Lin}, {and} \bibinfo{person}{Hongyuan Zha}.}
  \bibinfo{year}{2020}\natexlab{}.
\newblock \showarticletitle{Beyond Clicks: Modeling Multi-Relational Item Graph
  for Session-Based Target Behavior Prediction}. In
  \bibinfo{booktitle}{\emph{WWW}}.
\newblock


\bibitem[\protect\citeauthoryear{Wang and Liu}{Wang and Liu}{2019}]%
        {2dlatex}
\bibfield{author}{\bibinfo{person}{Zelun Wang} {and}
  \bibinfo{person}{Jyh{-}Charn Liu}.} \bibinfo{year}{2019}\natexlab{}.
\newblock \showarticletitle{Translating Mathematical Formula Images to LaTeX
  Sequences Using Deep Neural Networks with Sequence-level Training}.
\newblock \bibinfo{journal}{\emph{CoRR}}  \bibinfo{volume}{abs/1908.11415}
  (\bibinfo{year}{2019}).
\newblock


\bibitem[\protect\citeauthoryear{Wu, Tang, Zhu, Wang, Xie, and Tan}{Wu
  et~al\mbox{.}}{2019}]%
        {srgnn}
\bibfield{author}{\bibinfo{person}{Shu Wu}, \bibinfo{person}{Yuyuan Tang},
  \bibinfo{person}{Yanqiao Zhu}, \bibinfo{person}{Liang Wang},
  \bibinfo{person}{Xing Xie}, {and} \bibinfo{person}{Tieniu Tan}.}
  \bibinfo{year}{2019}\natexlab{}.
\newblock \showarticletitle{Session-based Recommendation with Graph Neural
  Networks}. In \bibinfo{booktitle}{\emph{AAAI}}.
\newblock


\bibitem[\protect\citeauthoryear{Xia, Yin, Yu, Wang, Cui, and Zhang}{Xia
  et~al\mbox{.}}{2021}]%
        {hypersbrs}
\bibfield{author}{\bibinfo{person}{Xin Xia}, \bibinfo{person}{Hongzhi Yin},
  \bibinfo{person}{Junliang Yu}, \bibinfo{person}{Qinyong Wang},
  \bibinfo{person}{Lizhen Cui}, {and} \bibinfo{person}{Xiangliang Zhang}.}
  \bibinfo{year}{2021}\natexlab{}.
\newblock \showarticletitle{Self-Supervised Hypergraph Convolutional Networks
  for Session-based Recommendation}. In \bibinfo{booktitle}{\emph{AAAI}}.
\newblock


\bibitem[\protect\citeauthoryear{Xu, Zhao, Liu, Sheng, Xu, Zhuang, Fang, and
  Zhou}{Xu et~al\mbox{.}}{2019}]%
        {gc-san}
\bibfield{author}{\bibinfo{person}{Chengfeng Xu}, \bibinfo{person}{Pengpeng
  Zhao}, \bibinfo{person}{Yanchi Liu}, \bibinfo{person}{Victor~S. Sheng},
  \bibinfo{person}{Jiajie Xu}, \bibinfo{person}{Fuzhen Zhuang},
  \bibinfo{person}{Junhua Fang}, {and} \bibinfo{person}{Xiaofang Zhou}.}
  \bibinfo{year}{2019}\natexlab{}.
\newblock \showarticletitle{Graph Contextualized Self-Attention Network for
  Session-based Recommendation}. In \bibinfo{booktitle}{\emph{IJCAI}}.
\newblock


\bibitem[\protect\citeauthoryear{Yang, Dai, Yang, Carbonell, Salakhutdinov, and
  Le}{Yang et~al\mbox{.}}{2019}]%
        {xlnet}
\bibfield{author}{\bibinfo{person}{Zhilin Yang}, \bibinfo{person}{Zihang Dai},
  \bibinfo{person}{Yiming Yang}, \bibinfo{person}{Jaime~G. Carbonell},
  \bibinfo{person}{Ruslan Salakhutdinov}, {and} \bibinfo{person}{Quoc~V. Le}.}
  \bibinfo{year}{2019}\natexlab{}.
\newblock \showarticletitle{XLNet: Generalized Autoregressive Pretraining for
  Language Understanding}. In \bibinfo{booktitle}{\emph{NeurIPS}}.
\newblock


\bibitem[\protect\citeauthoryear{You, Ying, and Leskovec}{You
  et~al\mbox{.}}{2019}]%
        {pgnn}
\bibfield{author}{\bibinfo{person}{Jiaxuan You}, \bibinfo{person}{Rex Ying},
  {and} \bibinfo{person}{Jure Leskovec}.} \bibinfo{year}{2019}\natexlab{}.
\newblock \showarticletitle{Position-aware Graph Neural Networks}. In
  \bibinfo{booktitle}{\emph{ICML}}, \bibfield{editor}{\bibinfo{person}{Kamalika
  Chaudhuri} {and} \bibinfo{person}{Ruslan Salakhutdinov}} (Eds.).
\newblock


\bibitem[\protect\citeauthoryear{Zhang, Yin, Huang, Du, Yang, and Lian}{Zhang
  et~al\mbox{.}}{2018}]%
        {ddl}
\bibfield{author}{\bibinfo{person}{Yan Zhang}, \bibinfo{person}{Hongzhi Yin},
  \bibinfo{person}{Zi Huang}, \bibinfo{person}{Xingzhong Du},
  \bibinfo{person}{Guowu Yang}, {and} \bibinfo{person}{Defu Lian}.}
  \bibinfo{year}{2018}\natexlab{}.
\newblock \showarticletitle{Discrete Deep Learning for Fast Content-Aware
  Recommendation}. In \bibinfo{booktitle}{\emph{WSDM}}.
\newblock


\bibitem[\protect\citeauthoryear{Zimdars, Chickering, and Meek}{Zimdars
  et~al\mbox{.}}{2001}]%
        {zimdars2001using}
\bibfield{author}{\bibinfo{person}{Andrew Zimdars},
  \bibinfo{person}{David~Maxwell Chickering}, {and}
  \bibinfo{person}{Christopher Meek}.} \bibinfo{year}{2001}\natexlab{}.
\newblock \showarticletitle{Using Temporal Data for Making Recommendations}. In
  \bibinfo{booktitle}{\emph{UAI}}.
\newblock


\end{thebibliography}

\appendix

\section{Proof of Property of ASPE}
\label{prf:prop-ape}
\begin{proof}
From~\cite{attention,tree}, SPE has a linear combination property as follows:
\begin{equation}
\label{eq:spe-linear}
    \begin{aligned}
    P_{x+y,2i}^l &=P_{x, 2 i}^l P_{y, 2 i+1}^l+P_{x, 2 i+1}^l P_{y, 2 i}^l \\
    P_{x+y,2i+1}^l &=P_{x, 2 i+1}^l P_{y, 2 i+1}^l-P_{x, 2 i}^l P_{y, 2 i}^l.
    \end{aligned}
\end{equation}

The modified RSPE is defined as:
\begin{equation}
\label{eq:m-rspe}
    \begin{aligned}
    P_{l-pos-1,2i}^l &=\cos(\left(l-pos-1\right)/10000^{2 i / d}) \\
    P_{l-pos-1,2i+1}^l &=\sin(\left(l-pos-1\right)/10000^{2 i / d}).
    \end{aligned}
\end{equation}

For simplicity, we redefine $POS=pos/10000^{2 i / d}$ and $L=\left(l-1\right)/10000^{2 i / d}$.

The addition of SPE and RSPE here is as follows:
\begin{equation}
\label{eq:spe+rspe}
    \begin{aligned}
    P_{pos,2i}^l &=\sin(POS) + \cos(L-POS) \\
    P_{pos,2i+1}^l &=\cos(POS) + \sin(L-POS).
    \end{aligned}
\end{equation}

For positions $x$ and $y$:
\begin{equation}
\label{eq:x}
    \begin{aligned}
    P_{x,2i}^l &=\sin(X) + \cos(L-X)\\
    &=(1+\sin(L))\sin{X}+\cos{L}\cos{X}\\
    P_{x,2i+1}^l &=\cos(X) + \sin(L-X)\\
    &=(1+\sin(L))\cos{X}+\cos{L}\sin{X}\\
    P_{y,2i}^l &=\sin(Y) + \cos(L-Y)\\
    &=(1+\sin(L))\sin{Y}+\cos{L}\cos{Y}\\
    P_{y,2i+1}^l &=\cos(Y) + \sin(L-Y)\\
    &=(1+\sin(L))\cos{Y}+\cos{L}\sin{Y}\\
    P_{x+y,2i}^l &=\sin(X+Y) + \cos(L-(X+Y))\\ &=(1+\sin(L))\sin(X+Y)+\cos(L)\cos(X+Y)\\
    P_{x+y,2i+1}^l &=\cos(X+Y) + \sin(L-(X+Y))\\
    &=(1+\sin(L))\cos(X+Y)+\cos(L)\sin(X+Y).
    \end{aligned}
\end{equation}

Similar to Eq.\ (\ref{eq:spe-linear}), we calculate the multiplication between encoding:
\begin{equation}
\label{eq:mul1}
    \begin{aligned}
    P_{x,2i}^lP_{y,2i+1}^l&=(1+ \sin(L))^{2} \sin (X) \cos (Y)+(1+\sin (L)) \cos (L) \sin (X) \sin (Y) \\
    &+(1+\sin (L)) \cos(L) \cos(X) \cos (Y)+\cos ^2(L) \cos (X) \sin (Y)
    \end{aligned}
\end{equation}

\begin{equation}
\label{eq:mul2}
    \begin{aligned}
    P_{x,2i+1}^lP_{y,2i}^l&=(1+ \sin(L))^{2} \sin (Y) \cos (X)+(1+\sin (L)) \cos (L) \cos (X) \cos (Y) \\
    &+(1+\sin (L)) \cos(L) \sin(X) \sin (Y)+\cos ^2(L) \sin (X) \sin (Y)
    \end{aligned}
\end{equation}

\begin{equation}
\label{eq:mul3}
    \begin{aligned}
    P_{x,2i+1}^lP_{y,2i+1}^l&=(1+ \sin(L))^{2} \cos (X) \cos (Y)+(1+\sin (L)) \cos (L) \cos (X) \sin (Y) \\
    &+(1+\sin (L)) \cos(L) \sin(X) \cos (Y)+\cos ^2(L) \sin (X) \sin (Y)
    \end{aligned}
\end{equation}

\begin{equation}
\label{eq:mul4}
    \begin{aligned}
    P_{x,2i}^lP_{y,2i}^l&=(1+ \sin(L))^{2} \sin (X) \sin (Y)+(1+\sin (L)) \cos (L) \sin (X) \cos (Y) \\
    &+(1+\sin (L)) \cos(L) \cos(X) \sin (Y)+\cos ^2(L) \cos (X) \cos (Y)
    \end{aligned}
\end{equation}

For simplicity, we define $A=\text{Eq.\ (\ref{eq:mul1})}+\text{Eq.\ (\ref{eq:mul2})}$, $B=\text{Eq.\ (\ref{eq:mul1})}-\text{Eq.\ (\ref{eq:mul2})}$, $C=\text{Eq.\ (\ref{eq:mul3})}+\text{Eq.\ (\ref{eq:mul4})}$ and $D=\text{Eq.\ (\ref{eq:mul3})}-\text{Eq.\ (\ref{eq:mul4})}$

Therefore, $P_{x+y,2i}^l$ and $P_{x+y,2i+1}^l$ can be computed based on $A,B,C$ and $D$:
\begin{equation}
\label{eq:add-x+y}
    \begin{aligned}
    P_{x+y,2i}^l &=\frac{A-\cos(L)C}{2(1+\cos^2(L))}+\frac{\cos(L)D}{2\sin(L)(1+\sin(L))}\\
    P_{x+y,2i+1}^l &=\frac{D}{2\sin(L)}+\frac{\cos(L)C-A}{s(1+\sin(L))(\cos(L)-1)}.
    \end{aligned}
\end{equation}
Therefore, the relationship between the PE of two positions based on the addition of SPE and RSPE cannot be represented as linear combination because $L$ is a variable among different sessions.
\end{proof}

\section{Example of 2DSPE}
\label{exp:2dspe}

\begin{equation}
    \begin{aligned}
P_{pos, 2 i}^l&=\sin \left(pos / 10000^{4 i / d}\right) \\
P_{pos, 2 i+1}^l&=\cos \left(pos / 10000^{4 i / d}\right) \\
P_{pos, 2 i+d / 2}^l&=\sin \left(l / 10000^{4 i / d}\right) \\
P_{pos, 2 i+1+d / 2}^l&=\cos \left(l / 10000^{4 i / d}\right)
    \end{aligned}
\end{equation}

\end{document}